\documentclass[
reprint,
nofootinbib,
amsmath,amssymb,
aps,
]{revtex4-1}
\pdfoutput=1
\usepackage[utf8]{inputenc}
\usepackage[english]{babel}
\usepackage[T1]{fontenc}
\usepackage{amsmath}
\usepackage{hyperref}

\usepackage[normalem]{ulem}
\useunder{\uline}{\ul}{}

\usepackage{tikz}
\usepackage{lipsum}

\usepackage{dsfont}
\usepackage[normalem]{ulem}
\usepackage{mathtools}
\usepackage[makeroom]{cancel}
\usepackage{bbm}
\usepackage{enumerate}

\usepackage{minitoc}

\usepackage{graphicx}
\usepackage{amsthm}

\usepackage{makecell}

\definecolor{ms}{rgb}{0,.4,1}
\newcommand{\ms}[1]{{\color{black}#1}}

\newcommand{\mb}[1]{\mathbb{#1}}

\newcommand{\Tr}[1]{\mathrm{Tr}\left[ #1\right]} 

\newcommand{\id}{\mathbbm{1}}

\newcommand{\RR}{\mb{R}}

\newcommand{\Norm}[1]{\left\Vert #1 \right\Vert}

\newcommand{\ket}[1]{\left.\left|{#1}\right.\right\rangle}

\newcommand{\bra}[1]{\left.\left\langle{#1}\right.\right|}

\newcommand{\ketbra}[2]{\ket{#1} \!\! \bra{#2}}

\newcommand{\sandwich}[3]
{\left\langle  #1 \right| #2 \left| #3 \right\rangle}

\newcommand{\bigo}[1]{\mathcal{O}\left (#1\right)}

\newcommand{\rom}[1]{\uppercase\expandafter{\romannumeral #1\relax}}
\newcommand{\norbra}[1]{\left( #1\right)}
\newcommand{\sqrbra}[1]{\left[ #1\right]}

\newcommand{\lind}{\mathcal{L}}
\newcommand{\J}{\mathbb{J}}

\newcommand{\timeorderedexp}{\mathcal{T}_{\leftarrow}{\text{exp}}}

\usepackage{amssymb}             
\newcommand{\dt}{{\rm d}t\,}
\newcommand{\de}{{\rm d}}

\newcommand{\TrR}[2]{\mathrm{Tr}_{#1}\left[ #2\right]}
\newcommand{\idO}{\mathds{I}}

\newtheorem{theorem}{Theorem}

\newtheorem{lemma}{Lemma}

\usepackage[final]{changes}

\usepackage{tcolorbox}
\tcbuselibrary{breakable}
\tcbuselibrary{skins}
\definecolor{my-green}{RGB}{0,144,81}
\definecolor{my-red}{RGB}{255,113,91}
\tcbset{enhanced jigsaw, colback=ms!10!white,colframe=ms!30!black,fonttitle=\bfseries}

\renewenvironment{proof}%
{\par\pushQED{\qed}\normalfont {\bfseries {\underline{Proof:}}}} 
{\popQED\endtrivlist\vspace{0.5em}}

\makeatletter 

\renewcommand\onecolumngrid{
	\do@columngrid{one}{\@ne}%
	\def\set@footnotewidth{\onecolumngrid}
	\def\footnoterule{\kern-6pt\hrule width 1.5in\kern6pt}%
}

\renewcommand\twocolumngrid{
	\def\footnoterule{
		\dimen@\skip\footins\divide\dimen@\thr@@
		\kern-\dimen@\hrule width.5in\kern\dimen@}
	\do@columngrid{mlt}{\tw@}
}%

\makeatother    

\setaddedmarkup{\textcolor{ms}{#1}}

\begin{document}
	
	\title{Thermalization in open many-body systems and KMS detailed balance}
	\date{\today}
	
	\author{\begingroup
		\hypersetup{urlcolor=navyblue}
		\href{https://orcid.org/0000-0002-6395-3971}{Matteo Scandi\endgroup}
	}
	\email{matteo.scandi@csic.es}
	\affiliation{Instituto de F\'{i}sica T\'{e}orica UAM/CSIC, C. Nicol\'{a}s Cabrera 13-15, Cantoblanco, 28049 Madrid, Spain}
	
	\author{\begingroup
		\hypersetup{urlcolor=navyblue}
		\href{https://orcid.org/0000-0002-5889-4022}{\'Alvaro M. Alhambra\endgroup}
	}
	\email{alvaro.alhambra@csic.es}
	\affiliation{Instituto de F\'{i}sica T\'{e}orica UAM/CSIC, C. Nicol\'{a}s Cabrera 13-15, Cantoblanco, 28049 Madrid, Spain}

	\begin{abstract}
		
		Starting from a microscopic description of weak system-bath interactions, we derive from first principles a quantum master equation that does not rely on the well-known rotating wave approximation. This includes generic many-body systems, with Hamiltonians with vanishingly small energy spacings that forbid that approximation. The equation satisfies a general form of detailed balance, called KMS, which ensures exact convergence to the many-body Gibbs state. Unlike the more common notion of GNS detailed balance, this notion is compatible with the absence of the rotating wave approximation. We show that the resulting Lindbladian dynamics not only reproduces the thermal equilibrium point \ms{up to a small renormalization of the system Hamiltonian}, but it also approximates the true system evolution with an error that grows at most linearly in time, giving an exponential improvement upon previous estimates. This master equation has quasi-local jump operators, can be efficiently simulated on a quantum computer, and reduces to the usual Davies dynamics in the limit \ms{of a coarse-graining time much larger than the inverse of the smallest frequency difference}. With it, we provide a rigorous model of many-body thermalization relevant to both open quantum systems and quantum algorithms.
	\end{abstract}
	
	\maketitle
	\onecolumngrid
	

	
	\begin{center}
		\normalsize{\underline{\textbf{Contents}}}
	\end{center}
	\begin{enumerate}[I.]
		\item \hyperref[sec:intro]{Introduction}
		\item \hyperref[sec:db]{Detailed balance for quantum systems}
		\item \hyperref[sec:thermalBaths]{Framework and properties of the bath}
		\item \hyperref[sec:CPderivation]{Microscopic derivation and approximate detailed balance}
		\begin{enumerate}[A.]
			\item \hyperref[sec:mainideas]{\emph{\added{Main ideas of the proof}}}
			\item \hyperref[subsec:bm]{\emph{The Born-Markov approximation}}
			\item \hyperref[subsec:timeaverage]{\emph{Time averaging of the dynamics}}
			\item \hyperref[subsec:improvederror]{\emph{Improved error scaling}}
			\item \hyperref[sec:completepositivity]{\emph{Complete positivity and approximate detailed balance of $\lind^{CG}_{{t}}$}}
		\end{enumerate}
		\item \hyperref[sec:exactDB]{Recovering Complete Positivity while preserving detailed balance}
		\item \hyperref[sec:qualg]{Quantum algorithmic simulation}
		\item \hyperref[sec:conclusions]{Conclusions}
	\end{enumerate}
	
	\section{Introduction}\label{sec:intro}
	
	A central question of statistical mechanics is how the stochastic dynamics observed at the macroscopic level emerges from reversible microscopic evolution. Boltzmann's renowned \mbox{H-theorem}~\cite{boltzmann1872weitere} showed that a coarse-graining of the dynamical degrees of freedom was sufficient to prove the approach to equilibrium, as seen by the decrease of the \mbox{H-function}. In doing so, Boltzmann also recognized that the time-reversal symmetry at the microscopic level mapped at the macroscopic scale into the principle of \emph{detailed balance}: at equilibrium, the rate of any transition is compensated by its reverse. Since its introduction, this fundamental symmetry has accompanied many breakthroughs in statistical mechanics, such as Einstein's absorption/emission theory~\cite{einstein1916quantentheorie} or Onsager's reciprocal relations~\cite{onsager1931reciprocal}. It also features prominently in Monte Carlo algorithms, such as the Metropolis–Hastings sampling~\cite{robert1999monte}, where the convergence to the stationary distribution is guaranteed by detailed balance. 
	
	A first quantum mechanical analogue of detailed balance was proposed by Agarwal and Alicki~\cite{agarwal1973open, alicki1976detailed}, under the name of GNS detailed balance (for Gelfand–Naimark–Segal). The work of Davies \cite{davies1974markovian,davies1976markovian} established its physical relevance, by showing how it emerges in the effective dynamics of small quantum systems when weakly coupled to a thermal bath. 
	However, both Davies' dynamics and GNS detailed balance  rely on the \emph{rotating wave} approximation (RWA), commonly used in quantum optics. This approximation applies as long as the system-bath interaction is much weaker than the relevant oscillation frequencies of the quantum system\footnote{This is connected to the energy-time uncertainty relations, as the duration of the dynamics (the inverse of the interaction strength) needs to become very large in order for the system to distinguish between different oscillation frequencies.}.
	
	While the RWA is natural for small systems, it becomes problematic in the many-body regime, since the spectra are exponentially dense in the system's size. In this context, the validity of the RWA demands the unrealistic assumption that the coupling is exponentially small with the size of the system. What this implies is that the most commonly used models of open system thermalization, namely, Davies' maps and their variants, are not relevant to even moderately large quantum systems, perhaps with the exception of systems with commuting interactions \cite{Kastoryano2016}. As such, there is a severe lack of models of open system thermalization in quantum many-body systems. Until recently, the study of these has been largely limited to simple dissipation models such as sums of single-site depolarizing or amplitude damping channels (see e.g. \cite{orus-review,fazio2025manybodyopenquantumsystems}), which do not have the many-body Gibbs state as a fixed point. This issue has motivated ideas in two complementary directions:
	
	\begin{itemize}
		\item The proposal of alternative, inequivalent, definitions for quantum detailed balance beyond GNS~\cite{fagnola2007generators, temme2010chi,amorim2021complete, scandi2023quantum}, which still guarantee convergence to thermal equilibrium.
		\item The derivation, from microscopic principles, of master equations beyond Davies' that model interactions with a bath while not relying on the RWA \cite{taj2008completely,mozgunov2020completely,nathan2020universal,Davidovi__2020,Trusheckin_2021,Kir_anskas_2018}.
	\end{itemize}
	It was previously not clear whether it is possible to merge these two ideas, into a model of interactions with an external bath that guaranteed convergence to the Gibbs state without relying on the RWA \cite{fazio2025manybodyopenquantumsystems}. In this work, we address this by deriving from first principles (that is, from a microscopic description of system-bath interaction) a quantum master equation under approximations that hold for many-body systems and that, at the same time, respects KMS detailed balance (for Kubo–Martin–Schwinger). This is a strictly weaker notion of detailed balance, that still implies convergence to thermal equilibrium \cite{fagnola2009two,amorim2021complete}. To summarise, we obtain:
	
	\textbf{Result 1.} (informal) \emph{\ms{There exists a Lindbladian that satisfies KMS detailed balanced with respect to the renormalised Hamiltonian $H_S^* := {H_S +\alpha^2 H_{LS}^{CG}(\alpha)/2}$,} and that approximates the reduced dynamics of a system coupled to a thermal bath in the weak coupling limit. It explicitly takes the form:
		\begin{align}
			\widehat{\lind}^{DB*} [\rho] 
			&=-\frac{i}{2}\sqrbra{H^{DB}_{LS},\rho}+\int_{-\infty}^{\infty}\de\omega^*\;\Big({\hat A}_\lambda(\omega^*)\rho {\hat A}^\dagger_\lambda(\omega^*)-\frac{1}{2}\{{\hat A}^\dagger_\lambda(\omega^*){\hat A}_\lambda(\omega^*),\rho\}\Big)\,,
		\end{align}
		where ${\hat A}_\lambda(\omega^*)$ are quasi-local jump operators, and $H^{DB}_{LS}$ is the Lamb shift Hamiltonian\footnote{\added{The Lamb shift Hamiltonian is often the name given to the energy shifts in a system that are induced by its external environment, following its discovery on the electron orbitals of the hydrogen atom \cite{LambShift}.}}.} 
	
	\ms{Here and in the following, we assume Einstein's summation over repeated indices.} We refer to Thm.~\ref{thm:dbLindlbadian} for the precise definitions of the quantities above.  \added{In particular, Eq. \eqref{eq:jumpOperatorsDB} shows the form of the jump operators ${\hat A}_\lambda(\omega^*)$, constructed through a convolution of the bath spectral function and a coarse-graining function chosen to be a Gaussian.} \ms{Moreover, it should be noticed that in agreement with previous literature, a dynamical renormalization of the system Hamiltonian naturally appears~\cite{winczewski2021renormalization}.}
	
	The derivation of this master equation is inspired by recent progress in dissipative quantum algorithms for Gibbs sampling \cite{chen2023quantumthermal,chen2023efficient,Ding_2025}, designed to prepare Gibbs states on a quantum computer. Interestingly, our equation can be simulated with the frameworks developed in \cite{chen2023quantumthermal,chen2023efficient,Ding_2025}, which efficiently implement Lindbladians with KMS (but no GNS) detailed balance. \added{While in those works KMS Lindbladians were only constructed as a starting point for their quantum algorithmic simulation, we here derive them from first principles under physically relevant assumptions of weak system-bath interactions.} Our derivation thus shows that the generators defined in those works have a clear physical justification, and that they can not only simulate the steady state of open-system thermalization, but also reproduce the entire dynamics. 
	
	\begin{table*}[]
		\begin{tabular}{|c|c|c|c|c|c|c|}
			\hline
			& {\;\ul CP} \quad & \makecell{\;{\ul Reduces to}\; \\ \;\;{\ul Davies}\;\;\vspace{0.4ex}} & \makecell{\;{\ul Detailed} \;\\ \;{\ul Balance}\;\vspace{0.4ex} } & {\ul No RWA} & {\;{\ul Quasi-local} } & {\ul Error estimates}                \\ \hline
			Redfield                                     &  $X$      & -         & $X$                    & \checkmark                & $X$                 & -                           \\ \hline
			Davies    ~\cite{davies1974markovian,davies1976markovian}          & \checkmark      & -                       & \checkmark (GNS)                    & $X$             & $X$                 & -                           \\ \hline 
			Taj-Rossi~\cite{taj2008completely}                 & \checkmark      & $X$                       & $X$                     & \checkmark                & \checkmark               & -                \\ \hline
			ULE~\cite{nathan2020universal}                & \checkmark      & $X$                       & $X$                      & \checkmark                & \checkmark               & $\mathcal{O}(e^{\alpha(\Gamma t)(\Gamma \tau)}-1)$                \\ \hline
			Local Davies/ULE~\cite{shiraishi2024}        & \checkmark      & $X$                     & $X$                      & $X$ / \checkmark            & \checkmark               & - 
			\\ \hline
			Coarse-grained~\cite{mozgunov2020completely} & \checkmark      & \checkmark                      & $X$                      & \checkmark                & \checkmark               & $\mathcal{O}(\alpha \sqrt{\Gamma \tau}\,(e^{\Gamma t}-1))$
			\\ \hline
			Thermod. consistent~\cite{soret2022thermodynamic} & \checkmark      & \checkmark                      & $X$ (no $H_{LS}^{CG}$)                     & \checkmark                & \checkmark               & -
			\\ \hline
			This work                                    & \checkmark      & \checkmark                     & \checkmark (KMS)                    & \checkmark               & \checkmark               & $\bigo{\alpha(\sqrt{\Gamma \tau}+\Gamma^2 \tau\, t)+(\alpha^2\,\Gamma^2\beta \,t)\,e^{(\alpha\,\Gamma\beta)^2}}$    \\ \hline
		\end{tabular}
		\caption{Comparison of the different master equations in the literature. In many previous works, the approximation errors at finite coupling have been either only numerically estimated, or largely uncontrolled. The time that appears in the estimates has been rescaled with a factor $\alpha^{-2}$ (see the discussion at the end of Sec.~\ref{sec:thermalBaths}). The error of this work has two notable features: on the one hand, it is non-zero even for $t=0$, which is the ingredient that allows us to obtain a linear scaling in time (compared to the exponential behavior of previous error estimates); on the other, a dependency on $\beta$ appears due to the introduction of detailed balance.}
	\label{tab:table1}
\end{table*}

From a technical point of view, our construction closely follows previous derivations \cite{taj2008completely,mozgunov2020completely,nathan2020universal} with additional non-trivial steps required to prove quantum detailed balance, aided by the characterization from \cite{fagnola2007generators,amorim2021complete}.  We begin by combining the methods in~\cite{taj2008completely} and~\cite{mozgunov2020completely} to get the generator of a completely positive semigroup which is approximately KMS detailed balanced. Then, inspired by~\cite{nathan2020universal}, we make a further approximation to regain exact detailed balance. Many steps are a combination of those in~\cite{taj2008completely,mozgunov2020completely} and some others in~\cite{nathan2020universal}, as it becomes apparent by comparing the jump operators in the two cases. The jump operators that generate the evolution are quasi-local for many-body systems, and the resulting master equation correctly reduces to the Davies dynamics in the limit in which the coupling is much smaller than the smallest frequency difference of the system.

We also go beyond previous works by improving on the error estimates between this effective master equation and the real evolution. In particular, we prove that the distance between them grows at most linearly with the evolution time, whereas the previous best estimates incurred in an exponentially large error \cite{mozgunov2020completely}. This gives the second main result of the paper: 

\textbf{Result 2.} (informal) \emph{The exact system dynamics $\rho_S(t)$ 
	is approximated by $e^{t\widehat{\lind}^{DB*} }(\rho_S(0))$ with an error that scales linearly in time as:
	\begin{align}
		\|    \rho_S(t) - e^{t\widehat{\lind}^{DB*} }(\rho_S(0))\|_1 \leq \bigo{\alpha\,(\Gamma\, t)\norbra{\Gamma\tau+(\alpha\,\Gamma\beta )\,e^{(\alpha\,\Gamma\beta)^2}}}\,,
	\end{align} 
	where $\alpha$ is the coupling constant, $\Gamma^{-1}$ and $\tau$ are proper timescales of the bath (defined in Eq.~\eqref{eq:GammaDef} and Eq.~\eqref{eq:TauDef}), and $\beta$ is the inverse temperature of the bath. \ms{Moreover, the fixed point of $\widehat{\lind}^{DB*}$ is close to the Gibbs state of the system Hamiltonian, up to corrections of order $\bigo{\alpha^2 (\beta \Gamma)}$.}
}

The precise statements are contained in Thm.~\ref{thm:smoothing}, Thm.~\ref{thm:coarseGrain},  Thm.~\ref{thm:DBerrorBoundIntegrated2} and Thm.~\ref{thm:interactionInsensitive}.
This result is possible thanks to the introduction of a family of states that interpolates between the exact reduced dynamics and the Lindbladian evolution we obtain from first principles (see Eq.~\eqref{eq:smoothedEv}). These states are defined by averaging the exact dynamics over a period of approximately  $T(\alpha)$, called the observation time, so that the rapidly oscillating terms in the reduced dynamics are smoothed out. The price we pay for this improvement is an additional small, time-independent error. We refer to Table~\ref{tab:table1} for a comparison of the various master equations in the literature. \ms{Finally, it should also be noticed that even if the steady state of the reduced evolution is not the exact thermal state of the system Hamiltonian, the corrections are one order larger in $\alpha$ than the one arising from the approximations on the dynamics.}



The paper is structured as follows: after introducing two main notions of detailed balance in Sec.~\ref{sec:db}, we discuss the main properties of thermal baths (and in particular the KMS condition) in Sec.~\ref{sec:thermalBaths}. Then, in Sec.~\ref{sec:CPderivation} we proceed to give a first derivation of a Lindbladian equation generating a completely positive evolution whose distance from the true state increases only linearly with time. In the same section, we also prove that, without further assumptions, this equation is already approximately KMS detailed balance. Still, in order to satisfy exact detailed balance, one needs to perform a further approximation, which is carried out in Sec.~\ref{sec:exactDB}. In the appendices, we discuss some of the more technical steps in the proofs, and other auxiliary results. 


\section{Detailed balance for quantum systems}\label{sec:db}

It is first useful to explain detailed balance in the classical regime.	
Let us consider a dissipative evolution specified by the probability vector $\vec{v}(t)$ describing the state of the system  at time $t$, and a matrix $\Phi$ inducing the stochastic dynamics:
\begin{align}
	v_i(t+1) := \sum_{j}\; \Phi_{ij}\,v_j(t)\,.\label{eq:stochasticEv}
\end{align}
The entries $\Phi_{ij}$ can be interpreted as the probability of transitioning into the {$i$-th} state when the $j$-th state is the initial condition. For this reason, we use the notation $\Phi_{ij} = P(i| j)$. Moreover, we denote the fixed point of the evolution as $ {(\gamma_S)}_i := P(i)$. Detailed balance encodes the requirement that, at equilibrium, the probability of any transition is completely compensated by its reverse. This can be written as:
\begin{align}
	P(i| j)P(j) = P(j| i)P(i)\,.\label{eq:dbClassical}
\end{align}
A simple rearranging gives the usual relation between the entries of $\Phi$:
\begin{align}
	\frac{\Phi_{ij}}{\Phi_{ji}} = \frac{ (\gamma_S)_i}{ (\gamma_S)_j}\,.\label{eq:dbClassicalRates}
\end{align}
Detailed balance is also directly related to Bayes' theorem. Dividing  both sides of the Eq.~\eqref{eq:dbClassical} by $P(i)$ gives:
\begin{align}
	P(i| j) = \frac{P(j| i)P(i)}{P(j)}\,,\label{eq:Bayes}
\end{align}
that is, the reverse transition is exactly described by Bayes' rule. The underlying  reason for this connection is simple: if we define the joint distribution $P(i,j)$ (i.e., the probability of starting in $j$ and then transitioning to $i$), time reversal symmetry corresponds to the requirement that $P(i,j) = P(j,i)$ (another rewriting of Eq.~\eqref{eq:dbClassical}). This is exactly the starting point to prove Bayes' theorem.

In the quantum case, the situation becomes more involved. The transition matrix $\Phi$ becomes a superoperator, so it is not immediately clear how to generalize Eq.~\eqref{eq:dbClassicalRates}. Even the generalizations of conditional probability and Bayes' theorem are still subject of dispute~\cite{Leifer_2013,parzygnat2023time}. Nonetheless, one can formally write a quantum version of Eq.~\eqref{eq:dbClassicalRates} as:
\begin{align}
	\Phi = \mathbb{M}_ {\gamma_S}\circ \Phi^\dagger\circ \mathbb{M}_ {\gamma_S}^{-1}\,,\label{eq:dbGeneralization}
\end{align}
where $\mathbb{M}_ {\gamma_S}$ is a superoperator that reduces to the multiplication by $ {\gamma_S}$ in the commuting case. 

The simplest choice  of $\mathbb{M}_ {\gamma_S}$ in Eq.~\eqref{eq:dbGeneralization} is given by the right multiplication operator $\RR_ {\gamma_S}[X] := X {\gamma_S} $~\cite{agarwal1973open,alicki1976detailed}. In this case the condition in  Eq.~\eqref{eq:dbGeneralization} is the aforementioned GNS detailed balance\footnote{The name is connected to the famous GNS construction from operator algebra~\cite{bratteli1997kms}.}. This is the most common definition in the literature~\cite{breuer2002theory}, and for good reasons: it is a very natural generalization of Eq.~\eqref{eq:dbClassicalRates}, it can be interpreted in terms of the scalar product induced by the usual covariance~\cite{alicki1976detailed}, and it naturally emerges in the usual microscopic derivation of the reduced Lindblad dynamics~\cite{davies1974markovian,davies1976markovian}. Still, it has one major drawback: it already implies that one has taken the rotating wave approximation.

\begin{figure}
	\centering
	\includegraphics[width=1.\linewidth]{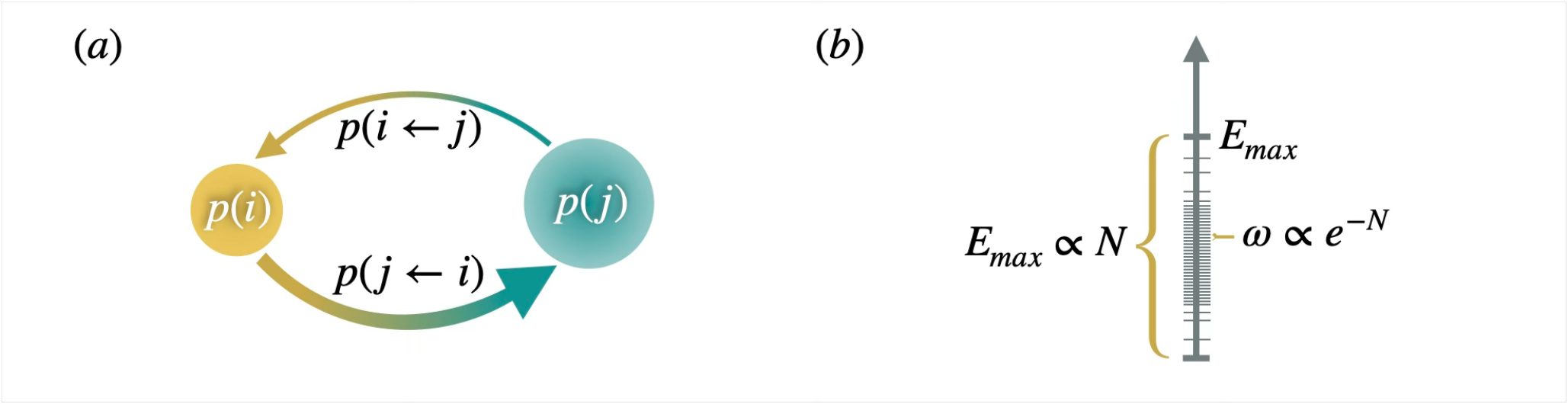}
	\caption{(a) The principle of detailed balance says that at equilibrium, if the $j$-th state is more probable than the $i$-th one, then the probability of the transition $p(j\leftarrow i)$ is equally more frequent than $p(i\leftarrow j)$; (b) Simplified depiction of the spectrum of a generic many-body Hamiltonian. Since the maximum energy is expected to scale with the size of the system (say $N$), in order to accommodate an exponential amount of different energy levels (as the dimension of the Hilbert space scales as $d^N$, where $d$ is the local dimension), the minimum spacing must become exponentially small.}
	\label{fig:dbrepresentation}
\end{figure}

To explain what this means, we first need to introduce some notation. Let us denote by $H_S$ the Hamiltonian associated to the Gibbs state $ {\gamma_S} = \frac{e^{-\beta H_S}}{\mathcal{Z}}$, and let
\begin{align}
	H_S := \sum_i E_i \ketbra{E_i}{E_i}\,,
\end{align}
be its eigendecomposition. Moreover, let $\Delta_ {\gamma_S}$ be the modular operator, defined as $\Delta_ {\gamma_S}[X] =  {\gamma_S}\,X\, {\gamma_S}^{-1}$. Since $\Delta_ {\gamma_S}$ is self-adjoint with respect to the Hilbert-Schmidt scalar product, any matrix $X$ can be decomposed as:
\begin{align}
	X = \sum_\omega\;X(\omega)\,,
\end{align}
where each $X(\omega)$ satisfies $\Delta_ {\gamma_S}[X(\omega)]=e^{-\beta\omega}X(\omega)$, i.e., it is an eigenoperator of $\Delta_ {\gamma_S}$. In coordinates, this decomposition takes the form:
\begin{align}
	X(\omega) :=\sum_{E_i,E_j}\; \delta((E_i-E_j)-\omega) \ketbra{E_i}{E_i}X\ketbra{E_j}{E_j}\,.\label{eq:freqDecomp}
\end{align}
Then, let $\{A_k\}$ be any orthonormal basis with respect to the Hilbert-Schmidt product. Any super-operator can be decomposed as:
\begin{align}
	\Phi[X]&=\sum_{k,l}\;c_{kl} \,A_l XA_k^\dagger = \sum_{k,l,\omega,\tilde{\omega}}\; c^{\omega,\tilde{\omega}}_{kl}\,A_l(\omega)  XA_k^\dagger(\tilde{\omega}) \,.\label{eq:superOperatorDec}
\end{align}
The RWA corresponds to dropping from Eq.~\eqref{eq:superOperatorDec} all the terms for which $\omega\neq\tilde{\omega}$. This is typically justified due to the fast rotation of the $\omega\neq\tilde{\omega}$ terms as compared to the overall speed of the evolution induced by $\Phi$, so that these terms quickly average out\footnote{Mathematically, this is justified with the application of the Riemann-Lebesgue lemma \cite{rivas2012open}.}. We denote the resulting evolution by $\Phi^{RW}$.

In small systems, the frequencies are of a fixed size, so the approximation is typically an excellent one in the weak coupling limit. However, as the system size increases, the spectrum of the system's Hamiltonian becomes exponentially dense with size, and so we require the unrealistic assumption that $\Phi$ is an exponentially slow process. Another way to see this is that the resulting jump operators require distinguishing different energy gaps exactly, which should involve an exponential amount of resources.
In fact, in a generic local Hamiltonian, the energy eigenvectors $\ket{E_i}$ are extremely non-local, so that the rotating wave approximated channel $\Phi^{RW}$ will also be non-local and potentially very complex, irrespective of the initial structure of $\Phi$. This seems unlikely to be nature's way of thermalizing.

What stems from this is that the RWA appears to be unphysical for large systems. As the next result shows, the same also applies to GNS detailed balance.
\begin{theorem}[GNS $\iff$ RW]\label{thm:GNSeqRW}
	Let $\Phi$ be a Hermitian preserving superoperator satisfying GNS detailed balance. Then, the rotating wave approximation is exact, that is $\Phi = \Phi^{RW}$.
\end{theorem}
\begin{proof} An operator is Hermitian preserving if $\Phi[X^\dagger] =(\Phi[X])^\dagger$. Then, by applying Eq.~\eqref{eq:dbGeneralization} repeatedly with $\mathbb{M}_ {\gamma_S} := \RR_ {\gamma_S}$, we have:
	\begin{align}
		\Phi\Delta^{-1}_ {\gamma_S}&[X] = \Phi[ {\gamma_S}^{-1}\,X\, {\gamma_S}] = \Phi^\dagger[(X^\dagger {\gamma_S}^{-1})^\dagger]\, {\gamma_S}= (\Phi^\dagger[X^\dagger {\gamma_S}^{-1}])^\dagger\, {\gamma_S}= (\Phi[X^\dagger] {\gamma_S}^{-1})^\dagger\, {\gamma_S}=\Delta^{-1}_ {\gamma_S}\Phi[X]\,,
	\end{align}
	or, alternatively, $\Phi = \Delta_ {\gamma_S}\Phi\Delta^{-1}_ {\gamma_S}$. Plugging this equality back in Eq.~\eqref{eq:superOperatorDec} we obtain:
	\begin{align}
		\Phi[X]&=\sum_{k,l,\omega,\tilde{\omega}}\; c^{\omega,\tilde{\omega}}_{kl}\,( {\gamma_S} A_l(\omega) {\gamma_S}^{-1})  X( {\gamma_S}^{-1} A_k(\tilde{\omega}) {\gamma_S})^\dagger =\sum_{k,l,\omega,\tilde{\omega}}\; c^{\omega,\tilde{\omega}}_{kl}e^{-\beta(\omega-\tilde{\omega})}\,A_l(\omega)  XA_k^\dagger(\tilde{\omega}) \,.\label{eq:proofGNS}
	\end{align}
	Since Eq.~\eqref{eq:proofGNS} is an eigendecomposition into frequencies, it has to hold that $c^{\omega,\tilde{\omega}}_{kl} = c^{\omega,\tilde{\omega}}_{kl}e^{-\beta(\omega-\tilde{\omega})}$, which is satisfied only if $\omega = \tilde{\omega}$. This proves the claim.
\end{proof}
This situation suggests going back to Eq.~\eqref{eq:dbGeneralization} and changing the definition of $\mathbb{M}_ {\gamma_S}$. A natural candidate is the symmetrized  version of the right multiplication operator $\mathbb{M}_ {\gamma_S} = \J_ {\gamma_S}$, where $\J_ {\gamma_S}[X] = \sqrt{ {\gamma_S}}\,X\sqrt{ {\gamma_S}}$. In this case Eq.~\eqref{eq:dbGeneralization} corresponds to KMS detailed balance, \ms
{sometimes also written in the equivalent form $\Phi( \sqrt{ {\gamma_S}}\,X\sqrt{ {\gamma_S}}) =  \sqrt{ {\gamma_S}}\,\Phi^\dagger(X)\sqrt{ {\gamma_S}}$}. This is the simplest generalization of $\mathbb{M}_ {\gamma_S}$ that does not reduce to the GNS detailed balance. Indeed, it has been shown that for $\mathbb{M}_ {\gamma_S}[X] =  {\gamma_S}^{\alpha} \, X \,  {\gamma_S}^{1-\alpha}$ for any $\alpha\in[0,1]$, the corresponding definition of detailed balance reduces to GNS, unless ${\alpha=1/2}$~\cite{fagnola2009two}. In this particular case the RWA is not needed, as can be seen from the following characterization:
\begin{theorem}[Structural characterization of KMS channels]\label{thm:KMSstructuralcharacterization}
	Let $\Phi$ be a Hermitian preserving superoperator satisfying KMS detailed balance. Then, its decomposition into frequencies satisfies:
	\begin{align}
		c^{\omega,\tilde{\omega}}_{kl} =\overline{ c^{-\omega,-\tilde{\omega}}_{kl}}e^{-\beta\frac{(\omega+\tilde{\omega})}{2}}\,.\label{eq:generalKMSDB}
	\end{align}
	Moreover, if $\Phi$ is GNS detailed balanced, it automatically is KMS detailed balanced as well.
\end{theorem}
\begin{proof}
	Eq.~\eqref{eq:generalKMSDB} can be verified by direct calculation. By expressing the right-hand side of Eq.~\eqref{eq:dbGeneralization} in the frequency representation we obtain:
	\begin{align}
		(\J_ {\gamma_S}\circ &\Phi^\dagger\circ \J_ {\gamma_S}^{-1})[X]=\sum_{k,l,\omega,\tilde{\omega}}\; \overline{c^{\omega,\tilde{\omega}}_{kl}}\,(\Delta_ {\gamma_S}^{\frac{1}{2}}[A^\dagger_l(\omega)])  X(\Delta_ {\gamma_S}^{\frac{1}{2}}[A^\dagger_k(\tilde\omega)])^\dagger=\sum_{k,l,\omega,\tilde{\omega}}\; \overline{c^{-\omega,-\tilde{\omega}}_{kl}} e^{-\beta\frac{(\omega+\tilde{\omega})}{2}}A_l(\omega)  XA_k^\dagger(\tilde{\omega}) \,,\label{eq:thm2eq19}
	\end{align}
	where in the first equality we used the fact that the dual of $\Phi$ is $\Phi^\dagger:=\sum_{k,l,\omega,\tilde{\omega}}\; \overline{c^{\omega,\tilde{\omega}}_{kl}}\,A^\dagger_l(\omega) XA_k(\tilde\omega)$, and we then applied the identity $X(\omega)^\dagger = X(-\omega)$, which can be verified from Eq.~\eqref{eq:freqDecomp}. Since the operators $\{A_k\}$ are a basis, we can equate the coefficients in Eq.~\eqref{eq:thm2eq19} with the ones of $\Phi$, which yields Eq.~\eqref{eq:generalKMSDB}. Finally, it should be noticed that, if the GNS condition holds, then $c^{\omega,\tilde{\omega}}_{kl} = 0$ unless, $\omega = \tilde\omega$. Then, $c^{\omega,{\omega}}_{kl} =\overline{ c^{-\omega,-{\omega}}_{kl}}e^{\beta\omega}$, which coincides with the characterization of GNS detailed balanced maps (see App.~\ref{app:channels}).
\end{proof}

With this choice of detailed balance, Eq.~\eqref{eq:dbGeneralization} becomes directly related to the Petz recovery map~\cite{petz1986sufficient}, a central concept in quantum information~\cite{junge2018universal,Sutter_2018}, which has been suggested as a possible generalization of Bayes' theorem to the quantum regime~\cite{aw2021fluctuation}. For a channel $\Phi$ with fixed point $\gamma_S$, its Petz recovery map is defined as:
\begin{align}
	\tilde{\Phi}_P :=  \J_ {\gamma_S}\circ \Phi^\dagger\circ \J_ {\gamma_S}^{-1}\,.
\end{align}
Thus the KMS detailed balance condition in Eq.~\eqref{eq:dbGeneralization} is the requirement that $\tilde{\Phi}_P = \Phi$, that the ``reverse channel" induced by the Petz recovery map is identical to the forward evolution channel. This offers a straightforward interpretation to Eq.~\eqref{eq:dbGeneralization} that is lacking for the GNS condition, since $\RR_ {\gamma_S}\circ \Phi^\dagger\circ\RR_ {\gamma_S}^{-1}$ is not a channel. 

The KMS is part of a larger family of Fisher detailed balance conditions further generalizing Eq.~\eqref{eq:dbGeneralization}, in which $\mathbb{M}_ {\gamma_S}$ is chosen to be a Fisher operator~\cite{temme2010chi, scandi2023quantum}. Within this family, the Petz recovery map is the only one that satisfies at the same time complete positivity, and the property $\widetilde{(\Phi\circ\Psi)}_P = \widetilde{\Psi}_P \circ\widetilde{\Phi}_P $, which is fundamental when treating Markovian semigroups (see Box~8 and 9 in~\cite{scandi2023quantum})\footnote{GNS detailed balance also implies this property.}.


So far we have discussed the detailed balance condition for channels. The same considerations apply to the generators $\lind$ of the Markovian dynamics defined by $\Phi_t = e^{t\lind}$, for which Eq.~\eqref{eq:dbGeneralization} implies that:
\begin{align}
	\lind &= \lim_{t\rightarrow0} \;\frac{\Phi_t-\idO}{t} = \lim_{t\rightarrow0}\; \J_ {\gamma_S}\circ \norbra{\frac{\Phi_t-\idO}{t}}^\dagger\circ \J_ {\gamma_S}^{-1}=\J_ {\gamma_S}\circ \lind^\dagger\circ \J_ {\gamma_S}^{-1}\,.\label{eq:lindKMSdef}
\end{align}
That is, KMS detailed balance holds also at the level of the generator. Thanks to the special behavior of the Petz recovery map under compositions, we can also prove the opposite direction. Notice that, for any operator $\Psi$, if $\Psi = \tilde\Psi_P$, then $\Psi^n = \tilde\Psi_P^n$. Then, suppose that Eq.~\eqref{eq:lindKMSdef} holds. Since we can rewrite $\Phi_t$ as:
\begin{align}
	\Phi_t = \lim_{n\rightarrow \infty}\;\norbra{\idO + \frac{t}{n}\lind}^n\,,
\end{align}
this shows that $\Phi_t$ is also KMS detailed balanced, so for Markovian dynamics with a time-independent generator, 
KMS detailed balance can be equivalently defined at the level of the Lindbladian, for which we have the analogous structural characterization:
\begin{theorem}[Structural characterization of KMS Lindbladians]\label{thm:KMSLindstructuralcharacterization}
	Let $\lind$ be the generator of a trace preserving evolution of the form $\Phi_t = e^{t\lind}$. Assume that $\lind$ satisfies KMS detailed balance. Then, it can be decomposed as:
	\begin{align}
		&\lind[\rho] = -i[H_{LS}^0 + H_{LS}, \rho] + \sum_{k,l,\omega,\tilde{\omega}}\; \gamma^{\omega,\tilde{\omega}}_{kl}\,\norbra{A_l(\omega)  \rho A_k^\dagger(\tilde{\omega}) -\frac{1}{2}\{A_k^\dagger(\tilde{\omega})A_l(\omega),\rho\}}\,,
	\end{align}
	where $H_{LS}^0$ and $H_{LS}$ are two Hermitian operators such that $[H_{LS}^0,\gamma_S]=0$, and $H_{LS} =\frac{1}{2}\sum_{\omega\neq\tilde{\omega}} S^{\omega,\tilde{\omega}}_{kl} A_k^\dagger(\tilde{\omega})A_l(\omega)$, where the coefficients $\gamma^{\omega,\tilde{\omega}}_{kl}$ and $S^{\omega,\tilde{\omega}}_{kl}$ satisfy:
	\begin{align}\label{eq:DBdefinition}
		\begin{cases}
			\gamma_{kl}^{\omega,\tilde{\omega}} = \overline{\gamma_{kl}^{-\omega,-\tilde{\omega}}}\,e^{-\beta(\frac{\omega+\tilde{\omega}}{2})} \\
			S_{kl}^{\omega,\tilde{\omega}} = i\tanh\norbra{\beta\norbra{\frac{\omega-\tilde{\omega}}{4}}}\gamma_{kl}^{\omega,\tilde{\omega}}	
		\end{cases}\;\;\;.
	\end{align}
\end{theorem}
We refer to App.~\ref{app:channels} for the proof, similar to that of Thm.~\ref{thm:KMSstructuralcharacterization}. \ms{Since KMS Lindbladians do not require the RWA, which is a restrictive assumption for many-body system and bath interactions, they are a natural candidate for describing the effective dynamics of open system thermalization}. We show this is the case in Sec.~\ref{sec:CPderivation} and Sec.~\ref{sec:exactDB}. 



\section{Framework and properties of the bath}\label{sec:thermalBaths}

In order to model the system-bath interaction we consider a Hamiltonian of the form $H := H_S+\alpha \,V+ H_B$, where $V$ is an interaction term which can be generically decomposed as:
\begin{align}
	V := \sum_k \; A_k\otimes B_k\,,\label{eq:interactionDec}
\end{align}
where $A_k$ are operators acting on the system only, and $B_k$ are (possibly unbounded) operators acting on the bath. We normalize the jump operators $\{A_k\}$ so that $\|A_k\|_\infty = 1$ for all $k$. 

In the interaction picture we substitute any operator $X$ with its time evolved $X(t) := e^{i( H_S+ H_B)t}X\,e^{-i( H_S+ H_B)t}$, so the dynamics is:
\begin{align}
	\partial_t\rho_{SB}(t) = -i\alpha[V(t),\rho_{SB}(t)]\,.\label{eq:diffEqOfMotion}
\end{align}
We denote the formal solution of the equation above by $\mathcal{V}(t_1,t_0) :=  \timeorderedexp\norbra{-i\alpha \int_{t_0}^{t_1}\de s \; V(s)}$, where $\timeorderedexp$ denotes the time ordered exponential constructed from right to left. Our aim is to characterize the evolution of the reduced state, that is:
\begin{align}
	\partial_t\rho_{S}(t) = -i\alpha\,\TrR{B}{[V(t),\rho_{SB}(t)]}\,.\label{eq:diffEqOfMotionRed}
\end{align}
To proceed forward, we require the following three assumptions:
\begin{enumerate}
	\item \textbf{\underline{Product initial conditions:}} The initial state of the evolution is $\rho_{SB}(0) := \rho_{S}(0)\otimes\gamma_B$ (where $\gamma_B$ is the Gibbs state of the bath). This is equivalent to setting the origin of the time to the moment in which the system starts interacting with the thermal bath.\label{it:productInCond}
	\item  \textbf{\underline{No energy shift in the bath on average:}} Each of the interaction term satisfies $\TrR{B}{B_k\gamma_B} = 0$. If this is not the case, we can shift the system Hamiltonian and the interaction as:
	\begin{align}
		\begin{cases}
			\tilde{B_k} = B_k - \TrR{B}{B_k\gamma_B} \idO_B\\
			\tilde{H}_S = H_S + \alpha \norbra{A_k\otimes\TrR{B}{B_k\gamma_B}\idO_B}
		\end{cases}\,,
	\end{align}
	so that this condition holds.\label{it:noEnergyShift}
	\item \textbf{\underline{Gaussianity of the thermal bath:}}		\label{it:gaussianBath} The thermal bath is Gaussian, that is, for any set of operators $\{B_k\}$, the expectation value of their product satisfies Wick's theorem:
	\begin{align}
		&\Tr{B_1B_2\dots B_n \gamma_B} =\sum_{j=2}^n \; \Tr{B_1B_j\gamma_B}\Tr{B_2\dots B_{j-1} B_{j+1}\dots B_n \gamma_B}\,.\label{eq:WicksTheorem}
	\end{align}
\end{enumerate}
While the requirements~\ref{it:productInCond} and~\ref{it:noEnergyShift} are generic, condition~\ref{it:gaussianBath} does restrict the scope of the derivation. Nonetheless, it is typically well justified for large, non-critical baths and for this reason it has become customary in the literature~\cite{mozgunov2020completely,nathan2020universal}.

Under these three assumptions, it is possible to rewrite Eq.~\eqref{eq:diffEqOfMotionRed} in the form (see App.~\ref{app:exactDynamics} for the derivation):
\begin{align}
	\partial_t\rho_{S}(t) = \alpha^2\int_{0}^{t}\de s\;\big (C_{kl}(t-s)&\,\TrR{B}{\sqrbra{\mathcal{V}(t,s) A_l(s) \rho_{SB}(s)\mathcal{V}^\dagger(t,s), A_k(t)}}+\nonumber
	\\
	&+C_{kl}(s-t)\,\TrR{B}{\sqrbra{ A_l(t),\mathcal{V}(t,s) \rho_{SB}(s) A_k(s)\mathcal{V}^\dagger(t,s)}}\big)\,,\label{eq:exactReducedDynamics}
\end{align}
where, here and in the following, we implicitly sum over repeated indexes $k,l$. Here we introduced the bath correlation function:
\begin{align}
	C_{kl} (t):= \TrR{B}{B_k(t) B_l\gamma_B}\,.
\end{align}
Since this quantity is central in describing the dynamics of the reduced state, we list some of its important properties. First, it behaves under complex conjugation as:
\begin{align}
	\overline{C_{kl}(t)} =  \TrR{B}{\gamma_BB_lB_k(t)} = \TrR{B}{B_l(-t)B_k\gamma_B} = C_{lk}(-t)\,,\label{eq:symmetryBathcorrelation}
\end{align}
that is, in matrix form, it holds that $\textbf{C}(t) = \textbf{C}(-t)^\dagger$. When considering its Fourier transform (which takes the name of bath spectral density):
\ms
{\begin{align}
		\widehat{C}_{kl}(\omega) := \int_{-\infty}^{\infty}\de t\; e^{it\omega} C_{kl}(t) = \frac{1}{2\pi}\TrR{B}{B_k(-\omega)B_l\gamma_B}\,,\label{eq:spectralDensity}
\end{align}}
the property in Eq.~\eqref{eq:symmetryBathcorrelation} implies that:
\begin{align}
	\widehat{C}_{kl}(\omega) =  \int_{-\infty}^{\infty}\de t\; e^{it\omega} \overline{C_{lk}(-t)} = \overline{\int_{-\infty}^{\infty}\de t\; e^{it\omega} C_{lk}(t)} = \overline{\widehat{C}_{lk}(\omega)}\,,\label{eq:correlationFunctionHermitian}
\end{align}
that is, $\widehat{\textbf{C}}(\omega)^\dagger=\widehat{\textbf{C}}(\omega)$, so the bath spectral density is a Hermitian matrix. Leveraging this fact, it can even be shown that $\widehat{\textbf{C}}(\omega)\geq 0$ (see App.~\ref{app:bath}). 

To prove Eq.~\eqref{eq:symmetryBathcorrelation}, it is sufficient to assume that $[H_B, \gamma_B]=0$. If in addition $\gamma_B$ is thermal, it also holds that:
\begin{align}
	C_{kl}(t) &= \TrR{B}{B_kB_l(-t)\gamma_B} = \frac{1}{\mathcal{Z}_B}\TrR{B}{(e^{\beta H_B}B_l(-t)e^{-\beta H_B})B_ke^{-\beta H_B}} =\\
	&= \TrR{B}{B_l(-t-i\beta)B_k\gamma_B} = C_{lk}(-t-i\beta)\,.
\end{align}
This property is often referred to as the \emph{KMS condition} (not to be confused with KMS detailed balance), and is equivalent to the fact that $\gamma_B$ is a Gibbs state~\cite{bratteli1997kms}. In the frequency domain, it takes a particularly simple form:
\ms
{\begin{align}
		\widehat{C}_{kl}(\omega) &= \frac{1}{2\pi}\TrR{B}{B_k(-\omega)B_l\gamma_B} =  \frac{1}{2\pi}\TrR{B}{B_k\gamma_B(\gamma_B^{-1}B_l(\omega)\gamma_B)} = \frac{e^{\beta\omega}}{2\pi}\TrR{B}{B_l(\omega)B_k\gamma_B} =e^{\beta\omega}\widehat{C}_{lk}(-\omega) \,,\label{eq:corrKMS}
	\end{align}
	where we used the definition in Eq.~\eqref{eq:spectralDensity}, and the relation $\Delta_ {\gamma_B^{-1}}[X(\omega)]=e^{\beta\omega}X(\omega)$.}

Since $\widehat{\textbf{C}}(\omega)$ is a positive matrix, it is also useful to introduce its  square root $\widehat{\textbf{g}}(\omega) := \sqrt{\widehat{\textbf{C}}(\omega)}$. This matrix can always be chosen to be positive. Moreover, the KMS condition implies that:
\begin{align}
	\widehat{g}_{k\lambda}(\omega)  = e^{\frac{\beta\omega}{2}}\widehat{g}_{\lambda k}(-\omega) \,.\label{eq:gspectKMSCond}
\end{align}
It can be easily connected to the bath correlation function by the convolution theorem:
\begin{align}
	C_{kl}(t) = \int_{-\infty}^{\infty}\de s\; g_{k\lambda}(s)g_{\lambda l}(t-s) = (g_{k\lambda}\ast g_{\lambda l})(t)\,,\label{eq:convTheorem}
\end{align}
where again the sum over $\lambda$ is implicit, and $\ast$ denotes the convolution.

The bath correlation function also determines the characteristic timescales of the dynamics. To this end, let us introduce the two sets of constants:
\ms
{
	\begin{align}
		\Gamma_0:&= \sum_{kl}\int_{-\infty}^{\infty}\dt\; |C_{kl}(t)|\,, \qquad \qquad\qquad\qquad \Gamma:= \sum_{kl}\int_{-\infty}^{\infty}\dt\; |g_{k\lambda}(t)|\int_{-\infty}^{\infty}\dt\; |g_{\lambda l}(t)|\,, \label{eq:GammaDef}\\
		\Gamma_0\tau_0 :&=  \sum_{kl}\int_{-\infty}^{\infty}\dt\;|t| |C_{kl}(t)|\,,\label{eq:TauDef} \qquad \;\qquad\qquad\;\Gamma\tau:=\sum_{kl}\, 2\int_{-\infty}^{\infty}\dt\; |t||g_{k\lambda}(t)|\int_{-\infty}^{\infty}\dt\; |g_{\lambda l}(t)|\,,\\
		K_0 &:=  \sum_{kl}\int_{-\infty}^{\infty}\dt\;|t|^2 |C_{kl}(t)|\,; \;\qquad \qquad\qquad\;K:=\sum_{kl} \,\int_{-\infty}^{\infty}\dt\int_{-\infty}^{\infty}\de s\; (|t|+|s|)^2|g_{k\lambda}(t)| |g_{\lambda l}(s)|\,,\label{eq:KDef}
	\end{align}
	where it can be proven that $|\widehat{C}_{kl}(\omega)|\leq\Gamma_0\leq \Gamma$, that $|\widehat{C}'_{kl}(\omega)|\leq\Gamma_0\tau_0\leq \Gamma\tau$ and that  $|\widehat{C}''_{kl}(\omega)|\leq K_0\leq K$ (see App.~D). It should be noticed that requiring these constants to be finite restrict the applicability of our results: for example, $\Gamma_0$ does not diverge only for infinite baths, because only in this case $C_{kl}(t)$ can decay sufficiently fast for the integral to converge. Moreover, requiring that $\tau$ is also finite, further restricts to the case in which the decay at infinity is faster than $1/t^2$ (and similarly, it should decay faster than $1/t^3$ for $K<\infty$).
	Interestingly, the constant $K$ only plays a role in the bounds of Sec.~\ref{sec:exactDB}, whereas the derivation in Sec.~\ref{sec:CPderivation} only rely on the finiteness of $\Gamma_0$ and $\tau_0$.} The latter can be interpreted as  a measure the correlation time of the bath (the timescale over which perturbations  decay\footnote{This intuitive interpretation becomes exact if $C_{kl}(t)\propto e^{-|t|/\tau^*}$ or $C_{kl}(t)\propto e^{-t^2/(\tau^*)^2}$, since in this case one can explicitly carry out the calculation and obtain $\tau_0 \propto \tau^*$}). Also, $\Gamma_0$ and $\Gamma$ can be related to the evolution rate of the state in a very precise manner~\cite{nathan2020universal}:
\begin{theorem}[Fastest rate in the system]\label{thm:fastestRate}
	For a Gaussian bath, and an evolution satisfying condition~\ref{it:productInCond} and~\ref{it:noEnergyShift}, the rate of change of the reduced system dynamics can be upper bounded as:
	\begin{align}
		\|\partial_t\rho_{S}(t)\|_{1}\leq 2\alpha^2 \, \Gamma_0\,,\label{eq:fastestRate}
	\end{align}
	where $\|\cdot\|_1$ is the trace norm $\|X\|_1:= \Tr{|X|}$.
\end{theorem}
We refer to App.~\ref{app:fastestRate} for a proof of this statement. There are two main messages in Eq.~\eqref{eq:fastestRate}: first, the speed at which the system changes directly depends on the constant $\Gamma_0$ defined at the level of the bath; second, the maximum rate is of order $\alpha^2$, so the system will take longer to evolve as the coupling constant decreases. In order to make the total evolution time independent of $\alpha$, we introduce the time-rescaled state $\tilde{\rho}_{S}(t) := {\rho}_{S}(t/\alpha^2)$, for which it holds that $\|\partial_{t}\tilde{\rho}_{S}(t)\|_{1}\leq 2 \, \Gamma_0$.

The equation presented in Eq.~\eqref{eq:exactReducedDynamics} is exact, and it already highlights aspects of the reduced system dynamics: its dependence on the bath correlation function, the natural appearance of a fastest rate $\Gamma_0$ and the scaling of the evolution time with $\alpha^2$. Still, from the quantitative point of view, Eq.~\eqref{eq:exactReducedDynamics} is intractable in general: its solution depends on the exact global evolution $\mathcal{V}(t,s)$, and it is highly non-local in time, since the evolution at time $t$ depends on the full history $\{\rho_{SB}(s)\}_{s\leq t}$. Along the derivation, we introduce several approximations that ultimately make the dynamics quantitatively manageable.


\section{Microscopic derivation and approximate detailed balance}\label{sec:CPderivation}
We now prove that there exists an operator $\lind_{{t}}$ which at the same time generates a completely positive semigroup, and can be used to express Eq.~\eqref{eq:exactReducedDynamics} as:
\begin{align}
	\partial_{{t}}\tilde{\rho}_{S}({t}) = \lind_{{t}}[\tilde{\rho}_{S}({t})] + \mathcal{E}_{{t}}\,,\label{eq:exactDecomposition}
\end{align}
where $\mathcal{E}_t$ is some arbitrary correction whose trace norm we can upper bound. Such a decomposition would allow to apply the following result:
\begin{theorem}[Propagation of errors]\label{thm:DBerrorBoundIntegrated}
	Consider two family of states ${\pi}({t})$ and ${\rho}({t})$ respectively satisfying the system of equations:
	\begin{align}
		\begin{cases}
			\partial_{{t}}\,{\pi}({t})={\lind}_{{t}} [ {\pi}({t})]\\
			{\pi}(0) = {\pi}(0)
		\end{cases}\;;
		\qquad\qquad\qquad
		\begin{cases}
			\partial_{{t}}\,{\rho}({t})={\lind}_{{t}}[ {\rho}({t})]+\mathcal{E}_{{t}}\\
			{\rho}(0) =  {\pi}(0)+\delta\rho(0)\label{eq:199}
		\end{cases}\;,
	\end{align}
	where we assume that: $(i)$ the norm of the perturbation $\delta\rho(0)$ is bounded, that is $\Norm{\delta\rho(0)}_1\leq R$; that $(ii)$ there exists an absolutely integrable function $m(t)$ (with primitive $M(t)$) such that $\|\mathcal{E}_t\|_1\leq m(t)$; and finally $(iii)$ that the evolution $\Phi(t_1,t_0):= \timeorderedexp\norbra{\int_{t_0}^{t_1}\de s \; \lind_s}$ is CPTP for all times. Then, the norm distance between $ {\pi}({t})$ and $ {\rho}({t})$ scales as:
	\begin{align}
		\| {\pi}({t})- {\rho}({t})\|_1\leq R+(M(t)-M(0))\,.
	\end{align}
\end{theorem}
\begin{proof}
	The equations in Eq.~\eqref{eq:199} can be formally integrated to give:
	\begin{align}
		\begin{cases}
			\pi(t) = \Phi({t},0)[{\pi}(0)]
			\\
			\rho(t) = \Phi({t},0)[{\pi}(0)] +\Phi({t},0)[{\delta\rho}(0)]+ \int_{0}^{t}\de s\; \Phi({t},s)[\mathcal{E}_s]
		\end{cases}\,.
	\end{align}
	Then, the trace distance between the two states can be estimated as:
	\begin{align}
		\| {\pi}({t})- {\rho}({t})\|_1&\leq\Norm{\Phi({t},0)[{\delta\rho}(0)]}_1+\left\| \int_{0}^{t}\de s\; \Phi({t},s)[\mathcal{E}_s]\right\|_1\leq
		\\
		&\leq R+\int_{0}^{t}\de s\;\left\|  \Phi({t},s)[\mathcal{E}_s]\right\|_1\leq R+ \int_{0}^{t}\de s\;\left\|  \mathcal{E}_s\right\|_1\leq R+(M(t)-M(0))
	\end{align}
	where we used the contractivity of the trace norm under CP-evolution, that is the fact that for every completely positive channel $\Psi$ it holds that~\cite{rivas2014quantum}: 
	\begin{align}
		\|\Psi(X)\|_1\leq \|X\|_1\,.\label{eq:traceNormNonIncreasing}
	\end{align}
	Finally, we applied the bound $\|\mathcal{E}_s\|_1\leq m(s)$ that holds for all $s$. This proves the claim.
\end{proof}
For most of the error terms below, $m(t)$ will be a constant, so that, in those cases, Thm.~\ref{thm:DBerrorBoundIntegrated} implies an error scaling linearly with time.

Thanks to the theorem just derived, in order to justify the reduced dynamics generated by some Lindbladian $\lind_t$, we just need to prove that a decomposition of $\partial_{{t}}\tilde{\rho}_{S}({t})$ akin to Eq.~\eqref{eq:exactDecomposition} can be found. We proceed in different steps, introducing a set of new states, each corresponding to a different approximation to the original evolution. In this context, it is particularly useful to apply the change of variables $\{\sigma = \alpha^2s\,; \, x = s-u\}$ to the integrated version of Eq.~\eqref{eq:exactReducedDynamics}:
\begin{align}
	\tilde{\rho}_{S}(t) =& {\rho}_{S}(0)+\nonumber
	\\&+\int_0^{t}\de \sigma\; \int_0^{\frac{\sigma}{\alpha^2}}\de x\;\;\bigg (C_{kl}\norbra{x}\,\TrR{B}{\sqrbra{\mathcal{V}_{\sigma,x}\, A_l\norbra{\frac{\sigma}{\alpha^2}-x} \tilde{\rho}_{SB}\norbra{\sigma-\alpha^2x}\mathcal{V}_{\sigma,x}^\dagger, A_k\norbra{\frac{\sigma}{\alpha^2}}}}+\nonumber
	\\
	&\qquad\qquad\qquad\qquad\qquad\qquad\quad+C_{kl}(-x)\,\TrR{B}{\sqrbra{ A_l\norbra{\frac{\sigma}{\alpha^2}},\mathcal{V}_{\sigma,x}\, \tilde{\rho}_{SB}\norbra{\sigma-\alpha^2x} A_k\norbra{\frac{\sigma}{\alpha^2}-x}\mathcal{V}_{\sigma,x}^\dagger}}\bigg)\,,\label{eq:exactReducedRescaled}
\end{align}
where, for brevity, we introduced the notation $\mathcal{V}_{\sigma,x}:=\mathcal{V}\norbra{\frac{\sigma}{\alpha^2},\frac{\sigma}{\alpha^2}-x}$. 

\begin{table}[]
	\ms{
		\begin{tabular}{|c|c|c|c|c|}
			\hline\makecell{\;\vspace{-2ex}\\\;\;{\ul State }\; \;\vspace{0.4ex}}  
			& \makecell{\;\vspace{-2ex}\\\;\;{\ul Notation }\; \;\vspace{0.4ex}}   &\makecell{\;\vspace{-2ex}\\\;\;{\ul Physical intuition }\; \;\vspace{0.4ex}} & \makecell{\;\vspace{-2ex}\\\;\;{\ul Main feature}\; \;\vspace{0.4ex}}& \makecell{\;\vspace{-2ex}\\\;\;{\ul Error bound}\; \;\vspace{0.4ex}} \\     \hline     
			Exact &$\tilde{\rho}_S$ & \makecell{\;\vspace{-2ex}\\\;\;{Obtained from tracing out }\;\; \\ \;\;{the environment}\;\;\vspace{0.4ex}} &\;\; Intractable but exact \;\;& 0\\ \hline     
			Born & $\tilde{\rho}^B_S$ & \makecell{\;\vspace{-2ex}\\\;\;{The bath dynamics is much faster than }\;\;\\ \;\;{the one of the system$\implies$ it always}\;\;\\\;\;{looks in thermal equilibrium}\;\;\vspace{0.4ex}} & \makecell{\;\vspace{-2ex}\\\;\;{Only depends on the  }\;\;\\ \;\;{bath through $C_{kl}(t)$}\;\;\vspace{0.4ex}}& $\bigo{\alpha^2 (\Gamma_0 t)(\Gamma_0\tau) }$\\\hline     
			Markov & $\tilde{\rho}^{BM}_S$ &  \makecell{\;\vspace{-2ex}\\\;\;{The bath dynamics is much faster than}\;\;\\ \;\;{the one of the system$\implies$ the  }\;\;\\\;\;{environment is memoryless}\;\;\vspace{0.4ex}}& \makecell{\;\vspace{-2ex}\\\;\;{Is local in time (admits}\;\;\\ \;\;{a differential form)}\;\;\vspace{0.4ex}}&$\bigo{\alpha^2 (\Gamma_0 t)(\Gamma_0\tau) }$ \\\hline     
			Redfield & $\tilde{\rho}^{RE}_S$ &\makecell{\;\vspace{-2ex}\\\;\;{The information about the initial time}\;\;\\ \;\;{is quickly erased}\;\;\vspace{0.4ex}} &  \makecell{\;\vspace{-2ex}\\\;\;{Differential form doesn't}\;\;\\ \;\;{depend on time in the}\;\;\\ \;\;{ Schrödinger's picture}\;\;\vspace{0.4ex}}& $\bigo{\frac{(\Gamma_0 t)(\Gamma_0\tau)}{\Gamma_0T(\alpha)} }$\\\hline     
			Coarse grained & $\tilde{\rho}^{CG}_S$ &\makecell{\;\vspace{-2ex}\\\;\;{Rapid oscillations in the rates justify}\;\;\\ \;\;{averaging over a time period}\;\;\\ \;\;{$T(\alpha)$ (called observation time)}\;\;\vspace{0.4ex}} & \makecell{\;\vspace{-2ex}\\\;\;{Is generated by a}\;\;\\ \;\;{Lindblad equation}\;\;\vspace{0.4ex}}&$\bigo{\alpha^2(\Gamma_0 t)(\Gamma_0T(\alpha))}$ \\\Xhline{1.2pt}     
			Smoothed & $\tilde{\rho}^S_S$ & \makecell{\;\vspace{-2ex}\\\;\;{Rapid oscillations in the rates are already}\;\;\\ \;\;{present in the exact dynamics $\implies$ }\;\;\\ \;\;{averaging over a time period $T(\alpha)$}\;\;\vspace{0.4ex}} & \makecell{\;\vspace{-2ex}\\\;\;{Is close to the exact}\;\;\\ \;\;{dynamics at all times}\;\;\vspace{0.4ex}} &$\bigo{\alpha^2(\Gamma_0T(\alpha))}$ \\
			\hline
		\end{tabular}
		\caption{Here we collect the notations used throughout the derivation, together with the physical intuition behind each of the approximations. Each of the approximation is carried out sequentially on top of the state coming from the previous stage of approximation. For this reason, the error bounds accounts for the difference in one-norm with respect to the previous step (e.g., for $\tilde{\rho}_{S}^{BM}$ it corresponds to an estimate of $\|\tilde{\rho}_{S}^{B}-\tilde{\rho}_{S}^{BM}\|_1$ ). The only exception is $\tilde{\rho}^S_S$, which approximates the exact dynamics directly (and for this reason the error is computed with respect to $\tilde{\rho}_S$). Some additional notations used are: the tilde in ${\tilde \rho}_S^X(t)$ (for generic $X$) refers to the rescaled time evolution, that is ${\tilde \rho}_S^X(t)={\rho}_S^X(t/\alpha^2)$; the extra $S$ in the subscript ${\tilde \rho}_{S,S}^X(t)$ refers to the Schrödinger's picture.}
		\label{tab:2}}
\end{table}

\ms{\subsection{Main ideas of the proof}\label{sec:mainideas}
	
	In order to find a decomposition of the type in Eq.~\eqref{eq:exactDecomposition}, we introduce a number of states approximating the exact evolution (see Table~\ref{tab:2}), by exploiting the fast dynamics of the bath as compared to the system: for example, since the bath goes back to equilibrium very quickly, it will effectively appear to the system as if it was always at equilibrium. This justifies the Born approximation (see Thm.~\ref{thm:bornApproximation}). Then, for the same reason, any information about the past evolution is quickly erased, justifying the Markov approximation (in Thm.~\ref{thm:markovApproximation}). This allows us to express the evolution in differential form (see Eq.~\eqref{eq:markovDynamics}). Finally, exploiting the Markovianity of the resulting dynamics, we carry out the Redfield approximation, which erases any information about the switching on of the interaction (see Thm.~\ref{thm:redfield}). This makes the generator of the differential form of the dynamics time-independent in the Schrödinger picture (see Eq.~\eqref{eq:decomposeRE}). These steps are the content of Sec.~\ref{subsec:bm}, and are standard in the literature \cite{taj2008completely,mozgunov2020completely,nathan2020universal}. 
	
	It is important to point out that even if these approximations simplify the dependency of the resulting state on the exact dynamics $\tilde{\rho}_S(t)$, this is anyway needed for their very definition. None of the states defined in this manner are helpful in finding a treatable expression of the reduced system dynamics. Their interest relies on the fact that, at each step, we have a provable error bound for the derivative of their difference. That is, since $\partial_t\tilde{\rho}^{RE}_{S}(t) = \lind^{RE}_{{t}}[\tilde{\rho}_{S}(t)]$, we can obtain an expression similar to Eq.~\eqref{eq:exactDecomposition} as follows:
	\begin{align}
		\partial_t\tilde{\rho}_{S}(t)& = \partial_t\norbra{\tilde{\rho}_{S}(t)-\tilde{\rho}^{B}_{S}(t) } + \partial_t\norbra{\tilde{\rho}^{B}_{S}(t)-\tilde{\rho}^{BM}_{S}(t) } + \partial_t\norbra{\tilde{\rho}^{BM}_{S}(t)-\tilde{\rho}^{RE}_{S}(t) } + \partial_t\tilde{\rho}^{RE}_{S}(t) =
		\\
		&= \mathcal{E}_t^{B}+\mathcal{E}_t^{BM}+\mathcal{E}_t^{RE} + \lind^{RE}_{{t}}[\tilde{\rho}_{S}(t)]\,.\label{eq:expl1}
	\end{align}
	Bounding the error terms $\mathcal{E}_t^{B}$,$\mathcal{E}_t^{BM}$, and $\mathcal{E}_t^{RE}$ is the content of Thm.~\ref{thm:bornApproximation},~\ref{thm:markovApproximation} and~\ref{thm:redfield} respectively. At this point, one might be tempted to define a state in terms of the differential equation:
	\begin{align}
		\begin{cases}
			\partial_{{t}}\,{\pi}({t})={\lind}^{RE}_{{t}} [ {\pi}({t})]\\
			{\pi}(0) = {\tilde{\rho}}_S(0)
		\end{cases}\;,
	\end{align}
	and then apply Thm.~\ref{thm:DBerrorBoundIntegrated} to prove that $\pi(t)$ diverges from $\tilde{\rho}_{S}(t)$ in a controlled manner. There is a fundamental reason that makes this impossible: the generator $\lind^{RE}_{{t}}$ is not completely positive, so that Thm.~\ref{thm:DBerrorBoundIntegrated} cannot be applied. Not only $\pi(t)$ is not guaranteed to stay positive throughout the evolution, but one can also generically expect the norm $\|\pi(t)-\tilde{\rho}_{S}(t)\|_1$ to grow exponentially in time (see the discussion in Sec.~\ref{subsec:improvederror}).
	
	At this point, we only need to find a further approximation that makes the generator CP. This is done by exploiting the oscillating nature of the rates in the differential equation in Eq.~\eqref{eq:expl1}, as explained at the beginning of Sec.~\ref{subsec:timeaverage}. Doing this directly, though, prevents us from finding a controlled error bound in one norm, which we need for the application of  Thm.~\ref{thm:DBerrorBoundIntegrated} (see~\cite{mozgunov2020completely}). 
	
	In order to avoid this, we go back to the beginning of the derivation and slightly change the starting point, which is where our derivation deviates from previous work (and in particular from the best previous rigorous proof \cite{mozgunov2020completely}). We introduce a smoothed state ${\tilde\rho}_{S}^S(t)$, corresponding to averaging the exact dynamics over a time window of order $T(\alpha)$, called the observation time (see Eq.~\eqref{eq:smoothedEv} for the precise expression). This is the main innovation in our derivation, and it is what allows to obtain error bounds that only scale linearly with time. Indeed, applying the same approximations as in Sec.~\ref{subsec:bm} allows to retain the controlled error bounds we derived before, but with the extra advantage of giving a controlled error bound also when passing to a CP generator. Moreover, we can also show that the exact evolution ${\tilde\rho}_S(t)$ and the smoothed one ${\tilde\rho}_{S}^S(t)$ stay at a constant distance throughout the evolution. This is all that is needed for the Thm.~\ref{thm:DBerrorBoundIntegrated} to hold, and it gives the main result of this section, Thm.~\ref{thm:coarseGrain}. Finally, in Sec. \ref{sec:completepositivity}, we go beyond previous derivations by analyzing how the resulting evolution has approximate detailed balance, which we later expand on in Sec. \ref{sec:exactDB}.}

\subsection{The Born-Markov approximation}\label{subsec:bm}

Since the evolution $\mathcal{V}(s,u) :=  \timeorderedexp\norbra{-i\alpha \int_{u}^{s}\de x \; V(x)}$ depends on $\alpha$, one can expect that  $\mathcal{V}_{\sigma,x} \simeq \idO$. Then, the substitution $\mathcal{V}_{\sigma,x}\rightarrow\idO$ in Eq.~\eqref{eq:exactReducedRescaled} is what is called the Born approximation. There are two regimes in which this is well justified: firstly, whenever $\alpha$ is small enough for the unitary evolution to be close to the identity operator (the weak-coupling limit); secondly, whenever the timescale characterizing the decay  of correlations in the bath (which, as we mentioned, is dictated by $\tau_0$ in Eq.~\eqref{eq:TauDef}) is much shorter than the fastest timescale of the system $\Gamma_0$. Both these intuitions are confirmed by the following result:
\begin{theorem}[Born approximation]\label{thm:bornApproximation} 
	Let us define the evolved state $\tilde{\rho}^B_{S}(t)$ as:
	\begin{align}
		\tilde{\rho}^B_{S}(t) = {\rho}_{S}(0)+\int_0^{t}\de \sigma\; \int_0^{\frac{\sigma}{\alpha^2}}\de x\;\;\bigg (&C_{kl}\norbra{x}\,\sqrbra{\, A_l\norbra{\frac{\sigma}{\alpha^2}-x} \tilde{\rho}_{S}\norbra{\sigma-\alpha^2x}, A_k\norbra{\frac{\sigma}{\alpha^2}}}+\nonumber
		\\
		&\qquad\qquad\qquad\qquad+C_{kl}(-x)\,\sqrbra{ A_l\norbra{\frac{\sigma}{\alpha^2}}, \tilde{\rho}_{S}\norbra{\sigma-\alpha^2x} A_k\norbra{\frac{\sigma}{\alpha^2}-x}}\bigg)\,.\label{eq:bornDynamics}
	\end{align}
	Let us define the error term $\mathcal{E}^B_t= \partial_t\norbra{\tilde{\rho}_{S}(t)-\tilde{\rho}^B_{S}(t)}$. Then, we can give the following bound:
	\begin{align}
		\|\mathcal{E}^B_t\|_1\leq 4\alpha^2\, \Gamma_0\,(\Gamma_0\tau_0)\,.\label{eq:bornApproximationE}
	\end{align}
	Moreover, it directly follows that:
	\begin{align}
		\|\tilde{\rho}_{S}(t)-\tilde{\rho}^B_{S}(t)\|_1\leq\int_0^t\de s\;\|\mathcal{E}^B_s\|_1\leq 4\alpha^2\, (\Gamma_0t)\,(\Gamma_0\tau_0)\,.\label{eq:bornApproximation}
	\end{align}
\end{theorem}

We refer to App.~\ref{app:born} for the proof. Thm.~\ref{thm:bornApproximation} confirms the intuition behind the Born approximation. Indeed, from Eq.~\eqref{eq:bornApproximation} we can see that  the error is negligible whenever $\Gamma_0 t$ is small (short times), or $\Gamma_0 \tau_0$ is small (fast decay of correlations in the bath), or $\alpha^2$ is small compared to the product of these two timescales (weak-coupling). Moreover, since the jump operators $A_k(t)$ only act on the system, the Born approximation $\mathcal{V}_{\sigma,x}\rightarrow\idO$ allows to rewrite the evolution directly in terms of the reduced state $\tilde{\rho}_{S}(t)$ in Eq.~\eqref{eq:bornDynamics}, without any reference to the state of the bath, which is a first step towards reducing the complexity of representing the dynamics. 

Still, even if Eq.~\eqref{eq:bornDynamics} is defined only in terms of the properties of the system and of the bath correlation function, the resulting evolution is still non-local in time. This problem is taken care through the Markov approximation, which consists in substituting $\tilde{\rho}_{S}\norbra{\sigma-\alpha^2x}$ with $\rho(\sigma)$ in Eq.~\eqref{eq:bornDynamics}. The intuitive justification of this procedure comes from  Thm.~\ref{thm:fastestRate}, that says that the derivative of $\tilde{\rho}_{S}$ is bounded at all times, so that the corresponding evolution cannot be very non-local. This is confirmed by the following theorem:
\begin{theorem}[Markov approximation]\label{thm:markovApproximation}
	Let us define the evolved state $\tilde{\rho}^{BM}_{S}(t)$ as:
	\begin{align}
		\tilde{\rho}^{BM}_{S}(t) = \rho_{S}(0) +\int_0^{t}\de \sigma \int_0^{\frac{\sigma}{\alpha^2}}\de x\;\bigg (&C_{kl}\norbra{x}\,\sqrbra{\, A_l\norbra{\frac{\sigma}{\alpha^2}-x} \tilde{\rho}_{S}\norbra{\sigma}, A_k\norbra{\frac{\sigma}{\alpha^2}}}+\nonumber
		\\
		&\qquad\qquad\qquad\qquad\qquad+C_{kl}(-x)\,\sqrbra{ A_l\norbra{\frac{\sigma}{\alpha^2}}, \tilde{\rho}_{S}\norbra{\sigma} A_k\norbra{\frac{\sigma}{\alpha^2}-x}}\bigg)\,.\label{eq:markovDynamicsImplicit}
	\end{align}
	Moreover, we introduce the error term $\mathcal{E}^{BM}_t= \partial_t\norbra{\tilde{\rho}^{B}_{S}(t)-\tilde{\rho}^{BM}_{S}(t)}$. Then, we can give the following bound:
	\begin{align}
		\|\mathcal{E}^{BM}_t\|_1\leq 2\alpha^2\, \Gamma_0\,(\Gamma_0\tau_0)\,,\label{eq:markovApproximationE}
	\end{align}
	which directly implies that:
	\begin{align}
		\|\tilde{\rho}_{S}^B(t)-\tilde{\rho}^{BM}_{S}(t)\|_1\leq\int_0^t\de s\;\|\mathcal{E}^{BM}_s\|_1\leq 2\alpha^2\, (\Gamma_0t)\,(\Gamma_0\tau_0)\,.\label{eq:markovApproximation}
	\end{align}
\end{theorem}
Once again, we defer the proof to App.~\ref{app:markov}. While the Born approximation allows to simplify the dependency of the dynamics on the degrees of freedom of the bath, the Markov approximation uses the fact that the decoherence time of the bath is much shorter than the typical timescales of interaction in order to make the evolution memoryless. This gives a dynamics that is local in time, since by differentiating Eq.~\eqref{eq:markovDynamicsImplicit} it is possible to obtain:
\begin{align}
	\partial_t\tilde{\rho}^{BM}_{S}(t)&=\int_0^{\frac{t}{\alpha^2}}\de x\;\bigg (C_{kl}\norbra{x}\,\sqrbra{\, A_l\norbra{\frac{t}{\alpha^2}-x} \tilde{\rho}_{S}\norbra{t}, A_k\norbra{\frac{t}{\alpha^2}}}+\nonumber
	\\
	&\qquad\qquad\qquad\qquad\qquad\qquad+C_{kl}(-x)\,\sqrbra{ A_l\norbra{\frac{t}{\alpha^2}}, \tilde{\rho}_{S}\norbra{t} A_k\norbra{\frac{t}{\alpha^2}-x}}\bigg) \label{eq:markovDynamics1}
	\\
	&= \lind^{BM}_t[\tilde{\rho}_{S}\norbra{t}]\,,\label{eq:markovDynamics}
\end{align}
or, alternatively:
\begin{align}
	\partial_{{t}}\tilde{\rho}_{S}({t})  &= \partial_{{t}}\tilde{\rho}^{BM}_{S}({t}) +  \partial_{{t}}(\tilde{\rho}^{B}_{S}({t}) -\tilde{\rho}^{BM}_{S}({t}) ) + \partial_{{t}}(\tilde{\rho}_{S}({t}) -\tilde{\rho}^{B}_{S}({t}) ) = 
	\\
	&= \lind^{BM}_t[\tilde{\rho}_{S}\norbra{t}] + \mathcal{E}^{BM}_t+\mathcal{E}^{B}_t\,.\label{eq:bornmarkoverror}
\end{align}
This expression is formally very similar to Eq.~\eqref{eq:exactDecomposition}, but with a key difference: even if we are able to bound the trace norm of $\mathcal{E}^{B}_t$ and $\mathcal{E}^{BM}_t$ (uniformly in time), the evolution generated by $\lind^{BM}_t$ is not completely positive.

An intermediate step is to extend the integral in Eq.~\eqref{eq:markovDynamics1} to infinity. Once again, this is done by exploiting the rapid decay of $C_{kl}(x)$ as $|x|\rightarrow\infty$:
\begin{theorem}[Redfield approximation]\label{thm:redfield}
	Let us define the evolved state $\tilde{\rho}^{RE}_{S}(t)$ as:
	\begin{align}
		\tilde{\rho}^{RE}_{S}(t) = \rho_{S}(0) +\int_0^{t}\de \sigma \int_0^{\infty}\de x\;e^{-\frac{(x/2)^2}{T(\alpha)^2}}\bigg (&C_{kl}\norbra{x}\,\sqrbra{\, A_l\norbra{\frac{\sigma}{\alpha^2}-x} \tilde{\rho}_{S}\norbra{\sigma}, A_k\norbra{\frac{\sigma}{\alpha^2}}}+\nonumber
		\\
		&\qquad\qquad\qquad+C_{kl}(-x)\,\sqrbra{ A_l\norbra{\frac{\sigma}{\alpha^2}}, \tilde{\rho}_{S}\norbra{\sigma} A_k\norbra{\frac{\sigma}{\alpha^2}-x}}\bigg)\,.\label{eq:redfieldDynamicsImplicit}
	\end{align}
	Moreover, we introduce the error term $\mathcal{E}^{RE}_t= \partial_t\norbra{\tilde{\rho}^{BM}_{S}(t)-\tilde{\rho}^{RE}_{S}(t)}$. Then, the following holds:
	\begin{align}
		\|\mathcal{E}^{RE}_t\|_1\leq 2\,\Gamma_0(1+\Gamma_0\tau_0) \norbra{\frac{1}{{2+\Gamma_0\,T(\alpha)}}+\frac{\alpha^2}{(\alpha^2+\Gamma_0 \,t)}}
		\,.\label{eq:redfieldApproximationE}
	\end{align}
	Integrating the right-hand side of Eq.~\eqref{eq:redfieldApproximationE} we also obtain:
	\begin{align}
		\|\tilde{\rho}_{S}^{BM}(t)-\tilde{\rho}^{RE}_{S}(t)\|_1\leq\int_0^t\de s\;\|\mathcal{E}^{BM}_s\|_1\leq 2\,(1+\Gamma_0\tau_0) \norbra{\frac{\Gamma_0\,t}{{2+\Gamma_0\,T(\alpha)}}+\alpha^2\log\norbra{1+\frac{\Gamma_0 \,t}{\alpha^2}}}\,.\label{eq:redfieldApproximation}
	\end{align}
	The error goes to zero with $\alpha$ whenever $T(\alpha)\rightarrow\infty$ in the same regime.
\end{theorem}
The proof of this result is postponed to App.~\ref{app:redfield}. The approximation presented here is a smoothed version of the usual Redfield approximation, obtained by sending $T(\alpha)\rightarrow\infty$. This leaves us with one last step in order to derive a CP evolution.

\subsection{Time averaging of the dynamics}\label{subsec:timeaverage}
Differentiating Eq.~\eqref{eq:redfieldDynamicsImplicit}, we define the generator $\lind^{RE}_{{t}}$ satisfying $\partial_t\tilde{\rho}^{RE}_{S}(t)=\lind^{RE}_{{t}}[\tilde{\rho}_{S}(t)]$, that is:
\begin{align}
	\lind^{RE}_{{t}}[\tilde{\rho}_{S}(t)]&=\int_0^{\infty}\de x\;e^{-\frac{(x/2)^2}{T(\alpha)^2}}\;\mathcal{U}_{\frac{t}{\alpha^2}}^\dagger\sqrbra{\norbra{C_{kl}\norbra{x}\,\sqrbra{\, A_l\norbra{-x} \mathcal{U}_{\frac{t}{\alpha^2}}[\tilde{\rho}_{S}\norbra{t}], A_k}+C_{kl}(-x)\,\sqrbra{ A_l, \mathcal{U}_{\frac{t}{\alpha^2}}[\tilde{\rho}_{S}\norbra{t}] A_k\norbra{-x}}}}\,,\label{eq:56}
\end{align}
where we introduced the notation $\mathcal{U}_{\frac{t}{\alpha^2}}$ for the operator $\mathcal{U}_{t}[X] = e^{-iH_S t} X e^{iH_S t}$. Then, Eq.~\eqref{eq:56} shows that we can decompose the Lindbladian $\lind^{RE}_{{t}}$ as:
\begin{align}
	\lind^{RE}_{{t}}[\tilde{\rho}_{S}(t)] = \mathcal{U}_{\frac{t}{\alpha^2}}^\dagger\circ\widehat{\lind}^{RE}\circ\mathcal{U}_{\frac{t}{\alpha^2}}[\tilde{\rho}_{S}(t)]\,,\label{eq:decomposeRE}
\end{align}
where we implicitly defined $\widehat{\lind}^{RE}_t$. This rewriting is particularly useful because it isolates two dynamics happening at very different timescales: on the one hand, the dependency of the state $\tilde{\rho}_{S}(t)$  is of order $\|\partial_t\tilde{\rho}_{S}(t)\|_{1}\simeq\bigo{\alpha^{0}}$, whereas, on the other, the unitary dynamics is extremely fast, with a rate of change of order $\|\partial_t\,\mathcal{U}_{\frac{t}{\alpha^2}}\|_{1-1}\simeq\bigo{\alpha^{-2}}$. To appreciate the effect of this scale separation, it is useful to expand Eq.~\eqref{eq:redfieldDynamicsImplicit} in the frequency domain:
\begin{align}
	&\tilde\rho^{RE}_{S}(t) =\rho_{S}(0)+\nonumber
	\\
	&+\int_0^{t}\de \sigma\; \int_0^{\infty}\de x\;e^{-\frac{(x/2)^2}{T(\alpha)^2}}\;e^{i\sigma \frac{\omega-\tilde{\omega}}{\alpha^2}}\bigg (e^{-ix\omega}C_{kl}\norbra{x}\,\sqrbra{ A_l(\omega) \tilde{\rho}_{S}\norbra{\sigma}, A^\dagger_k(\tilde{\omega})}+e^{ix\tilde{\omega}}C_{kl}(-x)\,\sqrbra{ A_l(\omega), \tilde{\rho}_{S}\norbra{\sigma} A_k^\dagger(\tilde{\omega})}\bigg)\,.\label{eq:bmFrequency}
\end{align}
Here and in the following, we implicitly integrate over all frequencies $\omega$ and $\tilde{\omega}$. In the limit $\alpha\rightarrow0$ the exponent of $e^{i\sigma \frac{\omega-\tilde{\omega}}{\alpha^2}}$ diverges whenever the frequencies are different. This allows for the application of the Riemann-Lebesgue lemma, which shows that the only terms surviving in Eq.~\eqref{eq:bmFrequency} are exactly the ones for which $\omega =\tilde{\omega}$~\cite{rivas2012open, breuer2002theory}. This amounts to performing the RWA, which straightforwardly yields a completely positive Lindbladian that is GNS detailed balanced (thanks to Thm.~\ref{thm:GNSeqRW}). Nonetheless, this limit is not always justified, and for this reason we need a more careful analysis.

As a matter of fact, even without the RWA, the difference in scaling of $\tilde{\rho}_{S}(t)$ as compared to $\|\partial_t\,\mathcal{U}_{\frac{t}{\alpha^2}}\|_{1-1}$ suggests that over a small period of time, the state appears to be frozen to the unitary evolution. This justifies the introduction of the smoothed state\footnote{By convention, whenever the argument of $\tilde{\rho}_S(t)$ is negative, we set $\tilde{\rho}_S(-|t|)=\tilde{\rho}_S(0)$. The same goes for all the other approximate states, that is $\tilde{\rho}^{B}_S(-|t|)=\tilde{\rho}^{BM}_S(-|t|)=\tilde{\rho}^{RE}_S(-|t|)=\tilde{\rho}_S(0)$.}:
\begin{align}
	\tilde{\rho}^{S}_S(t) := \int_{-\infty}^{\infty}\de q\; \frac{e^{-\frac{q^2}{T(\alpha)^2}}}{\sqrt{\pi} \,T(\alpha)}\;\tilde{\rho}_S(t+\alpha^2q)\,,\label{eq:smoothedEv}
\end{align}
where $T(\alpha)$ is a free parameter called the observation time. This procedure corresponds to averaging around the state at time $t$, over a window of width $T(\alpha)$ and a variation governed by $\alpha^2$. In real time, without rescaling, this gives a variation independent of $\alpha$, since $\tilde{\rho}_S(t+\alpha^2q) = {\rho}_S\norbra{\frac{t}{\alpha^2}+q}$. The new result shows that this smoothing is well justified:
\begin{theorem}[Smoothing of evolution]\label{thm:smoothing}
	Let $\tilde{\rho}^{S}_S(t) $ be the averaged state defined in Eq.~\eqref{eq:smoothedEv}. Then, the trace distance from the exact evolution can be bounded uniformly in $t$ as:
	\begin{align}
		\Norm{\tilde{\rho}_S(t)-\tilde{\rho}^{S}_S(t)}_1 \leq \frac{2\Gamma_0\alpha^2 T(\alpha)}{\sqrt{\pi}}\,.\label{eq:differenceSmoothing}
	\end{align}
	Thus, the error goes to zero when $\alpha\rightarrow0$ if $T(\alpha)< \alpha^{-2}$.
\end{theorem}
We defer the proof to App.~\ref{app:cg}. The bound imposed on $T(\alpha)$ can be intuitively understood as follows: on the one hand, one wants $T(\alpha)$  to be sufficiently large compared with the fast unitary dynamics (which happens at a timescale of order $\bigo{\alpha^2}$) so to smooth its effect; on the other, it cannot diverge arbitrarily fast, as the total time of evolution in real time will be of order $\bigo{\alpha^{-2}}$ (since $\tilde{\rho}_S(t) = {\rho}_S\norbra{\frac{t}{\alpha^2}}$). This justifies the name of \emph{observation time} for $T(\alpha)$.

Importantly, on top of always being close to the real evolution, the smoothed state  in Eq.~\eqref{eq:smoothedEv} also satisfies the differential equation (see App.~\ref{app:cg}):
\begin{align}
	\partial_t\tilde{\rho}^{S}_S(t) :&= \lind_t^{CG}[\tilde{\rho}^{S}_S(t)] + \mathcal{E}^{TOT}_t\,,
	\label{eq:diffS}
\end{align}
where $\mathcal{E}^{TOT}_t$ is an error term whose trace norm can be controlled, and where the coarse-grained Lindbladian $\lind^{CG}_{{t}}$ is defined as:
\begin{align}
	\lind^{CG}_{{t}}[\tilde{\rho}^S_{S}\norbra{t}] := \int_{-\infty}^{\infty}\de q\int_0^{\infty}\de x\; \frac{e^{-\frac{(q^2+(x/2)^2)}{T(\alpha)^2}}}{\sqrt{\pi} \,T(\alpha)}\bigg (&C_{kl}\norbra{x}\,\sqrbra{\, A_l\norbra{\frac{t}{\alpha^2}+q-\frac{x}{2}} \tilde{\rho}^S_{S}\norbra{t}, A_k\norbra{\frac{t}{\alpha^2}+q+\frac{x}{2}}}+\nonumber
	\\
	&\qquad+C_{kl}(-x)\,\sqrbra{ A_l\norbra{\frac{t}{\alpha^2}+q+\frac{x}{2}}, \tilde{\rho}^S_{S}\norbra{t} A_k\norbra{\frac{t}{\alpha^2}+q-\frac{x}{2}}}\bigg)\,.\label{eq:averagedLindblad}
\end{align}
\begin{figure}
	\centering
	\includegraphics[width=0.5\linewidth]{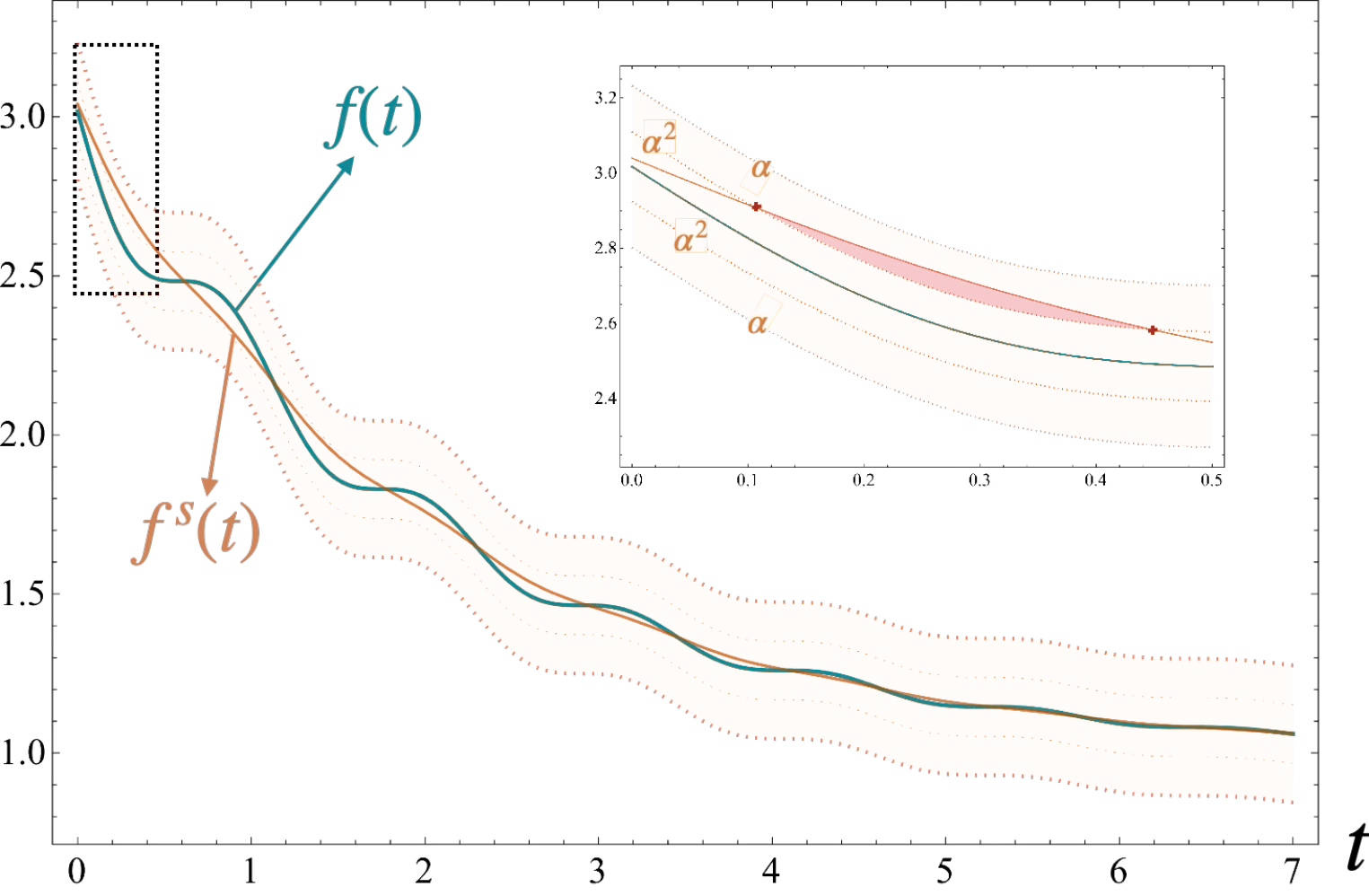}
	\caption{ Intuition behind Theorem \ref{thm:coarseGrain}. In order to illustrate the effect of the smoothing, consider the differential equation $f'(t) =F(t,\alpha):=- e^{-t/2} \left(1+\cos \left(\frac{t}{\alpha ^2}\right)\right)$. We define $f^{s}(t)$ to be the solution to the differential equation obtained by taking the Gaussian average of $F(t+\alpha^2q,\alpha)$. In analogy with the discussion in Thm.~\ref{thm:coarseGrain}, we choose $T(\alpha) = \alpha^{-1}$. The shaded region corresponds to an error of order $\bigo{\alpha}$ around the exact solution $f(t)$. As it can be seen, the averaged solution always falls within this region. The inset corresponds to zooming the beginning of the evolution (the region contained in the dotted box): this shows that even if $f^{s}(t)$ always stays within $\bigo{\alpha}$ of  $f(t)$, this distance can be larger than $\bigo{\alpha^2}$.} 
	\label{fig:averaging}
\end{figure}
In Eq.~\eqref{eq:averagedLindblad} there are two approximations happening at the same time. First, motivated by the discussion above, we average over different values of $\lind_t^{RE}$, obtaining:
\begin{align}
	\lind^{CG0}_{{t}}:=\int_{-\infty}^{\infty}\de q\; \frac{e^{-\frac{q^2}{T(\alpha)^2}}}{\sqrt{\pi} \,T(\alpha)}\; {\lind}^{RE}_{t+\alpha^2q}=\int_{-\infty}^{\infty}\de q\; \frac{e^{-\frac{q^2}{T(\alpha)^2}}}{\sqrt{\pi} \,T(\alpha)}\; \mathcal{U}_{\frac{t}{\alpha^2}+q}^\dagger\circ\widehat{\lind}^{RE}\circ\mathcal{U}_{\frac{t}{\alpha^2}+q} \,,\label{eq:lindcg0}
\end{align}
Then, exploiting the rapid decay of $C_{kl}(x)$ as $|x|\rightarrow\infty$, we  shift $q\rightarrow q+\frac{x}{2}$ inside of the two commutators, but not in the Gaussian, obtaining Eq.~\eqref{eq:averagedLindblad}. Importantly this happens at the same order of approximation. This procedure allows us to derive the key result:
\begin{theorem}[Complete positivity]\label{thm:CP}
	The Lindbladian $\lind^{CG}_{{t}}$ induces a completely positive, trace preserving evolution.
\end{theorem}
The proof is deferred to Sec.~\ref{sec:completepositivity}, where we provide an explicit structural characterization of $\lind^{CG}_{{t}}$ that ensures complete positivity. It can also be shown that $\lind^{CG}_{{t}}$ reduces to its RWA version in the limit of $T(\alpha)\rightarrow\infty$ whenever the spectrum of $H_S$ has a finite gap between any two energy differences (see App.~\ref{app:davies}).

Theorem~\ref{thm:CP} is the ingredient needed to obtain the final error estimates by applying Thm.~\ref{thm:DBerrorBoundIntegrated}:
\begin{theorem}[Time averaging]\label{thm:coarseGrain}
	Let us define the solution of the differential equation $\tilde{\rho}^{CG}_{S}(t)$ as:
	\begin{align}
		\begin{cases}
			\partial_t\tilde{\rho}^{CG}_S(t) := \lind_t^{CG}[\tilde{\rho}^{CG}_S(t)]\\
			\tilde{\rho}^{CG}_S(0) = \tilde{\rho}(0)
		\end{cases}\;;\label{eq:cgDynamics}
	\end{align}
	Then, it holds that the norm distance between $\tilde{\rho}^{CG}_{S}(t)$ and the true evolution has an optimal scaling with time given by: 
	\begin{align}
		\Norm{	\tilde{\rho}_S(t)-	\tilde{\rho}^{CG}_S(t)}_1\leq \frac{2\alpha \sqrt{2+3\,\Gamma_0\tau_0}}{\sqrt{\pi}} + \norbra{2\alpha\Gamma_0t}\norbra{2\sqrt{2+3\,\Gamma_0\tau_0}+3\alpha (\Gamma_0\tau_0)}+2\alpha^2\,\Gamma_0(1+\Gamma_0\tau_0)\log\norbra{1+\frac{\Gamma_0t}{\alpha^2}}\,.\label{eq:cgSmoothedE}
	\end{align}
	which corresponds to the choice of an optimal $T(\alpha)$ given by:
	\begin{align}
		T^*_{opt}(\alpha) := \frac{1}{2\alpha \Gamma_0}\sqrt{2+3\,\Gamma_0\tau_0}\,.\label{eq:optimalT}
	\end{align}
\end{theorem}
We provide the proof of this result in App.~\ref{app:cg}. 


\subsection{Improved error scaling} \label{subsec:improvederror}

We now highlight how the error estimate just presented improves upon previous ones in the literature (e.g., in~\cite{mozgunov2020completely}). A key step in most of the approaches is to bound the difference between the true evolved state and the reduced dynamics using a triangle inequality of the form:
\begin{align}
	\Norm{	\tilde{\rho}_S(t)-	\tilde{\rho}^{CG}_S(t)}_1 \leq \Norm{	\tilde{\rho}_S(t)-	\tilde{\pi}(t)}_1 + \Norm{	\tilde{\pi}(t)-	\tilde{\rho}^{CG}_S(t)}_1\,,\label{eq:triangle}
\end{align}
where $\tilde{\pi}(t)$ is some arbitrary chosen matrix. The most common choice is to set $\tilde{\pi}(t)$ as the solution to the differential equation induced by $\lind_t^{RE}$. 
This choice is far from optimal: since $\lind_t^{RE}$ does not induce a CP evolution, the trace norm of $\tilde{\pi}(t)$ is not bounded in this case, and one is only able to derive an exponential scaling (in time) for both terms on the right hand side of Eq.~\eqref{eq:triangle}. The reason behind this is fundamental: it can be proven that not only the trace norm contracts under CP-evolutions, but that also the opposite holds, that is it expands under non CP-evolutions when considering the doubled space~\cite{rivas2014quantum}. That is, if $\Psi$ is non-CP, there must exist an operator $X$ on $\mathcal{B}(\mathcal{H}\otimes\mathcal{H})$ such that:
\begin{align}
	\|\idO\otimes\Psi(X)\|_1\geq \|X\|_1\,.\label{eq:traceNormCPEquivalence}
\end{align}
Since one cannot rule out this possibility, the bounds on the trace norm of the error terms need to exponentially increase in time when acted upon by $\timeorderedexp\norbra{\int_{t_0}^{t_1}\de s \; \lind^{RE}_s}$. 

This behavior should be contrasted with the linear increase of Thm.~\ref{thm:DBerrorBoundIntegrated}. In order to prove this scaling, we choose $\tilde{\pi}(t)=\tilde{\rho}^{S}_S(t)$: since $\tilde{\rho}^{S}_S(t)$ is a convex combination of states, it is itself a state (and therefore $\Norm{\tilde{\rho}^{S}_S(t)}_1=1$); moreover, thanks to Thm.~\ref{thm:smoothing} we also know that it is close to the true evolved state,  even at large $t$, allowing to bound the first norm on the right hand side of Eq.~\eqref{eq:triangle}. This comes at a small cost: in contrast to previous derivations, in Thm.~\ref{thm:smoothing} we prove that the error in this case does not go to zero with time, but rather is constant and of order $\bigo{\alpha \sqrt{\Gamma_0\tau_0}}$. On the other hand, thanks to the fact that $\lind^{CG}_t$ generates a CP-semigroup (Thm.~\ref{thm:CP}), the evolution $\timeorderedexp\norbra{\int_{t_0}^{t_1}\de s \; \lind^{CG}_s}$ contracts the trace distance. Then, since $\mathcal{E}^{TOT}_t$ in Eq.~\eqref{eq:diffS} can be bounded in trace norm, this error can be propagated according to Thm.~\ref{thm:DBerrorBoundIntegrated} to also estimate the second norm. This shows that it is crucial to delay the application of Thm.~\ref{thm:DBerrorBoundIntegrated} to a point in which the reduced dynamics generates a CP-semigroup. 


\subsection{Complete positivity and approximate detailed balance of $\lind^{CG}_{{t}}$}
\label{sec:completepositivity}

Similarly with what happened with $\lind^{RE}_{{t}}$ (see Eq.~\eqref{eq:decomposeRE}), $\lind^{CG}_{{t}}$ can also be rewritten as $\lind^{CG}_{{t}}=\mathcal{U}_{\frac{t}{\alpha^2}}^\dagger\circ\widehat{\lind}^{CG}\circ\mathcal{U}_{\frac{t}{\alpha^2}}$, where  $\widehat{\lind}^{CG}$ does not depend on time. The additional unitary rotations are a consequence of being in the interaction picture. Rotating back into the Schrödinger picture (which we denote by $\tilde{\rho}^{CG}_{S,S}(t)$), we obtain\footnote{The dependence in $\alpha^{-2}$ of the unitary term is just a by-product of being in the rescaled time. Going back to real time corresponds to multiplying the right hand side of Eq.~\eqref{eq:cgSchroedinger} by $\alpha^2$, so that there is no divergence in the limit $\alpha\rightarrow0$.}:
\begin{align}
	\partial_t\tilde{\rho}^{CG}_{S,S}(t)  = -\frac{i}{\alpha^2}[H_S, \tilde{\rho}^{CG}_{S,S}(t)] + \widehat{\lind}^{CG}[\tilde{\rho}^{CG}_{S,S}(t)]\,.\label{eq:cgSchroedinger}
\end{align}
For this reason, it will be sufficient to characterize $\widehat{\lind}^{CG}$. Before doing so, it is useful to introduce some notation. First let us define $\omega_+ := \frac{\omega+\tilde{\omega}}{2}$ and $\omega_- := {({\omega}-\tilde\omega)}$, and the filter function $\mathcal{G}(q,x) :=(e^{-\frac{q^2+(x/2)^2}{T(\alpha)^2}})/({\sqrt{\pi} \,T(\alpha)})$. Then, $\widehat{\lind}^{CG}$ has the following frequency representation:
\begin{align}
	\widehat{\lind}^{CG}[\rho] = \int_{-\infty}^{\infty}\de q\int_0^{\infty}\de x\;\mathcal{G}(q,x)\;e^{iq\omega_-}\bigg (e^{-ix\omega_+}C_{kl}\norbra{x}\,\sqrbra{ A_l(\omega) \rho, A^\dagger_k(\tilde{\omega})}+e^{ix\tilde{\omega}_+}C_{kl}(-x)\,\sqrbra{ A_l(\omega), \rho A_k^\dagger(\tilde{\omega})}\bigg)\,.
\end{align}	
Again, unless otherwise specified, we imply an integration over all frequencies $\omega$ and $\tilde{\omega}$. Let us also introduce the two smoothed half-Fourier transform:
\begin{align}
	\widehat{C}^+_{kl}(\omega_+,\omega_-) :&= \int_{-\infty}^{\infty}\de q\; e^{iq\omega_-}\int_{0}^{\infty}\de x\, e^{-ix\omega_+}\mathcal{G}(q,x) C_{kl}(x)\,;
	\\
	\widehat{C}^-_{kl}(\omega_+,\omega_-)  :&=  \int_{-\infty}^{\infty}\de q\; e^{iq\omega_-}\int_{-\infty}^{0}\de x\, e^{-ix\omega_+}\mathcal{G}(q,x) C_{kl}(x)\,.
\end{align}
With this notation, $\widehat{\lind}^{CG}$ can be rewritten as:
\begin{align}
	\widehat{\lind}^{CG}[{\rho}] &= \int_{-\infty}^{\infty} \de\omega \int_{-\infty}^{\infty} \de\tilde\omega\;\bigg  (\widehat{C}^+_{kl}(\omega_+,\omega_-) \norbra{A_l(\omega)\rho A^\dagger_k(\tilde\omega)-A^\dagger_k(\tilde\omega)A_l(\tilde\omega)\rho}+\nonumber\\
	&\qquad\qquad\qquad\qquad\qquad\qquad\qquad\qquad\qquad\qquad+\widehat{C}^-_{kl}(\omega_+,\omega_-) \norbra{A_l(\omega)\rho A^\dagger_k(\tilde\omega)-\rho A^\dagger_k(\tilde\omega)A_l(\omega)}\bigg ) =
	\\
	&= \int_{-\infty}^{\infty} \de\omega \int_{-\infty}^{\infty} \de\tilde\omega\; \bigg(\gamma_{kl}^{\omega,\tilde{\omega}}\Big(A_l(\omega)\rho A^\dagger_k(\tilde\omega)-\frac{1}{2}\{A^\dagger_k(\tilde\omega)A_l(\omega),\rho\}\Big)-i\,S_{kl}^{\omega,\tilde{\omega}}\frac{1}{2}[A^\dagger_k(\tilde\omega)A_l(\omega),\rho]\bigg)	\,,\label{eq:averagedEvo}
\end{align}
where we have rewritten the integration over the frequencies for clarity, and we introduced the coefficients:
\begin{align}
	\gamma_{kl}^{\omega,\tilde{\omega}}&= 	\widehat{C}^+_{kl}(\omega_+,\omega_-)+	\widehat{C}^-_{kl}(\omega_+,\omega_-)=\int_{-\infty}^{\infty}\de q\; e^{iq\omega_-}\int_{-\infty}^{\infty}\de x\, e^{-ix\omega_+}\mathcal{G}(q,x)  C_{kl}(x)\,,
	\\
	S_{kl}^{\omega,\tilde{\omega}} &=i\, (	\widehat{C}^+_{kl}(\omega_+,\omega_-)-	\widehat{C}^-_{kl}(\omega_+,\omega_-)) = i\int_{-\infty}^{\infty}\de q\; e^{iq\omega_-}\int_{-\infty}^{\infty}\de x\, e^{-ix\omega_+}\mathcal{G}(q,x)  C_{kl}(x){\rm sign}(x)\,.\label{eq:defGammaS}
\end{align}
We are now ready to prove that Eq.~\eqref{eq:averagedEvo} generates a CP semigroup. Let us first analyze the commutator term. A direct calculation shows that the Lamb-shift Hamiltonian defined as:
\begin{align}
	H_{LS}^{CG} := \int_{-\infty}^{\infty} \de\omega \int_{-\infty}^{\infty} \de\tilde\omega\;S_{kl}^{\omega,\tilde{\omega}}A^\dagger_k(\tilde\omega)A_l(\omega)\,,\label{eq:LSHam}
\end{align} 
is indeed Hermitian (see App.~\ref{app:hLShermitian}). This implies that it generates a unitary rotation, which is CP. 

It remains to examine the dissipative part of the Lindbladian, which we dub $\widehat{\mathcal{D}}^{CG}$. In order to show that it generates a CP semi-group, we perform the change of variables $\{t_1 = q-x/2\,; \, t_2= q+x/2\}$. This has two effects. First, it maps the product of Gaussians to another product of Gaussians: 
\begin{align}
	\mathcal{G}(q,x)=  \frac{e^{-\frac{q^2+(x/2)^2}{T(\alpha)^2}}}{\sqrt{\pi}\, T(\alpha)} = \norbra{\frac{e^{-\frac{t_1^2}{2T(\alpha)^2}}}{(\sqrt{\pi}\, T(\alpha))^{1/2}}} \,\norbra{\frac{e^{-\frac{t_2^2}{2T(\alpha)^2}}}{(\sqrt{\pi}\, T(\alpha))^{1/2}}}\,,\label{eq:substitutionQXT}
\end{align}
Let us denote the resulting (unnormalized) Gaussians by ${f}^\alpha(t_1)$ and ${f}^\alpha(t_2)$. The second and most important effect of the change of variables is that it decouples the two frequencies $\omega$ and $\tilde{\omega}$. Indeed, an explicit calculation shows that:
\begin{align}
	\gamma_{kl}^{\omega,\tilde{\omega}}&=\int_{-\infty}^{\infty}\de q\int_{-\infty}^{\infty}\de x\;e^{i\norbra{q-\frac{x}{2}}\omega} e^{-i\norbra{q+\frac{x}{2}}{\tilde{\omega}}}\mathcal{G}(q,x)C_{kl}(x) =
	\\
	&=  \int_{-\infty}^{\infty}\de t_1\int_{-\infty}^{\infty}\de t_2\;{f}^\alpha(t_1){f}^\alpha(t_2)e^{it_1\omega} e^{-it_2{\tilde{\omega}}}C_{kl}(t_2-t_1) = 
	\\
	&=\int_{-\infty}^{\infty}\frac{\de\omega^*}{2\pi}\;{\widehat{C}}_{kl}(\omega^*)\norbra{\int_{-\infty}^{\infty}\de t_1\;{f}^\alpha(t_1)e^{it_1(\omega+\omega^*)}}\norbra{\int_{-\infty}^{\infty}\de t_2\;{f}^\alpha(t_2)e^{-it_2(\tilde{\omega}+\omega^*)}}=
	\\
	&=\int_{-\infty}^{\infty}\frac{\de\omega^*}{2\pi}\;{\widehat{C}}_{kl}(\omega^*){\widehat{f}}^\alpha(\omega+\omega^*)\overline{{\widehat{f}}^\alpha({\tilde\omega}+\omega^*)}\,,\label{eq:ratesFreqDecouples}
\end{align}
where we introduced the bath spectral density defined in Eq.~\eqref{eq:spectralDensity}. Then, using this decomposition, we obtain:
\begin{align}
	\widehat{\mathcal{D}}^{CG}[\rho] &=  \int_{-\infty}^{\infty} \de\omega \int_{-\infty}^{\infty} \de\tilde\omega\; \gamma_{kl}^{\omega,\tilde{\omega}}\Big(A_l(\omega)\rho A^\dagger_k(\tilde\omega)-\frac{1}{2}\{A^\dagger_k(\tilde\omega)A_l(\omega),\rho\}\Big) =\label{eq:86}
	\\
	&=\int_{-\infty}^{\infty}\frac{\de\omega^*}{2\pi}\;{\widehat{C}}_{kl}(\omega^*)\Big({\tilde A}_l(\omega^*)\rho {\tilde A}^\dagger_k(\omega^*)-\frac{1}{2}\{{\tilde A}^\dagger_k(\omega^*){\tilde A}_l(\omega^*),\rho\}\Big)\,,\label{eq:53}
\end{align}
where we introduced the smoothed version of the jump operators:
\begin{align}
	{\tilde A}_l(\omega^*) :& = \int_{-\infty}^{\infty}\de{\omega} \;{\widehat{f}}^\alpha({\omega}+\omega^*) A_l({\omega})=\int_{-\infty}^{\infty}\de t_1\;e^{it_1\omega^*}{f}^\alpha(t_1)A_l(t_1)\,,\label{eq:jumpOperatorsCG}
\end{align}
Thus, since ${\widehat{C}}_{kl}(\omega^*)$ is a positive matrix (see Sec.~\ref{app:bath}), we can conclude that the semigroup generated by:
\begin{align}
	\widehat{\lind}^{CG} [\rho] 
	&=-\frac{i}{2}\sqrbra{H^{CG}_{LS},\rho}+\int_{-\infty}^{\infty}\frac{\de\omega^*}{2\pi}\;{\widehat{C}}_{kl}(\omega^*)\Big({\tilde A}_l(\omega^*)\rho {\tilde A}^\dagger_k(\omega^*)-\frac{1}{2}\{{\tilde A}^\dagger_k(\omega^*){\tilde A}_l(\omega^*),\rho\}\Big)\label{eq:lindEq}
\end{align}
is completely positive. This proves Thm.~\ref{thm:CP}. Additionally, the expression of the  jump operators in Eq.~\eqref{eq:jumpOperatorsCG} allows to prove their quasi-locality (see App.~\ref{app:quasiLocalJump}), and that in the limit in which the observation time $T(\alpha)$ goes to infinity, Eq.~\eqref{eq:lindEq} correctly reduces to the usual Davies generator (see App.~\ref{app:davies}).

There are a number of properties that can be deduced from the expression of the rates in Eq.~\eqref{eq:ratesFreqDecouples} (and a similar transformation for $S_{kl}^{\omega,\tilde{\omega}}$). First, it can be shown that they can be rewritten as (see App.~\ref{app:compRates}):
\begin{align}
	&\gamma_{kl}^{\omega,\tilde{\omega}}=\frac{e^{-(T(\alpha)\,\omega_-)^2/4}}{\sqrt{\pi}}\int_{-\infty}^{\infty}\de\Omega\;{\widehat{C}}_{kl}\norbra{\frac{\Omega}{T(\alpha)}-\omega_+}\,e^{-\Omega^2}\label{eq:gammarateCG}\;;
	\\
	&S_{kl}^{\omega,\tilde{\omega}}=\frac{e^{-(T(\alpha)\,\omega_-)^2/4}}{\sqrt{\pi}}\int_{-\infty}^{\infty}\frac{\de\Omega}{\pi}\int_{-\infty}^{\infty}\de\omega_2\;\frac{\widehat{C}_{kl}(\omega_2)}{(\omega_2+\omega_+)-\frac{\Omega}{T(\alpha)}}\;e^{-\Omega^2} \,,\label{eq:SrateCG}
\end{align} 
These expressions are particularly useful when studying the reduction of $\widehat{\lind}^{CG}$ to the Davies generator: thanks to the exponential suppression in $(T(\alpha)\omega_-)$ present both in $\gamma_{kl}^{\omega,\tilde{\omega}}$ and $S_{kl}^{\omega,\tilde{\omega}}$, if the minimum gap between two frequencies $\omega^{min}_-$ is finite, the rotating wave limit is approached exponentially fast for $T(\alpha) \gg (\omega^{min}_-)^{-1}$.  This also shows what fails in the many-body setting: since $\omega^{min}_-$  is exponentially small (in the system size), the suppression justifying the RWA happens only for observation time that are exponentially large in the system size \ms{(precise bounds are provided in App.~\ref{app:davies})}. 

Another consequence of the rewriting in Eq.~\eqref{eq:gammarateCG}-\eqref{eq:SrateCG} is the following (see App.~\ref{app:approxDB}):
\begin{theorem}[Approximate detailed balance]
	\label{thm:approxDB}
	The coefficients of $\widehat{\lind}^{CG}$ satisfy approximate KMS detailed balance, that is:
	\begin{align}
		\left|\gamma_{kl}^{\omega,\tilde{\omega}} -e^{-\beta(\frac{\omega+\tilde{\omega}}{2})} \overline{\gamma_{kl}^{-\omega,-\tilde{\omega}}}\right|\leq\Gamma_0\,e^{-(T(\alpha)\,\omega_-)^2/4} e^{\frac{\beta^2}{4T(\alpha)^2}}{\rm erf}\norbra{\frac{\beta}{2T(\alpha)}}\leq \frac{\beta\,\Gamma_0\,e^{-\frac{(T(\alpha)\,\omega_-))^2}{4}} e^{\frac{\beta^2}{4T(\alpha)^2}}}{\sqrt{\pi}\,T(\alpha)}\,,\label{eq:approxGamma}
	\end{align}
	and:
	\begin{align}
		\left |- i\tanh\norbra{\beta\norbra{\frac{\omega-\tilde{\omega}}{4}}}\gamma_{kl}^{\omega,\tilde{\omega}}\right |\leq\frac{\Gamma_0\beta}{2 \sqrt{2 e}\, T(\alpha)}\,.\label{eq:approxS}
	\end{align}
	whenever $\omega\neq\tilde{\omega}$. Choosing $T(\alpha) = T^*_{opt}(\alpha)$, Eq.~\eqref{eq:approxGamma} scales as $\bigo{\alpha(\beta\,\Gamma_0)^2e^{\alpha^2(\Gamma_0\beta)^2}}$ and Eq.~\eqref{eq:approxS} as $\bigo{\alpha\beta\,\Gamma_0}$.
\end{theorem}
Comparing the inequalities in Thm.~\ref{thm:approxDB} with  Eq.~\eqref{eq:DBdefinition} \ms{we can see that the rates of the dissipative part of $\widehat{\lind}^{CG}$ (that is, $\gamma_{kl}^{\omega,\tilde{\omega}}$) are approximately KMS detailed balanced, and that the Hamiltonian correction needed to ensure detailed balance at a global level is also small, as it only contributes at order $\bigo{T(\alpha)^{-1}}$, which is well within the error induced by the other approximations.} The approximate detailed balance in $\gamma_{kl}^{\omega,\tilde{\omega}}$ follows exclusively from the KMS-condition of the bath, so we expect this behavior to appear whenever the environment is thermal. \ms{The next section is devoted to finding an exactly detailed balanced Lindbladian close to $\widehat{\lind}^{CG}$.}


\section{Recovering Complete Positivity while preserving detailed balance}\label{sec:exactDB}

Thm.~\ref{thm:approxDB} highlights the role of the observation time $T(\alpha)$ in bringing the rates of the dissipateve part of the evolution close to a detailed balanced map: as shown in App.~\ref{app:davies}, in the limit $T(\alpha)\rightarrow\infty$ the evolution reduces to the usual Davies dynamics, which is GNS detailed balanced (and thus KMS as well). Moreover, for $T(\alpha)$ large, but finite, Thm.~\ref{thm:approxDB} shows that even if strict KMS detailed balance does not hold, we have an approximate behavior, with a correction scaling as $\bigo{T(\alpha)^{-1}}$. \ms
{Still, it remains to be understood how the effect of the Lamb-shift Hamiltonian can be incorporated. In general it does not commute with the system Hamiltonian, and its norm does not scale to zero for $T(\alpha)\rightarrow\infty$. This subtlety was previously noticed in~\cite{winczewski2021renormalization}, where it was proposed that such a correction should be incorporated as a \emph{renormalization} of the system Hamiltonian. Let us rewrite the total Hamiltonian as:
	\begin{align}
		H := H_S+\alpha \,V+ H_B = H^*_S-\frac{\alpha^2}{2} H_{LS}^{CG}+\alpha \,V+ H_B \,,\label{eq:96star}
	\end{align}
	where $H_{LS}^{CG}$ is the one given in Eq.~\eqref{eq:LSHam}, and we implicitly defined  $H^*_S$ to be the renormalized system Hamiltonian . Then, defining the renormalized interaction picture 
	(that is $X(t) := e^{i( H^*_S+ H_B)t}X\,e^{-i( H^*_S+ H_B)t}$), we obtain that the dynamics becomes:
	\begin{align}
		\partial_t\rho_{SB}(t) = -i\alpha[V(t),\rho_{SB}(t)]+\frac{i\alpha^2}{2}[H_{LS}^{CG}(t),\rho_{SB}(t)]\,.
	\end{align}
	Then, since the approximations in Sec.~\ref{sec:CPderivation} do not depend on the interaction picture that is used, we can carry out the same kind of derivation to describe the effect of the interaction in terms of Lindbladian dynamics. In particular, we can rewrite Eq.~\eqref{eq:lindEq} in the renormalized interaction picture as:
	\begin{align}
		&\partial_t\tilde{\rho}^{CG}_{S}(t)  = \nonumber\\
		&=-\frac{i}{2}\sqrbra{(H^{CG}_{LS})^*-H^{CG}_{LS},\tilde{\rho}^{CG}_{S}(t)}+\int_{-\infty}^{\infty}\frac{\de\omega^*}{2\pi}\;{\widehat{C}}_{kl}(\omega^*)\Big({\tilde A}^*_l(\omega^*)\tilde{\rho}^{CG}_{S}(t) ({\tilde A}_k^*)^\dagger(\omega^*)-\frac{1}{2}\{({\tilde A}_k^*)^\dagger(\omega^*){\tilde A}^*_l(\omega^*),\tilde{\rho}^{CG}_{S}(t)\}\Big)\,,
	\end{align}
	where it should be noticed that we use ${\tilde A}^*_l(\omega^*)$ and  $(H^{CG}_{LS})^*$ to denote the jump operators and the Lamb-shift term defined with respect to the renormalised system Hamiltonian\footnote{Indeed, even if the bounds and the functional expressions of the quantities in Sec.~\ref{sec:CPderivation} are not altered by changing the type of interaction picture, the matrix expression of the jump operators (and therefore of the Lamb-shift Hamiltonian) do depend on it.}. Then, we can prove the following:
	\begin{theorem}[Closeness of the Lamb-shift Hamiltonians]\label{thm:LSIgnore}
		The difference $(H^{CG}_{LS})^*-H^{CG}_{LS}$ can be bounded in operator norm as:
		\begin{align}
			\|(H^{CG}_{LS})^*-H^{CG}_{LS}\|_{\infty} \leq\frac{4\alpha^2\Gamma_0^2 \,T(\alpha)}{\sqrt{\pi}} +2\alpha^2\Gamma_0^2( \tau_0+\alpha^2\Gamma_0\, T(\alpha)^2) + \frac{4\alpha^4\Gamma_0^3\tau_0 \,T(\alpha)}{\sqrt{\pi}} + \alpha^4\Gamma_0^2 K_0\,.
		\end{align}
		Then, defining $\tilde{\rho}^{CG*}_{S,S}(t)$ to be the evolution satisfying:
		\begin{align}
			\partial_t\tilde{\rho}^{CG*}_{S}(t)&= \widehat{\lind}^{CG*}[\tilde{\rho}^{CG*}_{S}(t)] = 
			\\
			&=\int_{-\infty}^{\infty}\frac{\de\omega^*}{2\pi}\;{\widehat{C}}_{kl}(\omega^*)\Big({\tilde A}^*_l(\omega^*)\tilde{\rho}^{CG*}_{S}(t) ({\tilde A}_k^*)^\dagger(\omega^*)-\frac{1}{2}\{({\tilde A}_k^*)^\dagger(\omega^*){\tilde A}^*_l(\omega^*),\tilde{\rho}^{CG*}_{S}(t)\}\Big)\,,
		\end{align}
		we also have that:
		\begin{align}
			\|\tilde{\rho}^{CG}_{S}(t)-\tilde{\rho}^{CG*}_{S}(t)\|_1\leq (\Gamma_0t)\norbra{\frac{4\alpha^2\Gamma_0 \,T(\alpha)}{\sqrt{\pi}} +2\alpha^2( \Gamma_0\tau_0+\alpha^2\, (\Gamma_0T(\alpha))^2) + \frac{4\alpha^4\Gamma_0^2\tau_0 \,T(\alpha)}{\sqrt{\pi}} + \alpha^4\Gamma_0 K_0}\,.
		\end{align}
		Then, assuming that $K_0<\infty$, if one chooses $T(\alpha) = T^*_{opt}(\alpha)$ (see Eq.~\eqref{eq:optimalT}), the error scales as $\bigo{\alpha (\Gamma_0t) \sqrt{\Gamma_0 \tau_0}}$.
	\end{theorem}
	We refer to App.~\ref{app:LSI} for the proof of this result. 
	Notice that the error resulting from disregarding the Lamb-shift correction defined with respect to the renormalized system Hamiltonian is of the same order as the one encountered in Thm.~\ref{thm:coarseGrain}. 
	
	Interestingly, Thm.~\ref{thm:approxDB} is unchanged by this discussion, since the rates are defined only in terms of properties of the bath. The discussion above thus shows that $\hat{\lind}^{CG*}$ has rates that are close to be detailed balance, and no extra Lamb-shift term that could hinder the thermalization process.} This discussion naturally raises the question: is there a way to further approximate $\widehat{\lind}^{CG*}$	in the limit of large, but finite observation time $T(\alpha)$, so that the generated evolution is exactly detailed balance? A straightforward attempt is to expand Eq.~\eqref{eq:gammarateCG} and Eq.~\eqref{eq:SrateCG} in powers of $\bigo{T(\alpha)^{-1}}$, yielding a Lindbladian with rates:
\begin{align}
	\begin{cases}
		\tilde{\gamma}_{kl}^{\omega,\tilde{\omega}} = e^{-(T(\alpha)\,\omega_-)^2/4}{\widehat{C}}_{kl}\norbra{-\omega_+}\\
		\tilde{S}_{kl}^{\omega,\tilde{\omega}} = i\tanh\norbra{\frac{\beta\omega_-}{4}}\tilde{\gamma}_{kl}^{\omega,\tilde{\omega}}	
	\end{cases}\;\;.\label{eq:140}
\end{align}
Unfortunately, this approximation breaks the complete positivity of the generated semi-group. However, a more careful expansion is possible in the form of an approximation of $\gamma_{kl}^{\omega,\tilde{\omega}}$ which decays with $\bigo{T(\alpha)^{-1}}$ and that allows to generate a CP-semigroup. First, we can derive the following fact:
\begin{theorem}[Detailed balanced approximation of $\gamma_{kl}^{\omega,\tilde{\omega}}$]\label{thm:approxGammaSquareRoot}
	The coefficient $\gamma_{kl}^{\omega,\tilde{\omega}}$ associated to $\widehat{\lind}^{CG*}$ satisfies the bound:
	\begin{align}
		|\gamma_{kl}^{\omega,\tilde{\omega}} - e^{-(T(\alpha)\,\omega_-)^2/4}\;{\widehat{g}}_{k\lambda}\norbra{-\tilde\omega } {\widehat{g}}_{\lambda l}\norbra{-{\omega} }|\leq\frac{\Gamma\,\tau}{\sqrt{\pi \,e}\,T(\alpha)}\norbra{\sqrt{2\,\pi}+\sqrt{e}+\frac{2\sqrt{\pi} K}{{\sqrt{ \,e}\,T(\alpha)}}}\,,\label{eq:GammaSquareroot}
	\end{align}
	where ${\widehat{g}}_{k\lambda}(\omega)$ is the positive square root of ${\widehat{C}}_{kl}\norbra{\omega}$, that is ${\widehat{C}}_{kl}\norbra{\omega} = {\widehat{g}}_{k\lambda}(\omega){\widehat{g}}_{\lambda l}(\omega)$ (see Sec.~\ref{sec:thermalBaths}). Choosing $T(\alpha) = T^*_{opt}(\alpha)$, Eq.~\eqref{eq:GammaSquareroot} scales as $\bigo{\alpha \Gamma \sqrt{\Gamma\tau}}$ (assuming $K<\infty$).
\end{theorem}

See App.~\ref{app:thmapproxGammaSquare} for the proof. This approximation is sufficient to yield an exactly detailed balance physical evolution:
\begin{theorem}[KMS detailed balanced Lindbladian]\label{thm:dbLindlbadian}
	Consider the Lindbladian $\widehat{\lind}^{DB*}$ defined by the rates:
	\begin{align}
		\begin{cases}
			\tilde{\gamma}_{kl}^{\omega,\tilde{\omega}} = e^{-(T(\alpha)\,\omega_-)^2/4}\;{\widehat{g}}_{k\lambda}\norbra{-\tilde\omega } {\widehat{g}}_{\lambda l}\norbra{-{\omega} }\\
			\tilde{S}_{kl}^{\omega,\tilde{\omega}} = i\tanh\norbra{{\frac{\beta \,\omega_-}{4}}}\,\tilde{\gamma}_{kl}^{\omega,\tilde{\omega}}	
		\end{cases}\;.\label{eq:DBcoeffCP}
	\end{align}
	Then, $\widehat{\lind}^{DB*}$ is exactly detailed balance \ms{with respect to the thermal state of the renormalised Hamiltonian $H^*_S={ H_S +\alpha^2 H_{LS}^{CG}(\alpha)/2}$ (see Eq.~\eqref{eq:LSHam} for the definition of $H_{LS}^{CG}(\alpha)$)}, and it induces a completely positive evolution of the form:
	\begin{align}
		\widehat{\lind}^{DB*} [\rho] 
		&=-\frac{i}{2}\sqrbra{H^{DB}_{LS},\rho}+\int_{-\infty}^{\infty}\de\omega^*\;\Big({\hat A}_\lambda(\omega^*)\rho {\hat A}^\dagger_\lambda(\omega^*)-\frac{1}{2}\{{\hat A}^\dagger_\lambda(\omega^*){\hat A}_\lambda(\omega^*),\rho\}\Big)
	\end{align}
	where we introduced the jump operators:
	\begin{align}
		\hat{A}_\lambda(\omega^*) = \sqrt{\frac{T(\alpha)}{\sqrt{\pi}}}\int_{-\infty}^{\infty} \de\omega\;e^{-T(\alpha)^2\frac{(\omega^*-{\omega})^2}{2}}\;{\widehat{g}}_{\lambda l}\norbra{-{\omega}}A^*_l({\omega})\,.\label{eq:jumpOperatorsDB}
	\end{align}
\end{theorem}
We defer the proof to App.~\ref{app:dbLindbladian}.
Thm.~\ref{thm:approxGammaSquareRoot} and Thm.~\ref{thm:dbLindlbadian} together show that, in the limit of large observation time, one can derive a Lindbladian dynamics which is physical and exactly detailed balanced. \added{At this point, one can easily compare the jump operators obtained in Eq. \eqref{eq:jumpOperatorsDB} with those of previous works \cite{chen2023efficient,chen2023quantumthermal,Ding_2025}. In particular, the Gaussian ones initially proposed in \cite{chen2023efficient} corresponds to a simple Gaussian function depending on $\beta$ in Eq. \eqref{eq:jumpOperatorsDB}, which instead also contains the bath correlation function. As long as that bath correlation function decays sufficiently quickly, these jump operators will also fit the general characterization given in \cite{Ding_2025}.} 

It remains to be proven that the evolution induced by $\widehat{\lind}^{DB*}$ stays close to the one induced by $\widehat{\lind}^{CG*}$. Once more, this relies on Thm.~\ref{thm:DBerrorBoundIntegrated}. First, we have the bound:
\begin{theorem}[Closeness of $\widehat{\lind}^{CG*}$ and $\widehat{\lind}^{DB*}$]\label{thm:lindbladErrorDB}
	Given any state $\rho$,  we have :
	\begin{align}
		\|\widehat{\lind}^{CG*}[\rho]-\widehat{\lind}^{DB*}[\rho]\|_1\leq\frac{\Gamma}{2\sqrt{\pi}\,T(\alpha)}\norbra{8\tau +\beta e^{\frac{\beta^2}{16T(\alpha)^2}}}+ \frac{2K}{T(\alpha)^2}\,.\label{eq:errorDB}
	\end{align}
	Choosing $T(\alpha) = T^*_{opt}(\alpha)$ (and assuming $K< \infty$), Eq.~\eqref{eq:errorDB} scales as $\bigo{\alpha\Gamma \sqrt{\Gamma\tau}\,(\Gamma\tau+ \Gamma\beta e^{(\alpha\,\Gamma\beta)^2})}$.
\end{theorem}
The proof of this result is deferred to App.~\ref{app:thmlindbladErrorDB}. Then, using Thm.~\ref{thm:DBerrorBoundIntegrated} we obtain:
\begin{theorem}[Error in the evolution]\label{thm:DBerrorBoundIntegrated2}
	Consider the solutions $\tilde{\rho}^{CG*}_{S}({t})$ and $\tilde{\rho}^{DB*}_{S}({t})$ to the following differential equations:
	\begin{align}
		\begin{cases}
			\partial_{{t}}\,\tilde{\rho}^{CG*}_{S}({t})=\widehat{\lind}^{CG*} [\tilde{\rho}^{CG*}_{S}({t})]\\
			\tilde{\rho}^{CG*}_{S}(0) =\tilde{\rho}_{S}(0)
		\end{cases}\;;
		\qquad\qquad\qquad
		\begin{cases}
			\partial_{\tilde{t}}\,\tilde{\rho}^{DB*}_{S}({t})=\widehat{\lind}^{DB*} [\tilde{\rho}^{DB*}_{S}({t})]\\
			\tilde{\rho}^{DB*}_{S}(0) = \tilde{\rho}_{S}(0)
		\end{cases}\;.
	\end{align}
	Then, the trace distance between the two solutions can be bounded as:
	\begin{align}
		\|\tilde{\rho}^{CG*}_{S}({t})-\tilde{\rho}^{DB*}_{S}({t})\|_1\leq\frac{\Gamma\,{t}}{2\sqrt{\pi}\,T(\alpha)}\norbra{8\tau +\beta e^{\frac{\beta^2}{16T(\alpha)^2}}}+ \frac{2\,K\,{t}}{T(\alpha)^2}\,.\label{eq:cgDberror}
	\end{align}
	Assuming $K<\infty$, and choosing $T(\alpha) = T^*_{opt}(\alpha)$ (Eq.~\eqref{eq:optimalT} in Thm.~\ref{thm:coarseGrain}), the error scales as $\bigo{\alpha\,(\Gamma\, t) \sqrt{\Gamma\tau}\,(\Gamma\tau+ \Gamma\beta e^{(\alpha\,\Gamma\beta)^2})}$.
\end{theorem}
This result shows that there exists an exactly KMS detailed balanced Lindbladian inducing a CP-semigroup, and whose evolution stays close to the exact reduced dynamics with an error scaling at most linearly in time. Putting together the results in Thm.~\ref{thm:coarseGrain}, Thm.~\ref{thm:LSIgnore} and Thm.~\ref{thm:DBerrorBoundIntegrated2}, a simple triangle inequality gives:
\begin{align}
	\Norm{	\tilde{\rho}_S(t)-	\tilde{\rho}^{DB*}_S(t)}_1 \leq \Norm{	\tilde{\rho}_S(t)-	\tilde{\rho}^{CG}_S(t)}_1 + \Norm{	\tilde{\rho}^{CG}_S(t)-	\tilde{\rho}^{DB*}_S(t)}_1\leq \bigo{\alpha\,(\Gamma\, t)\norbra{\Gamma\tau+(\Gamma\beta )\,e^{(\alpha\,\Gamma\beta)^2}}}\,,\label{eq:totalError}
\end{align}
where we choose $T(\alpha) \simeq \bigo{(\alpha\Gamma)^{-1}}$. This choice is justified by the fact that, thanks to the specific dependency of the error term in Eq.~\eqref{eq:cgDberror} (that is, inversely proportional to $T(\alpha)$), the $T^*_{opt}(\alpha)$ obtained by minimizing Eq.~\eqref{eq:totalError} has the same scaling as the one obtained in Eq.~\eqref{eq:optimalT}.

The derivation just presented is based on the only extra assumption that the bath satisfies the KMS condition in Eq.~\eqref{eq:corrKMS}, that is, that it should be in a Gibbs state. This shows that, similarly to the derivation of Davies \cite{davies1974markovian,davies1976markovian}, KMS detailed balance is a generic property of systems in contact with thermal baths. In this context, it is also interesting to point out the role of temperature in the error estimates in Thm.~\ref{thm:lindbladErrorDB} and Thm.~\ref{thm:DBerrorBoundIntegrated2}, which acts as an energy scale with respect to which $\alpha\Gamma$ needs to be small for the weak coupling regime to be valid. This is a novelty with respect to the derivation in Sec.~\ref{sec:CPderivation}, where only the bath timescales $\Gamma^{-1}$ and $\tau$ play a role. This highlights the special role that thermal equilibrium plays in ensuring detailed balance. \added{The exponential dependence of Eq.~\eqref{eq:totalError} on the temperature makes the bound less relevant for $\beta$ very big (at fixed $\alpha$). This can be taken care of by choosing $T(\alpha)\simeq\bigo{\beta \alpha^{-1}}$. Still, doing so makes the whole error bound scale as $\beta^{-1}$, so this choice will not be optimal when treating high temperatures ($\beta\rightarrow0$). For this reason, there exists a temperature $\beta^*$ for which one passes from one regime in which it is better to choose a temperature dependent observation time, to one in which $T(\alpha)$ is independent on $\beta$.}

\ms
{Finally, it should be noticed that the equilibrium is attained with respect to the thermal state of the renormalized Hamiltonian $H^*_S= H_S +\alpha^2 H_{LS}^{CG}(\alpha)/2$ (where, for clarity, we have highlighted the dependency of the Lamb-shift Hamiltonian on $\alpha$). One might wonder whether $H^*_S$ is somehow related to the mean force Hamiltonian $H^{(mf)}_S$, implicitly defined by the relation:
	\begin{align}
		e^{-\beta\,H^{(mf)}_S} := \TrR{B}{\frac{e^{-\beta\,H}}{\mathcal{Z}_B}}\,.
	\end{align}
	This would mean that at equilibrium, the state of the system is the reduced density matrix of the global Gibbs state, rather than thermal with respect to the bare system Hamiltonian. It is quite puzzling that this is not the case, as noticed already in~\cite{winczewski2021renormalization}, and discussed in App.~\ref{app:meanFieldH}. However, the following is still the case:
	\begin{theorem}[Closeness of fixed points]\label{thm:interactionInsensitive}
		Let $\gamma_S$ be the Gibbs state of the bare Hamiltonian, and define $\gamma_S^*$ with respect to the renormalized Hamiltonian  $H^*_S= H_S +\alpha^2 H_{LS}^{CG}(\alpha)/2$. It holds that:
		\begin{align}
			\|\gamma_S - \gamma_S^*\|_1 \leq \alpha^2 \beta \,\Gamma_0\,.
		\end{align}
	\end{theorem}
	We refer to the proof of this theorem to App.~\ref{app:intInsensitive}. This shows that even if the thermal fixed point of the evolution does not coincide with the one expected from thermodynamic considerations (i.e. the zeroth law of thermodynamics), this discrepancy is one order larger in $\alpha$ than the one arising from the approximations to the dynamics.
}



\section{Quantum algorithmic simulation}
\label{sec:qualg}

The KMS Lindbladian defined in this work can be efficiently simulated on a quantum computer with the schemes analyzed in \cite{chen2023quantumthermal,chen2023efficient,Ding_2025}. These are based on the paradigms of block-encoding and linear combinations of unitaries, standard in the quantum algorithms literature \cite{Childs2012,Gily_n_2019,Low_2019,GrandtutorialPRXQuantum}. We now outline the inner workings of those algorithms, and  summarize the requirements in terms of quantum resources.

The Lindbladian simulation algorithms require a block encoding of the jump operators and of the coherent term. In this case, this can be done through an approximation of the integral over a ``time'' variable that defines them. We recall the expression for the jump operators obtained in Eq.~\eqref{eq:jumpOperatorsDB}, and write it as the integral:

\begin{align}\label{eq:dbJump}
	\hat{A}_\lambda(\omega^*) &= \sqrt{\frac{T(\alpha)}{\sqrt{\pi}}}\int_{-\infty}^{\infty} \de\omega\;e^{-T(\alpha)^2\frac{(\omega^*-{\omega})^2}{2}}\;{\widehat{g}}_{\lambda l}\norbra{-{\omega}}A_l({\omega})\,=
	\\
	& = \int_{-\infty}^{\infty} \dt\; e^{i \omega^* t } A(t)  \left ( \frac{1}{\sqrt{\sqrt{\pi} T(\alpha)}}\,\norbra{e^{\frac{-t^2}{2 T(\alpha)^2}}  \ast g_{\lambda l}}(t) \right )=
	\\
	& :=  \int_{-\infty}^{\infty} \dt \; e^{i \omega^* t } A(t) G(t)  .
\end{align}

We can also write the coherent term (or Lamb-shit Hamiltonian) $H^{DB}_{LS}$ as an integral over time with the filter function $F(t)$ in Eq.~\eqref{eq:LSft}, as the Fourier transform of the coefficients in the last line of Eq.~\eqref{eq:DBcoeffCP}. Then, in order to simulate the master equation generated by these, we require:
\begin{itemize}
	\item A block encoding of the initial jump operators $A$. When $A$ is Hermitian but not unitary, this requires the construction of a larger unitary gate with a block or sub-matrix corresponding to $A$, which can be done for several classes of matrices of interest \cite{Gily_n_2019,Nguyen_2022,camps2023explicit,S_nderhauf_2024}.
	\item A coherent quantum state (oracle) encoding an approximation to the functions $G(t)$ and $F(t)$. This is constructed through an approximation of the time-integrals with a discrete sum, by truncating them to a finite time range, and then approximating with a numerical cuadrature. This can be efficiently done given the smoothness and decay of both  $G(t)$ and $F(t)$ \cite{Ding_2025}.
	\item Black-box access to the controlled time evolution $e^{iHt}$ for the relevant times in the numerical cuadrature. The largest of these depends on the timescales $T(\alpha),\Gamma^{-1},\beta$ of the functions  $G(t)$ and $F(t)$. This can be achieved with any of the various quantum simulation methods for unitary dynamics.
\end{itemize}
Once those ingredients are in place, one can resort to efficient Lindbladian simulation algorithms such as \cite{cleve_et_al,li_et_al,chen2023efficient,Ding_2024}.  According to the analysis of \cite{chen2023efficient,Ding_2025} applied to our setting, the total computational cost of running the algorithm for a time $t$ and up to error $\epsilon$ is roughly $\tilde{\mathcal{O}}\left(t(T(\alpha)+\Gamma^{-1}+\beta)  \text{polylog}(\frac{1}{\epsilon}) \right)$, which is by construction efficient in time and approximation error. This should be contrasted with the case of Davies dynamics, that for a many-body system of $N$ particles involves the RWA, meaning that $T(\alpha) \propto \alpha^{-1} \ge e^{\Omega(N)}$; thus the simulation of Davies maps takes exponential time in the number of components.

One can compare the Lindbladian obtained in this work with that in~\cite{chen2023efficient}, which we report here for convenience:
\begin{align}
	\widehat{\lind}^{AL} [\rho] 
	&=-\frac{i}{2}\sqrbra{H^{AL}_{LS},\rho}+\int_{-\infty}^{\infty}\de\omega^*\;\Big({\hat A^{AL}_\lambda}(\omega^*)\rho ({\hat A^{AL}_\lambda}(\omega^*))^\dagger-\frac{1}{2}\{({\hat A^{AL}_\lambda}(\omega^*))^\dagger{\hat A_\lambda^{AL}}(\omega^*),\rho\}\Big),\label{eq:lindAl}
\end{align}
where we introduced the jump operators $\hat{A}^{AL}_\lambda(\omega^*)$ defined as\footnote{In order to highlight the similarity with our derivation, we use a different notation as compared to~\cite{chen2023efficient}. In particular, the jump operators $\hat{A}^{AL}_\lambda(\omega^*)$ are related to the ones defined there as $\hat{A}^{AL}_\lambda(\omega^*) = \sqrt{\gamma(\omega^*)}\hat{\mathbf{A}}_\lambda(\omega^*)$. }:
\begin{align}
	\hat{A}^{AL}_\lambda(\omega^*) = \sqrt{\frac{1}{(\sqrt{2}\sigma_E)\sqrt{\pi}}}\int_{-\infty}^\infty\de\omega\;e^{-\frac{(\omega^*+\omega_\gamma)^2}{4 \sigma_\gamma^2}}e^{-\frac{(\omega^*-\omega)^2}{2 (\sqrt{2}\sigma_E)^2}}\;A_\lambda({\omega})\,,\label{eq:jumpAlg}
\end{align}
whereas $H^{AL}_{LS}$ is chosen according to Thm.~\ref{thm:KMSLindstructuralcharacterization} and $\omega_\gamma,\,\sigma_\gamma$ and $\sigma_E$ are free parameters satisfying the relation:
\begin{align}
	\beta = \frac{2\omega_\gamma}{\sigma_\gamma^2 + \sigma_E^2}\,.\label{eq:dbConstr}
\end{align}

By setting $\sigma_E := (\sqrt{2} \,T(\alpha))^{-1}$, the two Gaussian filters become identical; moreover, we see that the role of ${\widehat{g}}_{\lambda l}\norbra{-{\omega}}$ is taken by the term $G_{\omega_\gamma,\sigma_\gamma}(\omega^*):=\exp\norbra{-\frac{(\omega^*+\omega_\gamma)^2}{4 \sigma_\gamma^2}}$ in Eq.~\eqref{eq:jumpAlg}.



\section{Conclusions}
\label{sec:conclusions}

We have proven the existence of an explicit master equation that approximates the evolution of a quantum system weakly coupled to a bath. Notably, this comes equipped with error estimates that improve all previous derivations, and it relies on approximations that apply to many-body systems. As a main conceptual novelty, our master equation satisfies an exact notion of KMS detailed balance, thus guaranteeing thermality. \added{In our derivation, we introduced the Gaussian profile in order to coarse-grain the dynamics. As a result, the resulting generator is more akin to that initially derived in \cite{chen2023efficient}. However, notice that $i)$ the Lindbladian coefficients also depend on the correlation function of the bath, which may not be Gaussian and $ii)$ the choice of Gaussian is convenient but arbitrary, and other similar profile functions will likely yield similar bounds.}

This construction paves the way for the exploration of the open system dynamics of many-body physics, beyond the simple models commonly used in the literature \cite{orus-review,fazio2025manybodyopenquantumsystems}. This can allow us to uncover interesting new features of the thermal behavior of many-body systems \cite{kastoryano2024littlebit}. It also further motivates the development of quantum algorithms for Gibbs sampling that explicitly simulate these processes \cite{chen2023efficient,Ding_2025} in a way similar to how product formulae simulate unitary time evolution \cite{Childs_2021}. 
From the open systems and thermodynamics point of view, the master equation derived here can potentially yield insights into recently studied problems, such as the question of ``local'' vs. ``global'' master equations \cite{Rivas_2010,Hofer_2017,Correa17,DeChiara_2018}, or the corrections associated to Hamiltonians of mean force \cite{Potts_2021,Trushechkin22}.

Our results are also relevant in the wider context of quantum information and computation, where various works on Lindbladiands with KMS detailed balance have recently appeared. It has been shown that, at low temperature, they can reproduce arbitrary quantum computations \cite{rouze2024efficient}, they feature in proofs of spatial Markov properties of Gibbs states \cite{chen2025quantumgibbs,kato2025clustering}, and they can be used to efficiently prepare quantum states of interest \cite{rouze2024optimal,smid2025polynomialtime,tong2025fastmixing,zhan2025rapidquantumgroundstate}.

\acknowledgements

We thank Ariane Soret for pointing out the related and independent work~\cite{soret2022thermodynamic} after the completion of this manuscript, \added{and Zhiyan Ding and Yongtao Zhan for pointing to an error in a previous proof of KMS detailed balance}. We acknowledge support from the Spanish Agencia Estatal de Investigacion through the grants  ``IFT Centro de Excelencia Severo Ochoa CEX2020-001007-S” and ``PCI2024-153448" (for M.S.), and ``Ram\'on y Cajal RyC2021-031610-I'', financed by MCIN/AEI/10.13039/501100011033 and the European Union NextGenerationEU/PRTR (for A.M.A.). This project was funded within the QuantERA II Programme that has received funding from the EU’s H2020 research and innovation programme under the GA No 101017733.


\bibliography{bib.bib}

\appendix

\onecolumngrid


\newpage
\begin{center}
	\normalsize{\underline{\textbf{Appendix: Table of contents}}}
\end{center}
\begin{enumerate}
	\item[\underline{A:}] \hyperref[app:channels]{Abstract characterization of channels}
	\item[\underline{B:}] \hyperref[app:exactDynamics]{Derivation of the exact reduced system dynamics}
	\item[\underline{C:}] \hyperref[app:bath]{Positivity of the bath spectral function}
	\item[\underline{D:}] \hyperref[app:relationTimescales]{Relation between different timescales}
	\item[\underline{E:}] \hyperref[app:fastestRate]{Proof of Thm.~\ref{thm:fastestRate}: Fastest rate of change in the system}
	\item[\underline{F:}] \hyperref[app:errorTerms]{Bounding the different error terms}
	\begin{enumerate}[1.]
		\item \hyperref[app:born]{Proof of Thm.~\ref{thm:bornApproximation}: Born approximation}
		\item \hyperref[app:markov]{Proof of Thm.~\ref{thm:markovApproximation}: Markov approximation}
		\item \hyperref[app:redfield]{Proof of Thm.~\ref{thm:redfield}: Redfield approximation}
		\item \hyperref[app:cg]{Proof of Thm.~\ref{thm:smoothing} and Thm.~\ref{thm:coarseGrain}: Time averaging procedures}
	\end{enumerate}
	\item[\underline{G:}] \hyperref[app:cgLindProperties]{Properties of the coarse-grained Lindbladian}
	\begin{enumerate}[1.]
		\item \hyperref[app:hLShermitian]{Hermiticity of $H_{LS}$}
		\item \hyperref[app:compRates]{Decomposition of the rates}
		\item \hyperref[app:quasiLocalJump]{Quasilocality of the jump operators}
		\item \hyperref[app:davies]{Reduction to Davies generator}
		\item \hyperref[app:approxDB]{Proof of Thm.~\ref{thm:approxDB}: Approximate detailed balance of the coarse-grained Lindbladian}
	\end{enumerate}
	\item[\underline{H:}] \hyperref[app:exactDBLind]{Exactly detailed balanced Lindbladian}
	\begin{enumerate}[1.]
		\item \hyperref[app:thmapproxGammaSquare]{Proof of Thm.~\ref{thm:approxGammaSquareRoot}: Approximation to the coarse-grained rates}
		\item \hyperref[app:dbLindbladian]{Proof of Thm.~\ref{thm:dbLindlbadian}: $\widehat{\lind}^{DB} $ is exactly detailed balanced and CP}
		\item \hyperref[app:quasiLocalJump2]{Quasilocality of the jump operators}
		\item \hyperref[app:davies2]{Reduction to Davies generator}
		\item \hyperref[app:thmlindbladErrorDB]{Proof of Thm.~\ref{thm:lindbladErrorDB}: $\widehat{\lind}^{CG} $ and $\widehat{\lind}^{DB} $ are close}
	\end{enumerate}
\end{enumerate}


\section{Abstract characterization of channels}\label{app:channels}

The discussions of this appendix are inspired by~\cite{gorini1976completely, parravicini1977generator}, as recounted in~\cite{carlen2017gradient}. We begin by defining an orthonormal basis of matrices $\{A_k\}$ (with respect to Hilbert-Schmidt), that we choose in such a way that $\{A_k(\omega)\}$ is also a basis (with possibly different indexes\footnote{The most common choice in this context is $A_{(i,j)} := \ketbra{E_i}{E_j}$, so that $\{A_{(i,j)}\}=\{A_{(i,j)}(\omega)\}$ (since we have the equality $A_{(i,j)}(\omega) = \delta_{\omega,\omega_{i,j}} A_{(i,j)}$).}). Then, any superoperator $\Phi$ can be rewritten in the form:
\begin{align}
	\Phi[X]&=\sum_{k,l}\;c_{kl} \,A_l XA_k^\dagger = \sum_{k,l,\omega,\tilde{\omega}}\; c^{\omega,\tilde{\omega}}_{kl}\,A_l(\omega)  XA_k^\dagger(\tilde{\omega}) \,.
\end{align}
The coefficients $c_{kl}$ are the GKS matrix (for Gorini-Kossakowski-Sudarshan). This representation is uniquely defined once the basis $\{A_k\}$ is fixed, and it encodes the most important properties of the channel $\Phi$ in a simple manner. For example, $\Phi$ is Hermitian preserving (that is, $\Phi(X^\dagger)=\Phi(X)^\dagger$) if and only if $c_{kl}$ is Hermitian:
\begin{align}
	\Phi(X^\dagger) = \sum_{k,l}\;c_{kl} \,A_l X^\dagger A_k^\dagger = \norbra{\sum_{k,l}\;\overline{c_{lk}} \,A_l X A_k^\dagger}^\dagger\,.
\end{align}
In order for the expression to be equal to $\Phi(X)^\dagger$, thanks to the uniqueness of $c_{kl}$, we need that $c_{kl}=\overline{c_{lk}}$, proving the claim. A less obvious and more interesting fact is the following:
\begin{theorem}[Choi's theorem~\cite{choi1975completely}]\label{thm:choi}
	The channel $\Phi$ is completely positive if and only if the corresponding GKS matrix is positive semidefinite.
\end{theorem}
So, for channels, not only $c_{kl}$ is a Hermitian matrix, but it also has a non-negative spectrum. 
When referring to a one-parameter semigroup $\Phi_t = e^{t\lind}$, it is often useful not to work directly in terms of the full $c_{kl}$, but with a reduced version thereof. Let us fix $A_1 := \id$. With this choice at time $t=0$, $c_{kl}=\delta_{k,1}\delta_{l,1}$, since $\Phi_0 = \idO$ (the identity channel). Then, in order to see how the evolution in time affects $\Phi_t$, we define the reduced GKS matrix $\tilde{c}_{kl}$ to be the $(d-1)^2\times(d-1)^2$ matrix obtained by eliminating from $c_{kl}$ the first row and column. This allows for the following characterization~\cite{parravicini1977generator}:
\begin{theorem}\label{thm:lindCPreq}
	The one-parameter semigroup $\Phi_t = e^{t\lind}$ is CP if and only if the reduced GKS matrix $\tilde{c}_{kl}$ of $\lind$ is positive semidefinite.
\end{theorem}
\begin{proof}
	For clarity, let us introduce the notation $c_{kl}(\Psi)$ to denote the GKS matrix associated to the map $\Psi$ (and similarly for the reduced GKS matrix). Suppose $\Phi_t = e^{t\lind}$ is CP at all times. Let us introduce the difference:
	\begin{align}
		\frac{c_{kl}(\Phi_t-\idO)}{t} =\delta_{k,1}\delta_{l,1}\frac{(c_{1,1}(\Phi_t)-1)}{t} + \frac{\tilde{c}_{k,l}(\Phi_t)}{t}\,.
	\end{align}
	Since this decomposition holds for all times, we can take the limit $\lim_{t\rightarrow0^+}\frac{c_{kl}(\Phi_t-\idO)}{t} = c_{kl}(\lind)$, so that $\tilde{c}_{kl}(\lind) = \lim_{t\rightarrow0^+}\frac{\tilde{c}_{k,l}(\Phi_t)}{t}$. Then, since $\tilde{c}_{k,l}(\Phi_t)$  is positive semi-definite, the same has to hold for $\tilde{c}_{kl}(\lind)$ as well. 
	
	On the other hand, suppose that $\tilde{c}_{kl}(\lind)\geq0$. Then, expanding in $t$ we have that:
	\begin{align}
		{c}_{kl}(\Phi_t) = {c}_{kl}(\idO) + t\,{c}_{kl}(\lind)  +\bigo{t^2}\,.
	\end{align}
	Then, since ${c}_{kl}(\idO)=\delta_{k,1}\delta_{l,1}$ and $\tilde{c}_{kl}(\lind)$ is positive semidefinite, for $t$ sufficiently small (say $t^*$), ${c}_{kl}(\Phi_{t^*})$ is also positive semidefinite. The claim follows from the fact that $\Phi_t = \prod^{t/t^*}_{i=1} \Phi_{t^*}$, together with Thm.~\ref{thm:choi}.
\end{proof}

We can also characterize the property that $\Phi_t = e^{t\lind}$ is trace preserving independently of whether it is CP or not. Before doing so it is useful to notice that the action of any Hermitian-preserving map $\lind$ can be decomposed as:
\begin{align}
	\lind[X] = G X + X G^\dagger + \sum_{k,l}\;\tilde{c}_{kl} \,A_l XA_k^\dagger\,,\label{eq:decHermPres}
\end{align}
where $G := \sum_{l}\;{c}_{1,l} \,A_l$. Notice that only the reduced GKS matrix enters the  sum. Now $G$ can be decomposed in self-adjoint and anti-self-adjoint components as $G = K - i H$ (where both $K$ and $H$ are self-adjoint), and we can explicitly compute the anti-self-adjoint part as:
\begin{align}
	H := \frac{i}{2}\sum_{l>1} \norbra{c_{1,l} A_l - c_{l,1} A_l^\dagger}\,.
\end{align}
so $H$ is traceless. In this context, the requirement that $\Phi_t = e^{t\lind}$ is trace preserving fixes the expression of $K$:
\begin{theorem}\label{thm:tracePreserving}
	Let $\lind$ be a Hermitian preserving map generating a trace-preserving semigroup. Then, the self-adjoint component of $G$ in Eq.~\eqref{eq:decHermPres} reads:
	\begin{align}
		K = -\frac{1}{2}\sum_{k,l}\;\tilde{c}_{kl} A_k^\dagger A_l\,,\label{eq:app:kStruct}
	\end{align}
	that is, we can rewrite the action of $\lind$ as:
	\begin{align}
		\lind[X] = -i [H, X] + \sum_{k,l}\;\tilde{c}_{kl} \,\norbra{A_l XA_k^\dagger -\frac{1}{2}\{A_k^\dagger A_l, X\}}\,.\label{eq:lindTrPres}
	\end{align}
\end{theorem}
\begin{proof}
	First, notice that $\Phi_t$ is trace preserving if and only if $\Phi_t^\dagger(\id) = \id$. Since we want this relation to hold for all times, it directly implies that $\lind^\dagger(\id) = \lim_{t\rightarrow0}(\Phi_t^\dagger(\id)-\id)/t = 0$. Using Eq.~\eqref{eq:decHermPres} this condition can be expressed as:
	\begin{align}
		\lind^\dagger(\id) = G^\dagger + G + \sum_{k,l}\;\overline{\tilde{c}_{kl}}  \,A_l^\dagger A_k=G^\dagger + G + \sum_{k,l}\;\tilde{c}_{kl}  \,A_k^\dagger A_l  = 0\,,
	\end{align}
	where we implicitly used the Hermiticity of $\tilde{c}_{kl}$. Then, the claim follows from noticing that $G + G^\dagger  = 2 K$.
\end{proof}
One can recognize in Eq.~\eqref{eq:lindTrPres} the usual Lindblad form of the generator of a one-parameter semigroup. Still, this structural characterization is independent on whether $\Phi_t= e^{t \lind}$  is CP or not. On the other hand, Thm.~\ref{thm:lindCPreq} states that if this is the case for all times, then the matrix $\tilde{c}_{kl}$ in Eq.~\eqref{eq:lindTrPres} has to be positive semidefinite.

Finally, we provide an extended version of the structural characterization of KMS detailed balanced Lindbladians from Thm.~\ref{thm:KMSLindstructuralcharacterization}:
\begin{theorem}[Structural characterization of KMS Lindbladians~\cite{amorim2021complete}]
	Let $\lind$ be the generator of a trace preserving evolution of the form $\Phi_t = e^{t\lind}$. Assume that $\lind$ satisfies KMS detailed balance. Then, the following two conditions should hold:
	\begin{align}
		\begin{cases}
			\tilde{c}^{\,\omega,\tilde{\omega}}_{kl} =\overline{ (\tilde{c}^{\,-\omega,-\tilde{\omega}}_{kl})}e^{-\beta\frac{(\omega+\tilde{\omega})}{2}}\\
			\Delta^{1/2}_{\gamma_S}[G^\dagger] = G
		\end{cases}\,,\label{eq:app:a10}
	\end{align}
	where $G$ was introduced in Eq.~\eqref{eq:decHermPres}. Moreover, this implies the decomposition:
	\begin{align}
		&\lind[\rho] = -i[H_{LS}^0 + H_{LS}, \rho] + \sum_{k,l,\omega,\tilde{\omega}}\; \gamma^{\omega,\tilde{\omega}}_{kl}\,\norbra{A_l(\omega)  \rho A_k^\dagger(\tilde{\omega}) -\frac{1}{2}\{A_k^\dagger(\tilde{\omega})A_l(\omega),\rho\}}\,,\label{eq:app:a11}
	\end{align}
	where $H_{LS}^0$ and $H_{LS}$ are two Hermitian operators such that $[H_{LS}^0,\gamma_S]=0$, and $H_{LS} =\frac{1}{2}\sum_{\omega\neq\tilde{\omega}} S^{\omega,\tilde{\omega}}_{kl} A_k^\dagger(\tilde{\omega})A_l(\omega)$, where the coefficients $\gamma^{\omega,\tilde{\omega}}_{kl}$ and $S^{\omega,\tilde{\omega}}_{kl}$ satisfy:
	\begin{align}
		\begin{cases}
			\gamma_{kl}^{\omega,\tilde{\omega}} = \overline{\gamma_{kl}^{-\omega,-\tilde{\omega}}}\,e^{-\beta(\frac{\omega+\tilde{\omega}}{2})} \\
			S_{kl}^{\omega,\tilde{\omega}} = i\tanh\norbra{\beta\norbra{\frac{\omega-\tilde{\omega}}{4}}}\gamma_{kl}^{\omega,\tilde{\omega}}	
		\end{cases}\;\;\;.\label{eq:app:a12}
	\end{align}
\end{theorem}
\begin{proof}
	Thanks to Thm.~\ref{thm:KMSstructuralcharacterization} we know that:
	\begin{align}
		c^{\omega,\tilde{\omega}}_{kl}(\Phi_t) =\overline{ c^{-\omega,-\tilde{\omega}}_{kl}}(\Phi_t)e^{-\beta\frac{(\omega+\tilde{\omega})}{2}}\,.
	\end{align}
	Since this property is independent of $t$, it can be directly transferred to $\tilde{c}^{\omega,\tilde{\omega}}_{kl}(\lind)$. This gives the first condition in Eq.~\eqref{eq:app:a10}. Then, assuming that this holds, we can simplify the difference between $\lind$ and $\tilde\lind_P$ as:
	\begin{align}
		\lind[X]-\J_ {\gamma_S}\circ \lind^\dagger\circ \J_ {\gamma_S}^{-1}[X] = (G-\Delta_{\gamma_S}^{1/2}[G^\dagger])X + X (G-\Delta_{\gamma_S}^{1/2}[G^\dagger])^\dagger\,.
	\end{align}
	In order for $\lind$ to be KMS detailed balance, the difference above must be zero, that is it should hold that $G=\Delta_{\gamma_S}^{1/2}[G^\dagger]$, which proves the second condition in Eq.~\eqref{eq:app:a10}. 
	
	Then, let us decompose $G$ as $G = G_0 + G_\perp$, where $G_0$ is the diagonal part of $G$ in the eigenbasis of $H_S$. Obviously, $G_0$ always satisfies $G_0=\Delta_{\gamma_S}^{1/2}[G_0^\dagger]$, so it does not affect whether $\lind$ is detailed balanced or not (we denote the corresponding anti-self-adjoint part as $H^0_{LS}$). On the other hand, the components of $G_\perp := K_\perp - i H_\perp$ need to satisfy:
	\begin{align}
		H_\perp =- i\,\frac{(\idO-\Delta_{\gamma_S}^{1/2})}{(\idO+\Delta_{\gamma_S}^{1/2})}[K_\perp] \,,
	\end{align}
	for Eq.~\eqref{eq:app:a10} to hold.
	At this point, we can use the expression of $K$ found in Thm.~\ref{thm:tracePreserving} to obtain:
	\begin{align}
		H_\perp &= \sum_{k,l,\omega\neq\tilde{\omega}} \;\tilde{c}^{\,\omega,\tilde{\omega}}_{kl}\norbra{\frac{i}{2}\,\frac{(\idO-\Delta_{\gamma_S}^{1/2})}{(\idO+\Delta_{\gamma_S}^{1/2})}[A_k^\dagger(\tilde{\omega})A_l({\omega})] } = \frac{i}{2}\sum_{k,l,\omega\neq\tilde{\omega}} \;\tilde{c}^{\,\omega,\tilde{\omega}}_{kl}\norbra{\frac{1-e^{\frac{-\beta(\omega-\tilde{\omega})}{2}}}{1+e^{\frac{-\beta(\omega-\tilde{\omega})}{2}}}}A_k^\dagger(\tilde{\omega})A_l({\omega}) =
		\\
		&=\frac{i}{2}\sum_{k,l,\omega\neq\tilde{\omega}} \;\tilde{c}^{\,\omega,\tilde{\omega}}_{kl} \tanh\norbra{\beta\norbra{\frac{\omega-\tilde{\omega}}{4}}}\; A_k^\dagger(\tilde{\omega})A_l({\omega})\,,\label{eq:app:a17}
	\end{align}
	where we implicitly used the fact that $\Delta_{\gamma_S}[X(\tilde{\omega})^\dagger Y(\omega)] = e^{-\beta(\omega-\tilde{\omega})} X(\tilde{\omega})^\dagger Y(\omega)$, and the definition of the hyperbolic tangent. This proves the claim.
\end{proof}
The theorem just proven shows that for KMS-detailed balanced Lindbladians the two components of the matrix $G$ are closely related. This was used by the authors of~\cite{amorim2021complete} and~\cite{chen2023efficient} in order to enforce KMS detailed balance given a reduced GKS matrix. Indeed, suppose that $\tilde{c}^{\,\omega,\tilde{\omega}}_{kl}$ satisfies the first condition of Eq.~\eqref{eq:app:a12}. Then, it is always possible to define a KMS-detailed balanced Lindbladian by proceeding in two steps: first, we use the structural characterization of $K $ in Eq.~\eqref{eq:app:kStruct} to enforce trace preservation; then, since this term can generically break the KMS detailed balance of the resulting Lindbladian, we artificially add $H_\perp$ as defined in Eq.~\eqref{eq:app:a17} in order to recover exact detailed balance.

This discussion is particularly interesting when considering Thm.~\ref{thm:dbLindlbadian}--\ref{thm:lindbladErrorDB}: there we show that through a series of approximations one can define an exactly detailed balanced reduced evolution that stays close to the real dynamics. As it appears, whereas the condition on the reduced GKS matrix (that is, the rates $	\tilde{\gamma}_{kl}^{\omega,\tilde{\omega}}$) naturally emerges from the microscopic derivation, the one for $S_{kl}^{\omega,\tilde{\omega}}$ in the case $\omega\neq\tilde{\omega}$ has to be introduced by hand (similarly to what was proposed by~\cite{amorim2021complete, chen2023efficient}). 

Let us split the Lindbladian as $\lind = \lind_1 + \lind_2$, where $\lind_1[X] := \sum_{k,l}\;\tilde{c}_{kl} \,A_l XA_k^\dagger$ is exactly detailed balanced, and $\lind_2$ denotes the rest of the terms in Eq.~\eqref{eq:decHermPres}. Then, the evolution $\Phi_t = e^{t\lind}$ can be represented in Trotter form, as $\Phi_t = \lim_{n\rightarrow\infty} (e^{\frac{t}{n}\lind_2}e^{\frac{t}{n}\lind_1})^n$. In this context, every odd step corresponds to an exactly detailed balanced evolution, which breaks trace preservation (since this was enforced by the presence of $\lind_2$).  On the other hand, since $\lind_1+\lind_2$ is trace preserving, every even step acts as a projection into the space of matrices with constant trace. This generically breaks detailed balance. Then, the results of Thm.~\ref{thm:dbLindlbadian}--\ref{thm:lindbladErrorDB} can be read as two facts: first, that even if imposing trace preservation breaks detailed balance, this happens well within the uncertainty we have about the exact evolution; second, that we can counteract this effect (at the same order of approximation) by introducing a small purely unitary rotation.


\section{Derivation of the exact reduced system dynamics}\label{app:exactDynamics}
We provide a derivation of Eq.~\eqref{eq:exactReducedDynamics} from the main text using conditions~\ref{it:productInCond}--\ref{it:gaussianBath}. 
Before doing so, we remark a special property of Gaussian baths that follows from Wick's theorem in Eq.~\eqref{eq:WicksTheorem} and condition~\ref{it:noEnergyShift}~\cite{nathan2020universal}:
\begin{align}
	\TrR{B}{B_k(t) \mathcal{V}(t,0)\gamma_B\mathcal{V}^\dagger(t,0)} = -i\alpha \int_{0}^{t}\de s\;&\TrR{B}{B_k(t) \mathcal{V}(t,s) [V(s), \mathcal{V}(s,0)\gamma_B\mathcal{V}^\dagger(s,0)]\mathcal{V}^\dagger(t,s)} \,.\label{eq:wicksTheoremConsequence}
\end{align}
This identity can be verified by Trotterizing the evolution of $\mathcal{V}(t,0)$. For simplicity, here we present the proof for the asymmetric case, but the general identity follows the same reasoning:
\begin{align}
	\TrR{B}{B_k(t) \mathcal{V}(t,0)\gamma_B} &= \lim_{n\rightarrow\infty} \TrR{B}{B_k(t) \prod_{j=0}^{n}\norbra{\idO -i\alpha \Delta t\,V\norbra{j\Delta t}}\gamma_B} =
	\\
	&= \lim_{n\rightarrow\infty}\sum_{j'=0}^n\; \TrR{B}{B_k(t)\norbra{\idO -i\alpha \Delta t\,V\norbra{j'\Delta t}}\gamma_B}\TrR{B}{\prod_{\substack{j=0\\ j\neq j'}}^{n}\norbra{\idO -i\alpha \Delta t\,V\norbra{j\Delta t}}\gamma_B}=
	\\
	&=-i\alpha \lim_{n\rightarrow\infty}\sum_{j'=0}^n\; \Delta t\;\TrR{B}{B_k(t)\prod_{\substack{j=j'+1}}^{n} e^{-i\alpha \Delta t V\norbra{j\Delta t}}V\norbra{j'\Delta t}\prod_{\substack{j=0}}^{j'} e^{-i\alpha \Delta t V\norbra{j\Delta t}}\gamma_B}=
	\\
	&=-i\alpha \int_0^t\de s\; \TrR{B}{B_k(t)\mathcal{V}(t,s)V(s)\mathcal{V}(s,0)\gamma_B}\,,
\end{align}
where $\Delta t := t/n$, all the product are ordered from right to left, and we used $e^{-i\alpha \Delta t V\norbra{j\Delta t}} \simeq \idO -i\alpha \Delta t V\norbra{j\Delta t}$. In the third line we used condition~\ref{it:noEnergyShift} to get rid of the contribution arising from the identity, and used Eq.~\eqref{eq:WicksTheorem} to put back the two traces together. 
Then, we can apply Eq.~\eqref{eq:wicksTheoremConsequence} to the exact evolution in Eq.~\eqref{eq:diffEqOfMotionRed} to obtain:
\begin{align}
	\partial_t\rho_{S}(t) &= -i\alpha\,\TrR{B}{[V(t),\rho_{SB}(t)]}= -i\alpha\,\TrR{B}{[V(t),\mathcal{V}(t,0)\rho_S(0)\otimes\gamma_B\mathcal{V}^\dagger(t,0)]} = 
	\\
	&= -\alpha^2\int_{0}^{t}\de s\;\TrR{B}{\sqrbra{V(t),\norbra{\mathcal{V}(t,s) \sqrbra{V(s), \rho_{SB}(s)}\mathcal{V}^\dagger(t,s)}}}\,.\label{eq:exactReducedDynamicsapp}
\end{align}
Finally, expanding the interaction term in terms of the jump operators (see Eq.~\eqref{eq:interactionDec}) and applying Wick's theorem once more we obtain:
\begin{align}
	\partial_t\rho_{S}(t) =\alpha^2\int_{0}^{t}\de s\;\big (C_{kl}(t-s)&\,\TrR{B}{\sqrbra{\mathcal{V}(t,s) A_l(s) \rho_{SB}(s)\mathcal{V}^\dagger(t,s), A_k(t)}}+\nonumber
	\\
	&+C_{kl}(s-t)\,\TrR{B}{\sqrbra{ A_l(t),\mathcal{V}(t,s) \rho_{SB}(s) A_k(s)\mathcal{V}^\dagger(t,s)}}\big)\,,\label{eq:app:B8}
\end{align}
proving the claim. 

Even if this derivation heavily relies on the Gaussianity of the bath, there are other, more general, ways of obtaining an exact reduced dynamics akin to Eq.~\eqref{eq:app:B8}. In fact,  the key role of the bath correlation function in the characterization of the system dynamics does not depend on this assumption \cite{rivas2012open,nathan2020universal}. Nonetheless, assuming Gaussianity  of the bath allows for more straightforward error bounds on the Born-Markov approximation, and it is thus preferable for analytical calculations.

\section{Positivity of the bath spectral function}\label{app:bath}	
We aim to prove that the bath spectral function $\widehat{\textbf{C}}(\omega)$ is a positive matrix or, equivalently, that for any vector $\ket{\textbf{v}}$ we have:
\begin{align}
	\sandwich{{\textbf{v}}}{\widehat{\textbf{C}}(\omega)}{{\textbf{v}}}\geq 0\,.\label{eq:app:c1}
\end{align}
We follow the proof presented in~\cite{rivas2012open}. To this end, let us rewrite Eq.~\eqref{eq:app:c1} as:
\begin{align}
	\sandwich{{\textbf{v}}}{\widehat{\textbf{C}}(\omega)}{{\textbf{v}}} =  \int_{-\infty}^{\infty}\de t\; e^{it\omega} \,\norbra{\overline{v_k}\,\TrR{B}{B_k(t) B_l\gamma_B}\, v_l} = \int_{-\infty}^{\infty}\de t\; e^{it\omega} \,D(t)\geq 0\,,\label{eq:app:c2}
\end{align}
where we have implicitly defined the function $D(t)$. This rewriting is particularly useful thanks to the following:
\begin{lemma}[Bochner's theorem~\cite{reed1975ii}]
	The Fourier transform of a continuous function $f:\RR\rightarrow\mathbb{C}$ satisfies $\hat{f}(\omega)\geq0$ if and only if $f$ is positive definite.
\end{lemma}
Then, this reduces the problem to proving that $D(t)$ is positive definite. These are functions $h:\RR\rightarrow\mathbb{C}$ such that for any vector of reals $\{t_i\}_{i\in\{1,\dots,n\}}$ of arbitrary size $n$, the matrix $\textbf{A}$ defined in coordinates as $A_{ij} = h(t_i-t_j) $ is positive semi-definite. Before movingg to the actual proof, it is useful to rewrite $D(t)$ as:
\begin{align}
	D(t) &= \TrR{B}{e^{i H_B t}(v_k\,B_k)^\dagger e^{-i H_B t} (v_l\,B_l)\gamma_B} =
	\\
	&= \int_{-\infty}^\infty \de E_B\int_{-\infty}^\infty \de E'_B\; \sandwich{E_B}{e^{i H_B t}(v_k\,B_k)^\dagger e^{-i H_B t}\ketbra{E_B'}{E_B'} (v_l\,B_l)\gamma_B}{E_B} =
	\\
	&= \int_{-\infty}^\infty \de E_B\int_{-\infty}^\infty \de E'_B\; \frac{e^{-\beta E_B}}{\mathcal{Z}_B} e^{i (E_B-E_B')t}\,|\sandwich{E_B'}{ (v_l\,B_l)}{E_B}|^2\,,\label{eq:app:c5}
\end{align}
where we used the spectral measure $H_B := \int_{-\infty}^{\infty}\de E_B \;(E_B \ketbra{E_B}{E_B})$. Then, it remains to prove that for any vector $\ket{\textbf{w}}$ and any choice of reals $\{t_i\}_{i\in\{1,\dots,n\}}$ , the quantity $(\overline{w_k}\,D(t_k-t_l)w_l)$ is non-negative. Starting from Eq.~\eqref{eq:app:c5} this can be straightforwardly verified as:
\begin{align}
	(\overline{w_k}\,D(t_k-t_l)w_l) &= \int_{-\infty}^\infty \de E_B\int_{-\infty}^\infty \de E'_B\; \frac{e^{-\beta E_B}}{\mathcal{Z}_B} \norbra{\overline{w_k}\,e^{i (E_B-E_B')t_k}}\norbra{{w_l}\,e^{-i (E_B-E_B')t_l}}\,|\sandwich{E_B'}{ (v_l\,B_l)}{E_B}|^2=
	\\
	&=\int_{-\infty}^\infty \de E_B\int_{-\infty}^\infty \de E'_B\; \frac{e^{-\beta E_B}}{\mathcal{Z}_B} \left | \norbra{{w_m}\,e^{-i (E_B-E_B')t_m}}\right |^2\,|\sandwich{E_B'}{ (v_l\,B_l)}{E_B}|^2\geq 0\,,
\end{align}
where we remind the reader we always sum over repeated indexes. The claim follows from the fact that we are integrating over non-negative quantities. This proves that $D(t)$ is a positive definite function. Thus, it follows from Bochner's theorem and Eq.~\eqref{eq:app:c2} that $\widehat{\textbf{C}}(\omega)\geq0$. 

Another interesting property of positive definite functions is that  they are uniformly bounded by their value in zero. Thus, as a by-product we also obtain the following inequality for any vector $\ket{\textbf{v}}$:
\begin{align}
	\left | \norbra{\overline{v_k}\,C_{kl}(t)\, v_l}\right |\leq  \norbra{\overline{v_k}\,C_{kl}(0)\, v_l}\,.
\end{align}

\section{Relation between different timescales}\label{app:relationTimescales}
We are going to explore the relation of the timescales:
\ms
{
	\begin{align}
		\Gamma_0&:= \sum_{kl}\int_{-\infty}^{\infty}\dt\; |C_{kl}(t)|\,; \qquad \qquad\qquad\qquad \Gamma:= \sum_{kl}\int_{-\infty}^{\infty}\dt\; |g_{k\lambda}(t)|\int_{-\infty}^{\infty}\de s\; |g_{\lambda l}(s)|\,; \\
		\Gamma_0\tau_0 &:=  \sum_{kl}\int_{-\infty}^{\infty}\dt\;|t| |C_{kl}(t)|\,; \;\qquad \qquad\qquad\;\Gamma\tau:=\sum_{kl} \,2\int_{-\infty}^{\infty}\dt\; |t||g_{k\lambda}(t)|\int_{-\infty}^{\infty}\de s\; |g_{\lambda l}(s)|\,;\\
		K_0 &:=  \sum_{kl}\int_{-\infty}^{\infty}\dt\;|t|^2 |C_{kl}(t)|\,; \;\qquad \qquad\qquad\;K:=\sum_{kl} \,\int_{-\infty}^{\infty}\dt\int_{-\infty}^{\infty}\de s\; (|t|+|s|)^2|g_{k\lambda}(t)| |g_{\lambda l}(s)|\,,
	\end{align}
}
between each other and with the bath spectral function. First, it should be noticed that:
\ms
{
	\begin{align}
		&|\widehat{C}_{kl}(\omega)| = \left|\int_{-\infty}^{\infty}\de t\; e^{it\omega}\,{C}_{kl}(t)\right|\leq \sum_{k,l}\int_{-\infty}^{\infty}\de t\; \left|e^{it\omega}\,{C}_{kl}(t)\right| = \Gamma_0\,;\\
		&|\widehat{C}'_{kl}(\omega)|= \left|i\int_{-\infty}^{\infty}\de t\; e^{it\omega}\,t\,{C}_{kl}(t)\right|\leq \sum_{k,l}\int_{-\infty}^{\infty}\de t\; \left|e^{it\omega}\,t \,{C}_{kl}(t)\right|  = \Gamma_0\tau_0\,;\\
		&|\widehat{C}''_{kl}(\omega)|= \left|-\int_{-\infty}^{\infty}\de t^2\; e^{it\omega}\,t\,{C}_{kl}(t)\right|\leq \sum_{k,l}\int_{-\infty}^{\infty}\de t\; \left|e^{it\omega}\,t^2 \,{C}_{kl}(t)\right|  = K_0\,;
	\end{align}
}
In the same way, it can also be proven that $|\widehat{g}_{k\lambda}(\omega)||\widehat{g}_{\lambda l}(\tilde\omega)|\leq \Gamma$ (for any $\omega,\,\tilde{\omega}$), that $|\widehat{g}'_{k\lambda}(\omega)||\widehat{g}_{\lambda l}(\tilde\omega)|\leq \Gamma\tau$, 
\ms
{
	and finally that $\norbra{|\widehat{g}''_{k\lambda}(\omega)||\widehat{g}_{\lambda l}(\tilde\omega)|+|\widehat{g}'_{k\lambda}(\omega)||\widehat{g}'_{\lambda l}(\tilde\omega)|+|\widehat{g}_{k\lambda}(\omega)||\widehat{g}''_{\lambda l}(\tilde\omega)|}\leq K$.
}

Finally, it is also possible to relate $\Gamma_0$ with $\Gamma$ as:
\begin{align}
	\Gamma_0 = \sum_{kl}\int_{-\infty}^{\infty}\dt\; \left |\int_{-\infty}^{\infty}\de s\; g_{k\lambda}(s)g_{\lambda l}(t-s) \right|\leq 	\sum_{kl}\int_{-\infty}^{\infty}\dt\; \int_{-\infty}^{\infty}\de s\; \left |g_{k\lambda}(s)||g_{\lambda l}(t-s) \right|\leq \Gamma\,,
\end{align}
where in the first step we used Eq.~\eqref{eq:convTheorem} and the last inequality follows from a change of variables $t-s\rightarrow t$, which decouples the two integrands. Similarly, we have that:
\begin{align}
	\Gamma_0\tau_0 &=  \sum_{kl}\int_{-\infty}^{\infty}\dt\;|t| \left |\int_{-\infty}^{\infty}\de s\; g_{k\lambda}(s)g_{\lambda l}(t-s) \right|\leq \sum_{kl}\int_{-\infty}^{\infty}\dt\int_{-\infty}^{\infty}\de s\;|t|\left | g_{k\lambda}(s)||g_{\lambda l}(t-s) \right| \leq
	\\
	&\leq \sum_{kl}\int_{-\infty}^{\infty}\dt\int_{-\infty}^{\infty}\de s\;\norbra{|t|+|s|}\left | g_{k\lambda}(s)||g_{\lambda l}(t) \right| = \Gamma\tau\,,\label{eq:D9x}
\end{align}
where, once again, we used the change of variables $t-s\rightarrow t$, and the last step is just an application of the triangle inequality. 
\ms
{Finally, we also have:
	\begin{align}
		K_0 &:=  \sum_{kl}\int_{-\infty}^{\infty}\dt\;|t|^2 \left |\int_{-\infty}^{\infty}\de s\; g_{k\lambda}(s)g_{\lambda l}(t-s) \right| \leq \sum_{kl}\int_{-\infty}^{\infty}\dt\int_{-\infty}^{\infty}\de s\;|t|^2|g_{k\lambda}(s)||g_{\lambda l}(t-s)|\leq
		\\
		&\leq \sum_{kl}\int_{-\infty}^{\infty}\dt\int_{-\infty}^{\infty}\de s\;(|t|+|s|)^2|g_{k\lambda}(s)||g_{\lambda l}(s)| = K\,,
	\end{align}
	where, similarly with Eq.~\eqref{eq:D9x}, we changed variables $t-s\rightarrow t$, and applied the triangle inequality.
}

\section{Proof of Thm.~\ref{thm:fastestRate}: Fastest rate of change in the system}\label{app:fastestRate}
We begin from the expression in Eq.~\eqref{eq:exactReducedDynamics} for the derivative of the state:
\begin{align}
	\partial_t\rho_{S}(t) =\alpha^2\int_{0}^{t}\de s\;\big (C_{kl}(t-s)&\,\TrR{B}{\sqrbra{\mathcal{V}(t,s) A_l(s) \rho_{SB}(s)\mathcal{V}^\dagger(t,s), A_k(t)}}+\nonumber
	\\
	&+C_{kl}(s-t)\,\TrR{B}{\sqrbra{ A_l(t),\mathcal{V}(t,s) \rho_{SB}(s) A_k(s)\mathcal{V}^\dagger(t,s)}}\big)
\end{align}
We now bound the trace norm of the trace in the first line as follows:
\begin{align}
	\left\|\sqrbra{\TrR{B}{\mathcal{V}(t,s) A_l(s) \rho_{SB}(s)\mathcal{V}^\dagger(t,s)}, A_k(t)}\right\|_1 &\leq	2\,\left\|\TrR{B}{\mathcal{V}(t,s) A_l(s) \rho_{SB}(s)\mathcal{V}^\dagger(t,s)}\right\|_1 \|A_k(t)\|_\infty\leq
	\\
	&\leq 2
	\sup_{-\idO_S\leq P_S \leq \idO_S }\left | \Tr{P_S\otimes\idO_B\mathcal{V}(t,s) A_l(s) \rho_{SB}(s)\mathcal{V}^\dagger(t,s)}\right |\,,\label{eq:app:e3}
\end{align}
where in the first step we have taken $A_k(t)$ out of the trace (since it only acts on the system), then we applied the triangle and Hölder's inequality, and finally used the normalization $\|A_k(t)\|_\infty = 1$ and the standard rewriting of the trace distance in terms of a maximization over Hermitian operators whose larger eigenvalue is smaller than 1.  Then, using the cyclicity of the norm, we have that:
\begin{align}
	{\text{Eq.~\eqref{eq:app:e3}}}&=2
	\sup_{-\idO_S\leq P_S \leq \idO_S}\left | \Tr{\norbra{\mathcal{V}^\dagger(t,s)P\otimes\idO_B\mathcal{V}(t,s)} A_l(s) \rho_{SB}(s)}\right |\leq
	\\
	&\leq2
	\sup_{-\idO_{SB}\leq P_{SB} \leq \idO_{SB} }\left | \Tr{P_{SB} A_l(s) \rho_{SB}(s)}\right | =2\, \|A_l(s) \rho_{SB}(s)\|_1\leq 2\,\|A_l(s)\|_\infty \|\rho_{SB}(s)\|_1 \leq 2
\end{align}
where, once again, the last step follows from Hölder's inequality. The second trace can be bounded in the same manner. Putting everything together we have:
\begin{align}
	\|\partial_t\rho_{S}(t) \|_1 \leq 2\alpha^2 \,\sum_{k,l}\int_{0}^{t}\de s\; \norbra{|C_{kl}(t-s)| + |C_{kl}(s-t)|} \leq 2\alpha^2 \,\sum_{k,l}\int_{-\infty}^\infty\de s\;|C_{kl}(s)|\leq 2\alpha^2\Gamma_0\,,
\end{align}
proving the claim.

\section{Bounding the different error terms}\label{app:errorTerms}

\subsection{Proof of Thm.~\ref{thm:bornApproximation}: Born approximation}\label{app:born}
The difference between Eq.~\eqref{eq:exactReducedRescaled} and Eq.~\eqref{eq:bornDynamics} is that we have substituted $\mathcal{V}\norbra{\frac{\sigma}{\alpha^2},\frac{\sigma}{\alpha^2}-x}$ with the identity operator. In order to compare the two expressions, let us rewrite the action of the unitary $\mathcal{V}_{\sigma,x}$ on an operator $X$ as:
\begin{align}
	\mathcal{V}_{\sigma,x}\,X\,\mathcal{V}^\dagger_{\sigma,x} = X- i \alpha \int_{-x}^{0}\de w \; \sqrbra{V\norbra{\frac{\sigma}{\alpha^2}+w}, \,\mathcal{V}\norbra{\frac{\sigma}{\alpha^2}+w,\frac{\sigma}{\alpha^2}-x} \,X\, \mathcal{V}^\dagger\norbra{\frac{\sigma}{\alpha^2}+w,\frac{\sigma}{\alpha^2}-x}} \,.\label{eq:app:f1}
\end{align} 
Then, the Born approximation consists of disregarding the first term in the expansion. This allows us to express the error term $\mathcal{E}^B_t= \partial_t\norbra{\tilde{\rho}_{S}(t)-\tilde{\rho}^B_{S}(t)}$ as:
\begin{align}
	\mathcal{E}^B_t:&=\partial_t\int_0^{t}\de \sigma \int_0^{\frac{\sigma}{\alpha^2}}\de x\int_{-x}^{0}\de w\; H^B_{\sigma,x,w} = \int_0^{\frac{t}{\alpha^2}}\de x\int_{-x}^{0}\de w\; H^B_{t,x,w}\,,\label{eq:app:f3}
\end{align}
where, for brevity of notation, we implicitly defined the operator $H^B_{\sigma,x,w}$ which corresponds to substituting in Eq.~\eqref{eq:exactReducedRescaled} every term of the form $\mathcal{V}_{\sigma,x}\cdot\mathcal{V}_{\sigma,x}^\dagger$ with the second term on the right-hand side of Eq.~\eqref{eq:app:f1}. Then, we can proceed to first bound the trace norm of $H^B_{t,x,w}$, which explicitly reads:
\begin{align}
	&\|H^B_{t,x,w}\|_1 \leq \alpha\norbra{|C_{kl}(x)|+|C_{kl}(-x)|} \cdot\nonumber\\
	&\qquad\qquad\qquad\;\;\cdot\left \|\TrR{B}{\sqrbra{\sqrbra{V\norbra{\frac{t}{\alpha^2}+w}, \,\mathcal{V}^{\frac{t}{\alpha^2}+w}_{\frac{t}{\alpha^2}-x}\, A_l\norbra{\frac{t}{\alpha^2}-x} \tilde{\rho}_{SB}\norbra{t-\alpha^2x}\norbra{\mathcal{V}^{\frac{t}{\alpha^2}+w}_{\frac{t}{\alpha^2}-x}}^\dagger}, A_k\norbra{\frac{t}{\alpha^2}}}}\right \|_1\,,\label{eq:app:f5}
\end{align}
where, to make the expression above slightly more compact, we introduced the notation $\mathcal{V}^t_s := \mathcal{V}(t,s)$. Moreover, we denote the trace norm in Eq.~\eqref{eq:app:f5} by $\|\hat{H}^B_{t,x,w}\|_1$. Then, by applying Hölder's inequality and using the Gaussianity of the bath we obtain:
\begin{align}
	\|\hat{H}^B_{t,x,w}\|_1&\leq2\,\|A_k\|_\infty\;\left \|\TrR{B}{\sqrbra{V\norbra{\frac{t}{\alpha^2}+w}, \,\mathcal{V}^{\frac{t}{\alpha^2}+w}_{\frac{t}{\alpha^2}-x}\, A_l\norbra{\frac{t}{\alpha^2}-x} \tilde{\rho}_{SB}\norbra{t-\alpha^2x}\norbra{\mathcal{V}^{\frac{t}{\alpha^2}+w}_{\frac{t}{\alpha^2}-x}}^\dagger}}\right \|_1\leq4\alpha\Gamma_0\,,
\end{align}
where in the last step we used an expansion akin to Eq.~\eqref{eq:exactReducedDynamicsapp} (which gives the extra $\alpha$ factor associated to the contribution coming from having $V$ in the expression), and then applied the same steps presented in App.~\ref{app:fastestRate} to bound the trace norm with $2\Gamma_0$. 

Now plugging in Eq.~\eqref{eq:app:f3} the bound from Eq.~\eqref{eq:app:f5} we get:
\begin{align}
	\|\mathcal{E}^B_t\|_1&\leq \int_0^{\frac{t}{\alpha^2}}\de x\int_{-x}^{0}\de w\; \Norm{H^B_{t,x,w}} \leq 4\alpha^2\Gamma_0 \sum_{kl}\int_{0}^{\frac{t}{\alpha^2}}\de x\; |x|\norbra{|C_{kl}(x)|+|C_{kl}(-x)|} \leq 4\alpha^2 \Gamma_0 (\Gamma_0\tau_0)\,,
\end{align}
where in the last step we used the definition in Eq.~\eqref{eq:TauDef}. This proves the claim.
\subsection{Proof of Thm.~\ref{thm:markovApproximation}: Markov approximation}\label{app:markov}
Once again, we can express $\mathcal{E}^{BM}_t= \partial_t\norbra{\tilde{\rho}^{B}_{S}(t)-\tilde{\rho}^{BM}_{S}(t)}$ as in  Eq.~\eqref{eq:app:f3}, that is:
\begin{align}
	\Norm{\mathcal{E}^{BM}_t}_1\leq\int_0^{\frac{t}{\alpha^2}}\de x\; \Norm{H^{BM}_{t,x}}_1\,,
\end{align}
where the trace norm on the right-hand side can be upper bounded as:
\begin{align}
	\Norm{H^{BM}_{t,x}}_1&\leq \norbra{|C_{kl}(x)|+|C_{kl}(-x)|}\Norm{\sqrbra{A_l\norbra{\frac{t}{\alpha^2}-x} \norbra{\tilde{\rho}_{S}\norbra{t-\alpha^2x}-\tilde{\rho}_{S}\norbra{t}}, A_k\norbra{\frac{t}{\alpha^2}}}}_1\leq
	\\
	&\leq 2\norbra{|C_{kl}(s)|+|C_{kl}(-s)|}\Norm{ \tilde{\rho}_{S}\norbra{t-\alpha^2s}-\tilde{\rho}_{S}\norbra{t}}_1\,,
\end{align}
where we implicitly used the triangle inequality and the fact that $\Norm{A_k}_\infty = 1$.
Then, thanks to Thm.~\ref{thm:fastestRate}, it holds that: 
\begin{align}
	\Norm{ \tilde{\rho}_{S}\norbra{t-\alpha^2x}-\tilde{\rho}_{S}\norbra{t}}_1 = \left\|\int_{t}^{t-\alpha^2 x}\de u\; \partial_u \tilde{\rho}_{S}\norbra{u}\right\|_1\leq \frac{\alpha^2|x|}{2} \sup_{u\in [t, t-\alpha^2x]}\|\partial_u \tilde{\rho}_{S}\norbra{u}\|_1\leq \alpha^2|x|\Gamma_0\,.
\end{align}
Then, wrapping everything together we finally obtain:
\begin{align}
	\Norm{\mathcal{E}^{BM}_t}_1\leq2\alpha^2\Gamma_0\sum_{kl}\int_{0}^{\frac{t}{\alpha^2}}\de x\;|x|\norbra{|C_{kl}(x)|+|C_{kl}(-x)|} \leq 2 \alpha^2 \Gamma_0(\Gamma_0\tau)
\end{align}
proving the claim.

\subsection{Proof of Thm.~\ref{thm:redfield}: Redfield approximation}\label{app:redfield}
The error term $\mathcal{E}^{RE}_t= \partial_t\norbra{\tilde{\rho}^{BM}_{S}(t)-\tilde{\rho}^{RE}_{S}(t)}$ can be bounded through the same steps as in previous appendices as:
\begin{align}
	\|\mathcal{E}^{RE}_t\|_1&\leq2\,\sum_{kl}\int_{0}^{\infty}\de x\;|\theta\norbra{(t/\alpha^2)-x} -e^{-\frac{(x/2)^2}{T(\alpha)^2}}| \norbra{|C_{kl}(x)|+|C_{kl}(-x)|}\,,
\end{align}
where we denote by $\theta(x)$ the Heaviside step function. We can then split the integral in two parts, depending on the behavior of the $\theta$-function. The first term is bounded as:
\begin{align}
	2\,\sum_{kl}\int_{0}^{\frac{t}{\alpha^2}}\de x\;&\norbra{\frac{1 -e^{-\frac{(x/2)^2}{T(\alpha)^2}}}{|x|+\Gamma_0^{-1}}} (|x|+\Gamma^{-1}_0) \norbra{|C_{kl}(x)|+|C_{kl}(-x)|}\leq 2 (1+\Gamma_0\tau_0)\sup_{x\in \RR}\norbra{\frac{1 -e^{-\frac{(x/2)^2}{T(\alpha)^2}}}{|x|+\Gamma_0^{-1}}}\leq
	\\
	&\leq \frac{(1+\Gamma_0\tau_0)}{{\frac{T(\alpha)}{2}+\Gamma_0^{-1}}} = \frac{2\,\Gamma_0(1+\Gamma_0\tau_0)}{{2+\Gamma_0\,T(\alpha)}}\,,
\end{align}
where the first inequality follows from extending the limit of integration to infinity, and using Hölder's inequality. On the other hand, the second part can be bounded in a similar manner:
\begin{align}
	2\,\sum_{kl}\int_{\frac{t}{\alpha^2}}^{\infty}\de x\;&e^{-\frac{(x/2)^2}{T(\alpha)^2}}\norbra{\frac{|x| +\Gamma_0^{-1}}{|x|+\Gamma_0^{-1}}}  \norbra{|C_{kl}(x)|+|C_{kl}(-x)|}\leq \frac{2\,\alpha^2\,\Gamma_0(1+\Gamma_0\tau_0)}{(\alpha^2+\Gamma_0 \,t)}\,,
\end{align}
where we implicitly used the fact that $(|x|+\Gamma_0^{-1})^{-1}\leq(t/\alpha^2+\Gamma_0^{-1})^{-1}$ on the domain of interest. 
This is the only error bound that depends on the time $t$ in $\|\mathcal{E}^{RE}_t\|_1$. Upon integration, this gives the logarithmic behavior in Eq.~\eqref{eq:redfieldApproximation}.

\subsection{Proof of Thm.~\ref{thm:smoothing} and Thm.~\ref{thm:coarseGrain}: Time averaging procedures}\label{app:cg}
We begin by bounding the trace difference between the smoothed and the exact state:
\begin{align}
	\Norm{\tilde{\rho}_S(t)-\tilde{\rho}^{S}_S(t)}_1 &\leq \int_{-\infty}^{\infty}\de q\; \frac{e^{-\frac{q^2}{T(\alpha)^2}}}{\sqrt{\pi} \,T(\alpha)}\;\Norm{\tilde{\rho}_S(t)-\tilde{\rho}_S(t+\alpha^2q)}_1\leq \int_{-\infty}^{\infty}\de q\; \frac{e^{-\frac{q^2}{T(\alpha)^2}}}{\sqrt{\pi} \,T(\alpha)}\;\int_{0}^{\alpha^2|q|}\de x\;\Norm{\partial_x\tilde{\rho}_S(t+x)}_1\leq \nonumber
	\\
	&\leq2\Gamma_0\alpha^2\int_{-\infty}^{\infty}\de q\; \frac{e^{-\frac{q^2}{T(\alpha)^2}}}{\sqrt{\pi} \,T(\alpha)}|q| = \frac{2\Gamma_0\alpha^2 T(\alpha)}{\sqrt{\pi}}\,,\label{eq:app:f16}
\end{align}
where we applied Thm.~\ref{thm:fastestRate}. This proves Thm.~\ref{thm:smoothing}.

Then, we proceed to prove the differential expression in Eq.~\eqref{eq:diffS}, which can be deduced by direct calculation:
\begin{align}
	\partial_t\tilde{\rho}^{S}_S(t) :&= \int_{-\infty}^{\infty}\de q\; \frac{e^{-\frac{q^2}{T(\alpha)^2}}}{\sqrt{\pi} \,T(\alpha)}\;\partial_t\tilde{\rho}_S(t+\alpha^2q) =
	\\
	&= \int_{-\infty}^{\infty}\de q\; \frac{e^{-\frac{q^2}{T(\alpha)^2}}}{\sqrt{\pi} \,T(\alpha)}\;\norbra{\norbra{\mathcal{E}^B_{t+\alpha^2q}+\mathcal{E}^{BM}_{t+\alpha^2q}+\mathcal{E}^{RE}_{t+\alpha^2q}} + \partial_t\tilde{\rho}^{RE}_{S}(t+\alpha^2q)} = 
	\\
	&=\mathcal{E}^{D}_{t}+\int_{-\infty}^{\infty}\de q\; \frac{e^{-\frac{q^2}{T(\alpha)^2}}}{\sqrt{\pi} \,T(\alpha)}\;\norbra{\lind^{RE}_{{t+\alpha^2q}}[\tilde{\rho}_{S}(t+\alpha^2q)]}\,, \label{eq:app:f17}
\end{align}
where we implicitly defined the error term $\mathcal{E}^{D}_{t}$. We further rewrite the integral as:
\begin{align}
	\int_{-\infty}^{\infty}\de q\; \frac{e^{-\frac{q^2}{T(\alpha)^2}}}{\sqrt{\pi} \,T(\alpha)}\;\norbra{ \lind^{RE}_{{t+\alpha^2q}}[\tilde{\rho}_{S}(t+\alpha^2q)-\tilde{\rho}^S_{S}(t)]+ \lind^{RE}_{{t+\alpha^2q}}[\tilde{\rho}^S_{S}(t)]} =\mathcal{E}^{SM}_{t}+ \lind^{CG0}_{{t}}[\tilde{\rho}^S_{S}(t)],
\end{align}
where we introduced $\mathcal{E}^{SM}_{t}$ and we used the definition in Eq.~\eqref{eq:lindcg0}. Finally, introducing the notation  $\mathcal{E}^{CP}_{t}:=\lind^{CG0}_{{t}}[\tilde{\rho}^S_{S}(t)]-\lind^{CG}_{{t}}[\tilde{\rho}^S_{S}(t)]$, we can rewrite Eq.~\eqref{eq:app:f17} as:
\begin{align}
	\partial_t\tilde{\rho}^{S}_S(t) = \mathcal{E}^{D}_{t}+\mathcal{E}^{SM}_{t}+ \lind^{CG0}_{{t}}[\tilde{\rho}^S_{S}(t)] = \mathcal{E}^{D}_{t}+\mathcal{E}^{SM}_{t}+\mathcal{E}^{CP}_{t}+\lind^{CG}_{{t}}[\tilde{\rho}^S_{S}(t)]
\end{align}
This proves Eq.~\eqref{eq:diffS} with $\mathcal{E}^{TOT}_t $ given by the sum of the three error terms above.

Let us estimate the trace norm of the error terms one by one. First, a bound on $\mathcal{E}^{D}_{t}$ directly follows from the estimates in  Thm.~\ref{thm:bornApproximation},~\ref{thm:markovApproximation} and~\ref{thm:redfield}:
\begin{align}
	\Norm{\mathcal{E}^{D}_{t}}_1&\leq \int_{-\infty}^{\infty}\de q\; \frac{e^{-\frac{q^2}{T(\alpha)^2}}}{\sqrt{\pi} \,T(\alpha)}\; \norbra{\Norm{\mathcal{E}^B_{t+\alpha^2q}}_1+\Norm{\mathcal{E}^{BM}_{t+\alpha^2q}}_1+\Norm{\mathcal{E}^{RE}_{t+\alpha^2q}}_1}|\leq
	\\
	&\leq \int_{-\infty}^{\infty}\de q\; \frac{e^{-\frac{q^2}{T(\alpha)^2}}}{\sqrt{\pi} \,T(\alpha)}\;\norbra{6\alpha^2 \Gamma_0 (\Gamma_0\tau_0)+2\,\Gamma_0(1+\Gamma_0\tau_0) \norbra{\frac{1}{{2+\Gamma_0\,T(\alpha)}}+\frac{\alpha^2}{(\alpha^2+\Gamma_0 \,t)}}}\,.
\end{align}
In the second inequality we used the bound $\frac{\alpha^2}{(\alpha^2+\Gamma_0 \,(t+\alpha^2|q|))}\leq\frac{\alpha^2}{(\alpha^2+\Gamma_0 \,t)}$ to get rid of the $q$ dependence of  $\Norm{\mathcal{E}^{RE}_{t+\alpha^2q}}_1$. 
Then, integrating over $q$ gives a factor $1$, since the Gaussian is normalized. 

We can now estimate $\mathcal{E}^{SM}_t$, which we rewrite as:
\begin{align}
	\Norm{\mathcal{E}^{SM}_t}_1\leq\int_{-\infty}^{\infty}\de q\; \frac{e^{-\frac{q^2}{T(\alpha)^2}}}{\sqrt{\pi} \,T(\alpha)}\;& \norbra{\Norm{\lind^{RE}_{{t+\alpha^2q}}[(\tilde{\rho}_{S}(t+\alpha^2q)-\tilde{\rho}_{S}(t))+(\tilde{\rho}_{S}(t)-\tilde{\rho}^{S}_{S}(t))]}_1}\leq 
	\\
	&\leq\Lambda\int_{-\infty}^{\infty}\de q\; \frac{e^{-\frac{q^2}{T(\alpha)^2}}}{\sqrt{\pi} \,T(\alpha)}\norbra{\Norm{\tilde{\rho}_{S}(t+\alpha^2q)-\tilde{\rho}_{S}(t)}_1+\Norm{\tilde{\rho}_{S}(t)-\tilde{\rho}^{S}_{S}(t)}_1}\;\label{app:eq:f25}
\end{align}
where we introduced the constant $\Lambda := \sup_{t} \,\|\lind^{RE}_{t}\|_{1-1}$, which satisfies the bound:
\begin{align}
	\Lambda &\leq \sup_{\Norm{X}_1=1} \int_0^{\infty}\de x\;e^{-\frac{x}{T(\alpha)^2}}\Norm{\mathcal{U}_{\frac{t}{\alpha^2}}^\dagger\sqrbra{\norbra{C_{kl}\norbra{x}\,\sqrbra{\, A_l\norbra{-x} \mathcal{U}_{\frac{t}{\alpha^2}}[X], A_k}+C_{kl}(-x)\,\sqrbra{ A_l, \mathcal{U}_{\frac{t}{\alpha^2}}[X] A_k\norbra{-x}}}}}_1\leq
	\\
	&\leq\sup_{\Norm{X}_1=1} \int_0^{\infty}\de x\;(|C_{kl}(x)|+|C_{kl}(-x)|)\Norm{\sqrbra{\, A_l\norbra{-x} X, A_k}}_1 \leq  2\,\Gamma_0\,.
\end{align}
Moreover, Thm.~\ref{thm:smoothing} implies that the second norm in Eq.~\eqref{app:eq:f25} can be bounded as:
\begin{align}
	\norbra{\int_{-\infty}^{\infty}\de q\; \frac{e^{-\frac{q^2}{T(\alpha)^2}}}{\sqrt{\pi} \,T(\alpha)} }\Norm{\tilde{\rho}_{S}(t)-\tilde{\rho}^{S}_{S}(t)}_1\leq\frac{2\Gamma_0\alpha^2 T(\alpha)}{\sqrt{\pi}}\,.
\end{align}
Similarly the first term is bounded by:
\begin{align}
	\int_{-\infty}^{\infty}\de q\; \frac{e^{-\frac{q^2}{T(\alpha)^2}}}{\sqrt{\pi} \,T(\alpha)}\;&\Norm{\tilde{\rho}_{S}(t+\alpha^2q)-\tilde{\rho}_{S}(t)}_1 \leq \int_{-\infty}^{\infty}\de q\; \frac{e^{-\frac{q^2}{T(\alpha)^2}}}{\sqrt{\pi} \,T(\alpha)}\int_{0}^{\alpha^2 |q|}\de x\;\Norm{\partial_x\tilde{\rho}_{S}(t+x)}_1\leq
	\\
	&\leq2\alpha^2\Gamma_0\int_{-\infty}^{\infty}\de q\; \frac{e^{-\frac{q^2}{T(\alpha)^2}}}{\sqrt{\pi} \,T(\alpha)}|q|  = \frac{2\Gamma_0\alpha^2 T(\alpha)}{\sqrt{\pi}}\,,
\end{align}
where we used Thm.~\ref{thm:fastestRate} to upper-bound the norm of the derivative. This proves that:
\begin{align}
	\Norm{\mathcal{E}^{SM}_t}_1\leq \frac{4\Gamma_0\alpha^2 T(\alpha)}{\sqrt{\pi}}\,.
\end{align}
Finally, we need to estimate $\mathcal{E}^{CP}_{t}:=\lind^{CG0}_{{t}}[\tilde{\rho}^S_{S}(t)]-\lind^{CG}_{{t}}[\tilde{\rho}^S_{S}(t)]$. To this end, it is useful to rewrite $\lind^{CG0}_{{t}}$ as:
\begin{align}
	\lind^{CG0}_{{t}}[\tilde{\rho}^S_{S}\norbra{t}] := \int_{-\infty}^{\infty}\de q\int_0^{\infty}\de x\;& \frac{e^{-\frac{((q+x/2)^2+(x/2)^2)}{T(\alpha)^2}}}{\sqrt{\pi} \,T(\alpha)}\bigg (C_{kl}\norbra{x}\,\sqrbra{\, A_l\norbra{\frac{t}{\alpha^2}+q-\frac{x}{2}} \tilde{\rho}^S_{S}\norbra{t}, A_k\norbra{\frac{t}{\alpha^2}+q+\frac{x}{2}}}+\nonumber
	\\
	&\qquad\qquad\qquad\;\;\;\;+C_{kl}(-x)\,\sqrbra{ A_l\norbra{\frac{t}{\alpha^2}+q+\frac{x}{2}}, \tilde{\rho}^S_{S}\norbra{t} A_k\norbra{\frac{t}{\alpha^2}+q-\frac{x}{2}}}\bigg)\,,
\end{align}
where we implicitly carried out the substitution $q\rightarrow q+\frac{x}{2}$. Then, $\lind^{CG0}_{{t}}$ differs from $\lind^{CG}_{{t}}$ (see Eq.~\eqref{eq:averagedLindblad}) only for the argument of the exponential. This allow to estimate $\mathcal{E}^{CP}_{t}$ as:
\begin{align}
	\Norm{\mathcal{E}^{CP}_{t}}_1 &\leq 2\int_{-\infty}^{\infty}\de q\int_0^{\infty}\de x\; \frac{e^{-\frac{(x/2)^2}{T(\alpha)^2}}}{\sqrt{\pi} \,T(\alpha)}\,\left |e^{-\frac{(q+x/2)^2}{T(\alpha)^2}}-e^{-\frac{q^2}{T(\alpha)^2}}\right||C_{kl}(x)+C_{kl}(-x)| =
	\\
	&=2\int_{-\infty}^{\infty}\de q\int_0^{\infty}\de x\; \frac{e^{-\frac{(x/2)^2}{T(\alpha)^2}}}{\sqrt{\pi} \,T(\alpha)}\,\left |\int_0^{\frac{x}{2}}\de y \; \frac{e^{-\frac{(q+y/2)^2}{T(\alpha)^2}}}{T(\alpha)^2}\norbra{q+\frac{y}{2}}\right||C_{kl}(x)+C_{kl}(-x)|\leq
	\\
	&\leq2\int_0^{\infty}\de x\; \frac{e^{-\frac{(x/2)^2}{T(\alpha)^2}}}{\sqrt{\pi} \,T(\alpha)}\,\int_0^{\frac{x}{2}}\de y \norbra{\int_{-\infty}^{-\frac{y}{2}}\de q+\int^{\infty}_{-\frac{y}{2}}\de q}\; \frac{e^{-\frac{(q+y/2)^2}{T(\alpha)^2}}}{T(\alpha)^2}\left |q+\frac{y}{2}\right||C_{kl}(x)+C_{kl}(-x)|\,,\label{eq:app:f35}
\end{align}
where the first equality just follows from taking the derivative with respect to $q$ and then integrating again, while in the last step we split the integral in $q$ depending on the sign of $q+\frac{y}{2}$.  Let us examine one of the two halves (the other case is done similarly). Then, we can ignore the absolute value, and carry out explicitly the computation:
\begin{align}
	2\int_0^{\infty}\de x\; &\frac{e^{-\frac{(x/2)^2}{T(\alpha)^2}}}{\sqrt{\pi} \,T(\alpha)}\,|C_{kl}(x)+C_{kl}(-x)|\int_0^{\frac{x}{2}}\de y \int^{\infty}_{-\frac{y}{2}}\de q\; \frac{e^{-\frac{(q+y/2)^2}{T(\alpha)^2}}}{T(\alpha)^2}\norbra{q+\frac{y}{2}} =
	\\
	&=2\int_0^{\infty}\de x\; \frac{e^{-\frac{(x/2)^2}{T(\alpha)^2}}}{\sqrt{\pi} \,T(\alpha)}\,|C_{kl}(x)+C_{kl}(-x)|\int_0^{\frac{x}{2}}\de y \int^{\infty}_{0}\de w\; \frac{e^{-\frac{w^2}{T(\alpha)^2}}}{T(\alpha)^2}\,w =
	\\
	&=\frac{1}{2}\int_0^{\infty}\de x\; \frac{e^{-\frac{(x/2)^2}{T(\alpha)^2}}}{\sqrt{\pi} \,T(\alpha)}\,|C_{kl}(x)+C_{kl}(-x)||x| \leq \frac{\Gamma_0\tau_0}{2\sqrt{\pi}\,T(\alpha)}\,.
\end{align}
The other term in Eq.~\eqref{eq:app:f35} gives the same bound, so we finally obtain the estimate:
\begin{align}
	\Norm{\mathcal{E}^{CP}_{t}}_1\leq \frac{\Gamma_0\tau_0}{\sqrt{\pi}\,T(\alpha)}.
\end{align}
Wrapping everything together, we have the upper-bound:
\begin{align}
	\Norm{\mathcal{E}^{TOT}_t}_1\leq \norbra{6\alpha^2 \Gamma_0 (\Gamma_0\tau_0)+ \norbra{\frac{2\,\Gamma_0(1+\Gamma_0\tau_0)}{{2+\,\Gamma_0\,T(\alpha)}}+\frac{2\alpha^2\,\Gamma_0(1+\Gamma_0\tau_0)}{(\alpha^2+\Gamma_0 \,t)}}} + \frac{4\alpha^2 \Gamma_0(\Gamma_0T(\alpha))}{\sqrt{\pi}}+\frac{\Gamma_0\tau_0}{\sqrt{\pi}\,T(\alpha)}\,.\label{eq:app:31}
\end{align}
Then, we can identify the right hand side of Eq.~\eqref{eq:app:31} with the function $m(t)$ of Thm.~\ref{thm:DBerrorBoundIntegrated}. Moreover, thanks to Eq.~\eqref{eq:app:f16}, the constant $R$ is also known. It directly follows that:
\begin{align}
	&\Norm{	\tilde{\rho}^{S}_S(t)-	\tilde{\rho}^{CG}_S(t)}_1\leq \frac{2\Gamma_0\alpha^2 T(\alpha)}{\sqrt{\pi}} + \int_0^t  m(s) \de s  =
	\\
	&=2\alpha^2\norbra{\norbra{(1+2\Gamma_0t)\frac{\Gamma_0 T(\alpha)}{\sqrt{\pi}}+ \frac{(\Gamma_0t)(1+\Gamma_0\tau_0)}{\alpha^2+\alpha^2\Gamma_0 T(\alpha)}+t(\Gamma_0 \tau_0)\norbra{3\Gamma_0+\frac{\alpha^{-2}}{2\sqrt{\pi}\,T(\alpha)}}}+(1+\Gamma_0\tau_0)\log\norbra{1+\frac{\Gamma_0t}{\alpha^2}}}\,.\label{eq:app:f41}
\end{align}
Applying the triangle inequality, we finally obtain:
\begin{align}
	\Norm{	\tilde{\rho}_S(t)-	\tilde{\rho}^{CG}_S(t)}_1 &\leq \Norm{	\tilde{\rho}_S(t)-	\tilde{\rho}^{S}_S(t)}_1+\Norm{	\tilde{\rho}^{S}_S(t)-	\tilde{\rho}^{CG}_S(t)}_1 \leq\label{eq:app:f34}
	\\
	& \leq K_1 + K_2 t + K_3 \log\norbra{1+\frac{\Gamma_0t}{\alpha^2}}
\end{align}
where the three constants are defined as: 
\begin{align}
	&K_1 := \frac{4\Gamma_0\alpha^2 T(\alpha)}{\sqrt{\pi}} \;,
	\\
	&K_2 := 6\alpha^2 \Gamma_0 (\Gamma_0\tau_0)+\frac{1}{T(\alpha)}\norbra{2+3\,\Gamma_0\tau_0} + \norbra{4\alpha^2 \Gamma_0^2}T(\alpha)\,,
	\\
	&K_3:=2\alpha^2\,\Gamma_0(1+\Gamma_0\tau_0)\,.
\end{align}
It should be noticed that in the definition of $K_2$ we used the inequality $(2+\Gamma_0 T(\alpha))^{-1}\leq(\Gamma_0 T(\alpha))^{-1}$ and $\sqrt{\pi}^{-1}< 1$ to simplify the expression from Eq.~\eqref{eq:app:f41}. Moreover, there is an
additional contribution to $K_1$ which comes from the use of Eq.~\eqref{eq:app:f16} (Thm.~\ref{thm:smoothing}) to estimate the first trace distance in Eq.~\eqref{eq:app:f34}. Then, we can minimize $K_2$ as a function of $T(\alpha)$. The minimum is reached for:
\begin{align}
	K_2^* :&=  
	\alpha\norbra{4\Gamma_0\sqrt{2+3\,\Gamma_0\tau_0}+6\alpha \Gamma_0 (\Gamma_0\tau_0)}
\end{align}
where the optimal observation time takes the form:
\begin{align}
	T^*_{opt}(\alpha) :&=
	\frac{1}{2\alpha \Gamma_0}\sqrt{2+3\,\Gamma_0\tau_0}
\end{align}
This proves Thm.~\ref{thm:coarseGrain}.

\section{Properties of the coarse-grained Lindbladian}\label{app:cgLindProperties}
\subsection{Hermiticity of $H_{LS}$}\label{app:hLShermitian}
Let us start by showing that indeed the Lamb-shift $H_{LS}$ defined in Eq.~\eqref{eq:LSHam} generates a unitary rotation. This can be done by showing its hermiticity through a direct calculation:
\begin{align}
	(H_{LS}^{CG})^\dagger&= -i\int_{-\infty}^{\infty}\de q\; e^{-iq\omega_-}\int_{-\infty}^{\infty}\de x\, e^{ix\omega_+}\mathcal{G}(q,x) \overline{C_{kl}(x)}{\rm sign}(x)A^\dagger_l(\omega)A_k(\tilde\omega) = 
	\\
	&= -i\int_{-\infty}^{\infty}\de q\; e^{iq\omega_-}\int_{-\infty}^{\infty}\de x\, e^{-ix\omega_+}\mathcal{G}(q,x) \overline{C_{lk}(-x)}{\rm sign}(-x)A^\dagger_k(\tilde\omega)A_l(\omega)=
	\\
	&= i\int_{-\infty}^{\infty}\de q\; e^{iq\omega_-}\int_{-\infty}^{\infty}\de x\, e^{-ix\omega_+}\mathcal{G}(q,x)  C_{kl}(x){\rm sign}(x)A^\dagger_k(\tilde\omega)A_l(\omega)= H_{LS}^{CG}\label{eq:93}
\end{align}
where we implicitly integrate over all frequencies, and in the second line we exchanged indexes, frequencies and performed the change of variable $x\rightarrow -x$, while in the last step we used Eq.~\eqref{eq:symmetryBathcorrelation} together with the fact that ${\rm sign}(-x) = -{\rm sign}(x)$. 

\subsection{Decomposition of the rates}\label{app:compRates}
The function ${\widehat{f}}^\alpha(\omega)$ in Eq.~\eqref{eq:ratesFreqDecouples} can be explicitly computed, giving:
\begin{align}
	{\widehat{f}}^\alpha(\omega) = \int_{-\infty}^{\infty}\de t_1\;\norbra{\frac{e^{-\frac{t_1^2}{2T(\alpha)^2}}}{(\sqrt{\pi}\, T(\alpha))^{1/2}}}e^{it_1\omega} = (2 T(\alpha)\sqrt{\pi})^{1/2} e^{-\frac{T(\alpha)^2\omega^2}{2}}\,.
\end{align}
Then, plugging this expression in Eq.~\eqref{eq:ratesFreqDecouples}, we obtain:
\begin{align}
	\gamma_{kl}^{\omega,\tilde{\omega}}=\int_{-\infty}^{\infty}\frac{\de\omega^*}{2\pi}\;{\widehat{C}}_{kl}(\omega^*)&{\widehat{f}}^\alpha(\omega+\omega^*)\overline{{\widehat{f}}^\alpha({\tilde\omega}+\omega^*)}=\frac{T(\alpha)}{\sqrt{\pi}}\int_{-\infty}^{\infty}\de\omega^*\;{\widehat{C}}_{kl}(\omega^*)\,e^{-\frac{T(\alpha)^2\norbra{(\omega+\omega^*)^2+(\tilde{\omega}+\omega^*)^2}}{2}} =
	\\
	&=\frac{e^{-(T(\alpha)\,\omega_-)^2/4}}{\sqrt{\pi}}T(\alpha)\int_{-\infty}^{\infty}\de\omega^*\;{\widehat{C}}_{kl}(\omega^*)\,e^{-T(\alpha)^2\norbra{\omega^*+\omega_+}^2} =
	\\
	&=\frac{e^{-(T(\alpha)\,\omega_-)^2/4}}{\sqrt{\pi}}\int_{-\infty}^{\infty}\de\Omega\;{\widehat{C}}_{kl}\norbra{\frac{\Omega}{T(\alpha)}-\omega_+}\,e^{-\Omega^2}
\end{align}
where in the last step we substituted $\{\Omega = T(\alpha)(\omega^*+\omega_+)\}$. This proves Eq.~\eqref{eq:gammarateCG}.

Similar manipulations can be carried out for $S_{kl}^{\omega,\tilde{\omega}}$. To this end, we will need the following well-known relation:
\begin{align}
	\int_{-\infty}^{\infty} \de x\; e^{ix\omega}\,C_{kl}(x){\rm sign}(x) = \frac{i}{\pi}\int_{-\infty}^{\infty} \de\omega_1 \; \frac{\widehat{C}_{kl}(\omega_1)}{\omega-\omega_1}\,,\label{eq:G8x}
\end{align}
which connects the Fourier transform of a function multiplied by ${\rm sign}(x)$ to the Hilbert transform of the original function. Then, we can recast $S_{kl}^{\omega,\tilde{\omega}}$ as:
\begin{align}
	S_{kl}^{\omega,\tilde{\omega}}&=i\int_{-\infty}^{\infty}\de t_1\int_{-\infty}^{\infty}\de t_2\;{f}^\alpha(t_1){f}^\alpha(t_2)e^{it_1\omega} e^{-it_2{\tilde{\omega}}}C_{kl}(t_2-t_1)  {\rm sign}(t_2-t_1)=
	\\
	&= \int_{-\infty}^{\infty}\frac{\de\omega_1}{2\pi}\int_{-\infty}^{\infty}\frac{\de\omega_2}{\pi}\;\frac{\widehat{C}_{kl}(\omega_2)}{\omega_2-\omega_1}\;{\widehat{f}}^\alpha(\omega+\omega_1)\overline{{\widehat{f}}^\alpha({\tilde\omega}+\omega_1)} =
	\\
	&=\frac{e^{-(T(\alpha)\,\omega_-)^2/4}}{\pi\sqrt{\pi}}T(\alpha)\,\int_{-\infty}^{\infty}\de\omega_1\int_{-\infty}^{\infty}\de\omega_2\;\frac{\widehat{C}_{kl}(\omega_2)}{\omega_2-\omega_1}\;e^{-T(\alpha)^2\norbra{\omega_1+\omega_+}^2}\,.\label{eq:126}
\end{align} 
We can further simplify the above expression by carrying out the change of variables $\{\Omega = T(\alpha)(\omega_1+\omega_+)\}$, which gives:
\begin{align}
	S_{kl}^{\omega,\tilde{\omega}}&=\frac{e^{-(T(\alpha)\,\omega_-)^2/4}}{\sqrt{\pi}}\int_{-\infty}^{\infty}\frac{\de\Omega}{\pi}\int_{-\infty}^{\infty}\de\omega_2\;\frac{\widehat{C}_{kl}(\omega_2)}{(\omega_2+\omega_+)-\frac{\Omega}{T(\alpha)}}\;e^{-\Omega^2} \,.\label{eq:app:g13}
\end{align}
This proves Eq.~\eqref{eq:SrateCG}.


\subsection{Quasilocality of the jump operators}\label{app:quasiLocalJump}

Let us now study the setting where $H$ is a local Hamiltonian on a lattice $\Lambda$, for which the jump operators become quasi-local. We follow closely a similar argument from \cite{rouze2024optimal}. To recall the jump operators from Eq.~\eqref{eq:jumpOperatorsCG}:
\begin{align}
	{\tilde A}_l(\omega^*) :& =\int_{-\infty}^{\infty}\de t_1\;e^{it_1\omega^*}{f}^\alpha(t_1)A_l(t_1)\,,
\end{align}
where $f^{\alpha}(t)=\frac{e^{-\frac{t^2}{2T(\alpha)^2}}}{(\sqrt{\pi}\, T(\alpha))^{1/2}}$ . The coarse-graining induced by the filter function $f$ implies that these have vanishingly small tails as controlled by the Lieb-Robinson bound.  The bound states that for any operator $A_l$ supported on site $l \in \Lambda$ \cite{lieb1972finite,Haah_2021} (so that in this case $l$ is a label for the location):

\begin{equation}\label{eq:LiebRobinson}
	\vert \vert e^{-iHt}A_l e^{iHt}-e^{-iH_{B_l(r)}t}A_l e^{iH_{B_l(r)}t} \vert \vert \le \vert \vert A_l \vert \vert \frac{( v_{\text{LR}}\vert t \vert)^r}{r!}
\end{equation}
where $B_l(r)$ is a ball of radius $r$ centered at $l$, so that $H_{B_l(r)}$ consists of the Hamiltonian terms on a ball of radius $r$ centered at the support of $A_{l}$. Let us also define ${\tilde A}_{l,r}(\omega^*)$ as in Eq. \eqref{eq:jumpOperatorsCG} but with $A_{l,r}(t_1):=e^{-iH_{B_l(r)}t}A_l e^{iH_{B_l(r)}t}$ instead, so that its support is confined to $B_l(r)$. Then,
\begin{align}
	\vert \vert {\tilde A}_l(\omega^*)-{\tilde A}_{l,r}(\omega^*) \vert \vert &\le \int_{-\infty}^{\infty}\de t_1\;{f}^\alpha(t_1)\vert \vert A_l(t_1)-A_{l,r}(t_1) \vert \vert\leq
	\\
	& \le  \vert \vert A_l \vert \vert  \int_{-t_0}^{t_0} {f}^\alpha(t) \frac{( v_{\text{LR}} \vert t \vert )^r}{r!} \text{d} t + 4 \vert \vert A_l \vert \vert \int_{t_0}^\infty  {f}^\alpha(t) \text{d} t \leq
	\\
	& \le  \vert \vert A_l \vert \vert (\sqrt{\pi}\, T(\alpha))^{1/2}
	\frac{(v_{\text{LR}} e  \vert t_0 \vert )^r}{r^{r}} + 4 \vert \vert A_l \vert \vert \frac{T(\alpha)^{3/2}}{t_0} e^{-\frac{t_0^2}{2 T(\alpha)^2}},
\end{align}
where in the second line we divided the integral into two and applied the Lieb-Robinson bound to the first term, and in the second we used that $r! \ge \frac{r^r}{e^{r-1}}$ and the upper bound on the complementary error function $\text{Erfc}(x) \le \frac{e^{-x^2}}{\sqrt{\pi}x}$.

We now choose the parameter $t_0 = \frac{r g(v_{\text{LR}} T(\alpha))} {v_{\text{LR}} e}$ with $g(x)=\frac{\sqrt{x}}{1+\sqrt{x}}$, so that

\begin{align}
	\vert \vert {\tilde A}_l(\omega^*)-{\tilde A}_{l,r}(\omega^*) \vert \vert &\le \vert \vert A_l \vert \vert  \left ( (\sqrt{\pi}\, T(\alpha))^{1/2} g(v_{\text{LR}} T(\alpha))^r + \frac{4e}{r} \sqrt{v_{\text{LR}}} T(\alpha) (1+\sqrt{v_{\text{LR}} T(\alpha)}) e^{\frac{- r^2}{2 e^2 (1+\sqrt{v_{\text{LR}} T(\alpha)})v_{\text{LR}} T(\alpha)}} \right)\leq
	\\ & \le \vert \vert A_l \vert \vert  \mathcal{O}\left(  g(v_{\text{LR}} T(\alpha))^r \right),
\end{align}
which decays exponentially with the distance $r$ and a rate controlled by $T(\alpha)$, since $g(v_{\text{LR}} T(\alpha)) \rightarrow 0$ as $T(\alpha)\rightarrow 0$ and $g(v_{\text{LR}} T(\alpha)) \rightarrow 1$ as $T(\alpha)\rightarrow \infty$. Recall from Eq. \eqref{eq:optimalT} that the optimal value is $T^*_{opt}(\alpha) := \frac{1}{2\alpha \Gamma_0}\sqrt{2+3\,\Gamma_0\tau_0}$.



\subsection{Reduction to Davies generator}\label{app:davies}
\ms
{
	The expression of the rates in Eq.~\eqref{eq:gammarateCG} and Eq.~\eqref{eq:SrateCG} can be directly used to prove that $\widehat{\lind}^{CG}$ approaches the usual Davies generator in the limit $T(\alpha)\rightarrow\infty$. Define $\omega^{min}_{-} := \min_{\omega\neq \tilde\omega} (\omega-\tilde\omega)$. This quantity measures the minimum spacing between different frequencies, and it becomes exponentially small in the system size for many-body systems, but could in principle be finite for special classes of Hamiltonians (e.g., commuting ones). As it will be shown, $\omega^{min}_{-}$ governs how big $T(\alpha)$ needs to be for the reduction to the Davies generator to be accurate.
	
	To this end, let us rewrite the rates $\gamma_{kl}^{\omega,\tilde{\omega}}$ as:
	\begin{align}
		\gamma_{kl}^{\omega,\tilde{\omega}} &=\frac{e^{-(T(\alpha)\,\omega_-)^2/4}}{\sqrt{\pi}}\int_{-\infty}^{\infty}\de\Omega\;{\widehat{C}}_{kl}\norbra{\frac{\Omega}{T(\alpha)}-\omega_+}\,e^{-\Omega^2} =
		\\
		&=e^{-(T(\alpha)\,\omega_-)^2/4}\norbra{\widehat{C}_{kl}\norbra{-\omega_+}+\frac{1}{\sqrt{\pi}}\int_{-\infty}^{\infty}\de\Omega\int_{-\omega_+}^{\frac{\Omega}{T(\alpha)}-\omega_+}\de\omega_1\;{\widehat{C}}'_{kl}\norbra{\omega_1}\,e^{-\Omega^2}}\,.
	\end{align}
	The first term in the parenthesis is uniformly bound as $|\widehat{C}_{kl}\norbra{-\omega_+}|\leq \Gamma_0$. On the other hand, we can bound the second term as:
	\begin{align}
		\left|\frac{1}{\sqrt{\pi}}\int_{-\infty}^{\infty}\de\Omega\int_{-\omega_+}^{\frac{\Omega}{T(\alpha)}-\omega_+}\de\omega_1\;{\widehat{C}}'_{kl}\norbra{\omega_1}\,e^{-\Omega^2} \right| \leq \|\widehat{C}'_{kl}\|_\infty\frac{1}{\sqrt{\pi}}\int_{-\infty}^{\infty}\de\Omega\;\frac{|\Omega|}{T(\alpha)}\,e^{-\Omega^2} \leq \frac{\Gamma_0\tau_0}{\sqrt{\pi} \, T(\alpha)}\,.\label{eq:G22x}
	\end{align}
	
	The Davies generator has rates that are diagonal in the frequency representation, that is $(\gamma_{kl}^{\omega,\tilde{\omega}} )^{DAV}= \delta_{\omega,\tilde{\omega}} \widehat{C}_{kl}\norbra{-\omega}$. In our case, on the other hand, we have that:
	\begin{align}
		\gamma_{kl}^{\omega,\tilde{\omega}} = 
		\begin{cases}
			\widehat{C}_{kl}\norbra{-\omega}  +\bigo{\frac{\Gamma_0 \tau_0}{\sqrt{\pi} T(\alpha)}}& \qquad\qquad{\rm if}\; \omega=\tilde{\omega}\\
			\bigo{e^{-(T(\alpha)\,\omega_-^{min})^2/4}\,\Gamma_0\norbra{1+\frac{ \tau_0}{\sqrt{\pi} T(\alpha)}}} & \qquad\qquad{\rm otherwise}
		\end{cases}\;.\label{eq:app:g17}
	\end{align}
	This shows that in the limit $T(\alpha)\rightarrow\infty$ we obtain $\gamma_{kl}^{\omega,\tilde{\omega}} \rightarrow(\gamma_{kl}^{\omega,\tilde{\omega}} )^{DAV}$. Moreover, it also proves that the convergence of the off-diagonal is exponential for $T(\alpha)\gg (\omega_-^{min})^{-1}$. This means that if $\omega_-^{min}$ is bounded, the Davies generator is a good approximation already for small observation times.
	
	A similar discussion can also be made for the Lamb-shift Hamiltonian. Indeed, in the case of a Davies generator we have that:
	\begin{align}
		(S_{kl}^{\omega,\omega})^{DAV} = {\rm P.V.}\norbra{\frac{1}{\pi}\int_{-\infty}^{\infty}\de\omega^*\;\frac{\widehat{C}_{kl}(\omega^*)}{(\omega^*+\omega)}}\,.
	\end{align}
	In order to prove that the rates we obtain are close to the Davies' ones in the limit of large observation time, let us rewrite $S_{kl}^{\omega,\omega}$ as:
	\begin{align}
		S_{kl}^{\omega,\omega} &= \frac{e^{-(T(\alpha)\,\omega_-)^2/4}}{\pi}\int_{-\infty}^{\infty}\frac{\de\Omega}{\sqrt{\pi}}\int_{-\infty}^{\infty}\de\omega_2\;\frac{\widehat{C}_{kl}\norbra{\omega_2+\frac{\Omega}{T(\alpha)}}}{(\omega_2+\omega_+)}\;e^{-\Omega^2} =
		\\
		&=e^{-(T(\alpha)\,\omega_-)^2/4} \norbra{\frac{1}{\pi}\int_{-\infty}^{\infty}\de\omega_2\;\frac{\widehat{C}_{kl}\norbra{\omega_2}}{(\omega_2+\omega_+)}+\int_{-\infty}^{\infty}\frac{\de\Omega}{\sqrt{\pi}}\int_{0}^{\frac{\Omega}{T(\alpha)}}\de \omega_1\int_{-\infty}^{\infty}\frac{\de\omega_2}{\pi}\;\frac{\widehat{C}'_{kl}\norbra{\omega_2 +\omega_1}}{(\omega_2+\omega_+)}\;e^{-\Omega^2}}\,.
	\end{align}
	We can bound both terms inside of the parenthesis thanks to the relation in Eq.~\eqref{eq:G8x}. First, it should be noticed that:
	\begin{align}
		\left|\frac{1}{\pi}\int_{-\infty}^{\infty}\de\omega_2\;\frac{\widehat{C}_{kl}\norbra{\omega_2}}{(\omega_2+\omega_+)}\right| \leq \int_{-\infty}^{\infty} \de x\; \left|e^{ix\omega}\,C_{kl}(x){\rm sign}(x) \right| \leq \Gamma_0\,.
	\end{align}
	On the other hand, noticing that the inverse Fourier transform of $\widehat{C}'_{kl}\norbra{\omega_2 +\omega_1}$ (for fixed $\omega_1$) is given by $e^{i\omega_1 x}\,x\,{C}_{kl}\norbra{x}$, we can also bound the second term as:
	\begin{align}
		&\left|\int_{-\infty}^{\infty}\frac{\de\Omega}{\sqrt{\pi}}\int_{0}^{\frac{\Omega}{T(\alpha)}}\de \omega_1\int_{-\infty}^{\infty}\frac{\de\omega_2}{\pi}\;\frac{\widehat{C}'_{kl}\norbra{\omega_2 +\omega_1}}{(\omega_2+\omega_+)}\;e^{-\Omega^2}\right|\leq
		\\
		&\qquad\qquad\leq\int_{-\infty}^{\infty}\frac{\de\Omega}{\sqrt{\pi}}\int_{0}^{\frac{|\Omega|}{T(\alpha)}}\de \omega_1 \int_{-\infty}^{\infty} \de x\; \left|e^{ix(\omega_1-\omega_+)}x\,C_{kl}(x){\rm sign}(x) e^{-\Omega^2} \right|\leq \Gamma_0\tau_0 \int_{-\infty}^{\infty}\frac{\de\Omega}{\sqrt{\pi}}\int_{0}^{\frac{|\Omega|}{T(\alpha)}}\de \omega_1\;e^{-\Omega^2}=
		\\
		&\qquad\qquad=\frac{\Gamma_0\tau_0}{\sqrt{\pi} \, T(\alpha)}\,.
	\end{align}
	Then, also in this case we have a similar behaviour for the rates of the Lamb-shift Hamiltonian, that is:
	\begin{align}
		S_{kl}^{\omega,\tilde{\omega}} = 
		\begin{cases}
			(S_{kl}^{\omega,\omega})^{DAV}   +\bigo{\frac{\Gamma_0 \tau_0}{\sqrt{\pi} T(\alpha)}}& \qquad\qquad{\rm if}\; \omega=\tilde{\omega}\\
			\bigo{e^{-(T(\alpha)\,\omega_-^{min})^2/4}\,\Gamma_0\norbra{1+\frac{ \tau_0}{\sqrt{\pi} T(\alpha)}}} & \qquad\qquad{\rm otherwise}
		\end{cases}\;.\label{eq:app:g18}
	\end{align}
	
	Plugging the results just obtained into Eq.~\eqref{eq:LSHam} and Eq.~\eqref{eq:86} we obtain that:
	\begin{align}
		\lim_{T(\alpha)\rightarrow\infty}\widehat{\lind}^{CG} [\rho] = \widehat{\lind}^{DAV} [\rho]
		&=-\frac{i}{2}\sqrbra{H^{DAV}_{LS},\rho}+\int_{-\infty}^{\infty}\frac{\de\omega}{2\pi}\;{\widehat{C}}_{kl}(-\omega)\Big({ A}_l(\omega)\rho { A}^\dagger_k(\omega)-\frac{1}{2}\{{ A}^\dagger_k(\omega){ A}_l(\omega),\rho\}\Big)\,,\label{eq:daviesGen}
	\end{align}
	where we defined:
	\begin{align}
		H^{DAV}_{LS}:= \int_{-\infty}^{\infty} \de\omega \norbra{{\rm P.V.}\norbra{\frac{1}{\pi}\int_{-\infty}^{\infty}\de\omega^*\;\frac{\widehat{C}_{kl}(\omega^*)}{(\omega^*+\omega)}}}A^\dagger_k(\omega)A_l(\omega)\,.
	\end{align}
	This shows that when the observation time $T(\alpha)$ goes to infinity (in fact, as soon as $T(\alpha)\gg\omega_-^{min}$) we regain the usual Davies generator (see the expressions in e.g.~\cite{breuer2002theory, rivas2012open}). Moreover, thanks to the decomposition in Eq.~\eqref{eq:daviesGen}, we can directly apply Thm.~\ref{thm:GNSeqRW}, obtaining the well-known result that the Davies generator satisfies GNS detail balanced.
}

\subsection{Proof of Thm.~\ref{thm:approxDB}: Approximate detailed balance of the coarse-grained Lindbladian}\label{app:approxDB}
Let us begin with the first inequality. Let us first rewrite $\overline{\gamma_{kl}^{-\omega,-\tilde{\omega}}}$ as:
\begin{align}
	\overline{\gamma_{kl}^{-\omega,-\tilde{\omega}}} &= \frac{e^{-(T(\alpha)\,\omega_-)^2/4}}{\sqrt{\pi}}\int_{-\infty}^{\infty}\de\Omega\;{\widehat{C}}_{lk}\norbra{-\frac{\Omega}{T(\alpha)}+\omega_+}\,e^{-\Omega^2} = \\
	&=\frac{e^{-(T(\alpha)\,\omega_-)^2/4}}{\sqrt{\pi}}\int_{-\infty}^{\infty}\de\Omega\;{\widehat{C}}_{lk}\norbra{\frac{\Omega}{T(\alpha)}-\omega_+}\,e^{-\Omega^2}e^{-\frac{\beta\Omega}{T(\alpha)}}e^{\beta\omega_+}
\end{align}
where in the first step we implicitly changed $\Omega\rightarrow-\Omega$ and used the hermiticity of ${\widehat{C}}_{kl}(\omega)$ (see  
Eq.~\eqref{eq:correlationFunctionHermitian}) and the KMS condition from Eq.~\eqref{eq:corrKMS}. Then, we simply have:
\begin{align}
	\left|\gamma_{kl}^{\omega,\tilde{\omega}} -e^{-\beta(\frac{\omega+\tilde{\omega}}{2})} \overline{\gamma_{kl}^{-\omega,-\tilde{\omega}}}\right| &= \frac{e^{-(T(\alpha)\,\omega_-)^2/4}}{\sqrt{\pi}}\left|\int_{-\infty}^{\infty}\de\Omega\;{\widehat{C}}_{lk}\norbra{\frac{\Omega}{T(\alpha)}-\omega_+}\,e^{-\Omega^2}(1-e^{-\frac{\beta\Omega}{T(\alpha)}})\right|\leq
	\\
	&\leq\|\hat{C}_{kl}\|_\infty \,e^{-(T(\alpha)\,\omega_-)^2/4} e^{\frac{\beta^2}{4T(\alpha)^2}}{\rm erf}\norbra{\frac{\beta}{2T(\alpha)}} \,.
\end{align}
Thanks to the relation $\|\hat{C}_{kl}\|_\infty\leq \Gamma_0$ (see App.~\ref{app:relationTimescales}) and ${\rm erf}(x)\leq \frac{2}{\sqrt{\pi}}x$, we obtain Eq.~\eqref{eq:approxGamma}.

On the other hand, we also have:
\begin{align}
	\left|\tanh\norbra{\beta\norbra{\frac{\omega-\tilde{\omega}}{4}}}\gamma_{kl}^{\omega,\tilde{\omega}} \right|&=\left| \tanh\norbra{\frac{\beta\omega_-}{4}}\frac{e^{-(T(\alpha)\,\omega_-)^2/4}}{\sqrt{\pi}}\int_{-\infty}^{\infty}\de\Omega\;{\widehat{C}}_{kl}\norbra{\frac{\Omega}{T(\alpha)}-\omega_+}\,e^{-\Omega^2}\right|\leq
	\\
	&\leq\frac{\|\hat{C}_{kl}\|_\infty\beta|\omega_-|}{4}\,e^{-(T(\alpha)\,\omega_-)^2/4} \leq\frac{\Gamma_0\beta}{2 \sqrt{2 e}\, T(\alpha)}\,,\label{eq:139}
\end{align}
where in the first inequality we applied the relation $|\tanh(x)|\leq |x|$, and then we maximized over all possible values of $\omega_-$. Then, thanks to the triangle inequality, the bound in Eq.~\eqref{eq:139} gives the claim.

\ms
{
	\section{The effect of renormalizing the system Hamiltonian}\label{app:interactionPicture}
	\subsection{Proof of Thm.~\ref{thm:LSIgnore}: Closeness of the Lamb-shift Hamiltonians}\label{app:LSI}
	
	In this section we aim to prove that the solutions of the two differential equations:
	\begin{align}
		\begin{cases}
			\partial_{{t}}\,\tilde{\rho}^{CG}_{S}({t})=\widehat{\lind}^{CG} [\tilde{\rho}^{CG}_{S}({t})]\\
			\tilde{\rho}^{CG}_{S}(0) =\tilde{\rho}_{S}(0)
		\end{cases}\;;
		\qquad\qquad\qquad
		\begin{cases}
			\partial_{\tilde{t}}\,\tilde{\rho}^{CG*}_{S}({t})=\widehat{\lind}^{CG*} [\tilde{\rho}^{CG*}_{S}({t})]\\
			\tilde{\rho}^{CG*}_{S}(0) = \tilde{\rho}_{S}(0)
		\end{cases}\;.
	\end{align}
	diverge at most linearly in time. To this end, it should be noticed that in the renormalized interaction picture, the two generating Lindbladian differ in one-norm as:
	\begin{align}
		\|\widehat{\lind}^{CG}[\rho] - \widehat{\lind}^{CG*}[\rho]\|_1 = \frac{1}{2}\left\|\sqrbra{(H^{CG}_{LS})^*-H^{CG}_{LS},{\rho}}\right\|_1 \leq \|(H^{CG}_{LS})^*-H^{CG}_{LS}\|_\infty\,,\label{eq:appH2x}
	\end{align}
	where in the last step we used Hölder's inequality and the fact that $\|\rho\|_1 = 1$. Then, by controlling the right hand side of Eq.~\eqref{eq:appH2x}, we have a way to apply Thm.~\ref{thm:DBerrorBoundIntegrated} and estimate how the difference between $\tilde{\rho}^{CG}_{S}(t)$ and  $\tilde{\rho}^{CG*}_{S}(t)$ evolves in time.  
	
	To this end, let us define $\mathcal{U}^{\alpha^2}_{t}[X] = e^{-i(H_S+\alpha^2 H^{CG}_{LS}) t} X e^{i(H_S +\alpha^2 H^{CG}_{LS})t}$, so that $\mathcal{U}^{0}_{t}[X]\equiv \mathcal{U}_{t}[X]$. Moreover, by differentiating and integrating again, we can relate these two evolutions as:
	\begin{align}
		\mathcal{U}^{\alpha^2}_{t}[X]  = \mathcal{U}^{0}_{t}[X] - i\int_0^{\alpha^2}\de \nu \int_0^t\de\tau\; \mathcal{U}^{\nu}_{t-\tau}\sqrbra{[H^{CG}_{LS}, \mathcal{U}^{\nu}_{\tau}[X]]} = \mathcal{U}^{0}_{t}[X] + \Delta\mathcal{U}^{\alpha^2}_{t}[X]\,.
	\end{align}
	Then, using the rates defined in Eq.~\eqref{eq:defGammaS}, we can express $H^{CG}_{LS}$ in the time representation and give an upper-bound on its operator norm as:
	\begin{align}
		\left\|H^{CG}_{LS}\right\|_\infty &= \left\| i\int_{-\infty}^{\infty}\de q \int_{-\infty}^{\infty}\de x\; \mathcal{G}(q,x)  C_{kl}(x){\rm sign}(x) A_k\norbra{q+\frac{x}{2}}A_l\norbra{q-\frac{x}{2}}\right\|_\infty\leq
		\\
		&\leq\sum_{kl}\int_{-\infty}^{\infty}\de q \int_{-\infty}^{\infty}\de x\;\frac{e^{-\frac{q^2+(x/2)^2}{T(\alpha)^2}}}{\sqrt{\pi}\, T(\alpha)}|C_{kl}(x)|\leq \Gamma_0\,,\label{eq:upperBoundHLS}
	\end{align}
	where we used the fact that $\|A_k\|_\infty=1$ and the definition of $\Gamma_0$. We are now ready to bound the quantity in Eq.~\eqref{eq:appH2x}. First, it should be noticed that:
	\begin{align}
		(H^{CG}_{LS})^* &=  i\int_{-\infty}^{\infty}\de q \int_{-\infty}^{\infty}\de x\; \mathcal{G}(q,x)  C_{kl}(x){\rm sign}(x) (\mathcal{U}^{\alpha^2}_{q+\frac{x}{2}})^\dagger [A_k](\mathcal{U}^{\alpha^2}_{q-\frac{x}{2}})^\dagger[A_l] =	
		\\
		&= H^{CG}_{LS} +i\int_{-\infty}^{\infty}\de q \int_{-\infty}^{\infty}\de x\; \mathcal{G}(q,x)  C_{kl}(x){\rm sign}(x) \Big( (\Delta\mathcal{U}^{\alpha^2}_{q+\frac{x}{2}})^\dagger [A_k]\mathcal{U}^\dagger_{q-\frac{x}{2}}[A_l]+\mathcal{U}_{q+\frac{x}{2}}^\dagger [A_k](\Delta\mathcal{U}^{\alpha^2}_{q-\frac{x}{2}})^\dagger[A_l]+\nonumber
		\\
		&\qquad\qquad\qquad\qquad\qquad\qquad\qquad\qquad\qquad\qquad\qquad\qquad\qquad\qquad\qquad+(\Delta\mathcal{U}^{\alpha^2}_{q+\frac{x}{2}})^\dagger [A_k](\Delta\mathcal{U}^{\alpha^2}_{q-\frac{x}{2}})^\dagger[A_l]\Big) \,.
	\end{align}
	We can bound the operator norm of the first term in the parenthesis as:
	\begin{align}
		\left\| (\Delta\mathcal{U}^{\alpha^2}_{q+\frac{x}{2}})^\dagger [A_k]\mathcal{U}^\dagger_{q-\frac{x}{2}}[A_l]\right\|_\infty &\leq \int_0^{\alpha^2}\de \nu \int_0^{|q+\frac{x}{2}|}\de\tau\;\left\| (\mathcal{U}^{\nu}_{t-\tau})^\dagger\sqrbra{[H^{CG}_{LS}, (\mathcal{U}^{\nu}_{\tau})^\dagger[A_k]]}\right\|\|\mathcal{U}^\dagger_{q-\frac{x}{2}}[A_l]\|\leq 
		\\
		&\leq 2\Gamma_0 \int_0^{\alpha^2}\de \nu \int_0^{|q+\frac{x}{2}|} \de\tau= 2\alpha^2\Gamma_0\left|q+\frac{x}{2}\right|\,,
	\end{align}
	where we again used $\|A_k\|_\infty = 1$, the bound in Eq.~\eqref{eq:upperBoundHLS} and the fact that unitary dynamics keep the operator norm constant. The other terms can be estimated similarly. This allows to bound  Eq.~\eqref{eq:appH2x} as:
	\begin{align}
		&\|(H^{CG}_{LS})^*-H^{CG}_{LS}\|_\infty\leq
		\\
		&\qquad\leq \int_{-\infty}^{\infty}\de q \int_{-\infty}^{\infty}\de x\; \frac{e^{-\frac{q^2+(x/2)^2}{T(\alpha)^2}}}{\sqrt{\pi}\, T(\alpha)} |C_{kl}(x)|\norbra{2\alpha^2\Gamma_0\norbra{\left|q+\frac{x}{2}\right|+\left|q-\frac{x}{2}\right|}+ 4\alpha^4\Gamma_0^2 \left|q+\frac{x}{2}\right|\left|q-\frac{x}{2}\right|}\leq
		\\
		&\qquad\leq \int_{-\infty}^{\infty}\de q \int_{-\infty}^{\infty}\de x\; \frac{e^{-\frac{q^2+(x/2)^2}{T(\alpha)^2}}}{\sqrt{\pi}\, T(\alpha)} |C_{kl}(x)|\norbra{2\alpha^2\Gamma_0\norbra{2 |q| + |x|}+ 4\alpha^4\Gamma_0^2\norbra{|q|+\left|\frac{x}{2}\right|}^2}\,,
	\end{align}
	where in the last step we applied the triangle inequality to simplify the expression. There are five different terms that appear, which we study separately. First, the two linear terms can be upper-bounded as:
	\begin{align}
		&4\alpha^2\Gamma_0\int_{-\infty}^{\infty}\de q \int_{-\infty}^{\infty}\de x\; \frac{e^{-\frac{q^2+(x/2)^2}{T(\alpha)^2}}}{\sqrt{\pi}\, T(\alpha)} |C_{kl}(x)|\;|q| \leq \frac{4\alpha^2\Gamma_0^2 \,T(\alpha)}{\sqrt{\pi}}\,;
		\\
		&2\alpha^2\Gamma_0\int_{-\infty}^{\infty}\de q \int_{-\infty}^{\infty}\de x\; \frac{e^{-\frac{q^2+(x/2)^2}{T(\alpha)^2}}}{\sqrt{\pi}\, T(\alpha)} |C_{kl}(x)|\;|x| \leq 2\alpha^2\Gamma_0^2 \tau_0\,,
	\end{align}
	where in both cases we upper-bounded the Gaussian in $x$ with its value in zero. The quadratic terms, on the other hand can be estimated as:
	\begin{align}
		&4\alpha^4\Gamma_0^2\int_{-\infty}^{\infty}\de q \int_{-\infty}^{\infty}\de x\; \frac{e^{-\frac{q^2+(x/2)^2}{T(\alpha)^2}}}{\sqrt{\pi}\, T(\alpha)} |C_{kl}(x)|\;|q|^2 \leq 2\alpha^4\Gamma_0^3\, T(\alpha)^2
		\,;
		\\
		&4\alpha^4\Gamma_0^2\int_{-\infty}^{\infty}\de q \int_{-\infty}^{\infty}\de x\; \frac{e^{-\frac{q^2+(x/2)^2}{T(\alpha)^2}}}{\sqrt{\pi}\, T(\alpha)} |C_{kl}(x)|\;|q|\,|x|\leq \frac{4\alpha^4\Gamma_0^3\tau_0 \,T(\alpha)}{\sqrt{\pi}}\,;
		\\
		& 4\alpha^4\Gamma_0^2\int_{-\infty}^{\infty}\de q \int_{-\infty}^{\infty}\de x\; \frac{e^{-\frac{q^2+(x/2)^2}{T(\alpha)^2}}}{\sqrt{\pi}\, T(\alpha)} |C_{kl}(x)|\;\frac{|x|^2}{4} \leq \alpha^4\Gamma_0^2 K_0\,.
	\end{align}
	Wrapping everything together we obtain:
	\begin{align}
		\|(H^{CG}_{LS})^*-H^{CG}_{LS}\|_\infty\leq \frac{4\alpha^2\Gamma_0^2 \,T(\alpha)}{\sqrt{\pi}} +2\alpha^2\Gamma_0^2( \tau_0+\alpha^2\Gamma_0\, T(\alpha)^2) + \frac{4\alpha^4\Gamma_0^3\tau_0 \,T(\alpha)}{\sqrt{\pi}} + \alpha^4\Gamma_0^2 K_0\,.
	\end{align}
	Then, thanks to Eq.~\eqref{eq:appH2x}, this allows us to apply Thm.~\ref{thm:DBerrorBoundIntegrated} to obtain:
	\begin{align}
		\|\tilde{\rho}^{CG}_{S}({t}) -\tilde{\rho}^{CG*}_{S}({t})\|_1 \leq (\Gamma_0t)\norbra{\frac{4\alpha^2\Gamma_0 \,T(\alpha)}{\sqrt{\pi}} +2\alpha^2( \Gamma_0\tau_0+\alpha^2\, (\Gamma_0T(\alpha))^2) + \frac{4\alpha^4\Gamma_0^2\tau_0 \,T(\alpha)}{\sqrt{\pi}} + \alpha^4\Gamma_0 K_0}\,.
	\end{align}
	Choosing $T(\alpha) = T^*_{opt}(\alpha)$ (defined in Eq.~\eqref{eq:optimalT}) we get that the error scales as $\bigo{\alpha (\Gamma_0t) \sqrt{\Gamma_0 \tau_0}}$, where we assumed that $K_0<\infty$.

	\subsection{Comparison with the mean force state}\label{app:meanFieldH}
	The zeroth law of thermodynamics says that a system in contact with a thermal bath will eventually thermalize at the same temperature. Then, from this consideration we would expect the fixed point of the evolution to take the form:
	\begin{align}
		\rho_S^{DB*}(\infty) = \TrR{B}{\frac{e^{-\beta H}}{\mathcal{Z}_{SB}}}= \frac{e^{-\beta\,H^{(mf)}_S}}{\mathcal{Z}_S^{(mf)}} \,.
	\end{align}
	As noted in the main text, this is not the case. In order to show this discrepancy, let us define the mean force Hamiltonian as:
	\begin{align}
		e^{-\beta\,H^{(mf)}_S}:= \TrR{B}{\frac{e^{-\beta\,(H_S +\alpha V+H_B)}}{\mathcal{Z}_B}}\,,\label{eq:H21x}
	\end{align}
	and let us assume, for simplicity, that it is analytic in $\alpha$. This means, that we can expand $H^{(mf)}_S$ in a power series as:
	\begin{align}
		H^{(mf)}_S = \sum_{n=0}^\infty \; \frac{\alpha^n}{n!}\,H^{(n)}_S\,,
	\end{align}
	where $H^{(n)}_S$ are independent on $\alpha$. First of all, setting $\alpha=0$ in Eq.~\eqref{eq:H21x} gives $H^{(0)}_S=H_S$. Moreover, if we take the derivative of the right hand side of Eq.~\eqref{eq:H21x} and we set it to zero, we also obtain that:
	\begin{align}
		\partial_{\alpha} \TrR{B}{\frac{e^{-\beta\,(H_S +\alpha V+H_B)}}{\mathcal{Z}_B}}\bigg|_{\alpha=0} = -\beta\int_0^1 \de x\; e^{-\beta\,H_S(1-x)} \TrR{B}{ V\gamma_B}e^{-\beta\,H_Sx} = 0\,,
	\end{align}
	where we applied the assumption~\ref{it:noEnergyShift} from the main text (that is $\TrR{B}{B_k\gamma_B}=0$). This shows that $H^{(1)}_S=0$. Finally, let us focus on the second order contribution to the mean force Hamiltonian. Let us first take the second derivative of the left hand side of Eq.~\eqref{eq:H21x}. Since $H^{(1)}_S=0$, this is given by:
	\begin{align}
		\partial_{\alpha}^2(e^{-\beta\,H^{(mf)}_S})\Big|_{\alpha=0} = -\beta\int_0^1 \de x\; e^{-\beta\,H_S(1-x)}H_S^{(2)}e^{-\beta\,H_Sx}\,.\label{eq:appH24x}
	\end{align}
	On the other hand, the second derivative of the right hand side of Eq.~\eqref{eq:H21x} can be expressed through the Dyson series as:
	\begin{align}
		\partial_{\alpha}^2 \TrR{B}{\frac{e^{-\beta\,(H_S +\alpha V+H_B)}}{\mathcal{Z}_B}}\bigg|_{\alpha=0} &= 2\beta^2\int_0^1 \de x_1\int_0^{x_1}\de x_2\; \TrR{B}{\pi_{SB}^{1-x_1}\,V\,\pi_{SB}^{x_1-x_2}V\,\pi_{SB}^{x_2}}=
		\\
		&=2\beta^2\int_0^1 \de x\int_{\frac{x-1}{2}}^{\frac{1-x}{2}}\de q\; \TrR{B}{\pi_{SB}^{\frac{1}{2}-q-\frac{x}{2}}\,V\,\pi_{SB}^{x}V\,\pi_{SB}^{\frac{1}{2}+q-\frac{x}{2}}}\,,
	\end{align}
	where we have introduced the notation $\pi_{SB}:= e^{-\beta\,H_S}\otimes \gamma_B$, and in the second line we carried out the substitution $\{q = \frac{x_1+x_2}{2}-1\,; \, x= x_1-x_2\}$. At this point, it is useful to expand the equation above in the frequency domain as:
	\begin{align}
		&2\beta^2\sum_{kl}\int_0^1 \de x\int_{\frac{x-1}{2}}^{\frac{1-x}{2}}\de q\; \TrR{B}{\pi_{SB}^{\frac{1}{2}-q-\frac{x}{2}}\,A^\dagger_k(\tilde{\omega})\otimes B_k(\Omega)\,\pi_{SB}^{x}A_l(\omega)\otimes B_l\,\pi_{SB}^{\frac{1}{2}+q-\frac{x}{2}}} = 
		\\
		&=e^{-\frac{\beta\,H_S}{2}}(A^\dagger_k(\tilde{\omega})A_l(\omega))e^{-\frac{\beta\,H_S}{2}}\norbra{2\beta^2\int_0^1 \de x\int_{\frac{x-1}{2}}^{\frac{1-x}{2}}\de q\; e^{-\beta (\tilde{\omega}-\omega)q}e^{-\beta \frac{(\tilde{\omega}+\omega)}{2}x}\,\TrR{B}{\gamma_B^{-x} B_k(\Omega) \gamma_B^x B_l\gamma_B}}\,,\label{eq:H28x}
	\end{align}
	where, here and in the following, we always implicitly integrate over all frequencies.
	We can further express the trace as:
	\begin{align}
		\TrR{B}{\gamma_B^{-x} B_k(\Omega) \gamma_B^x B_l\gamma_B} = e^{\beta \Omega x} \int_\infty^\infty \frac{\de t}{2\pi}  \; e^{-i\Omega t}\;\TrR{B}{ B_k(t) B_l\gamma_B} = \frac{1}{2\pi}\, e^{\beta \Omega x}  \,\widehat{C}_{kl}(-\Omega)\,,
	\end{align}
	where in the last step we used the definition of spectral density (see Eq.~\eqref{eq:spectralDensity}).
	
	At this point, it is useful to expand Eq.~\eqref{eq:appH24x} in the frequency representation as well:
	\begin{align}
		\partial_{\alpha}^2(e^{-\beta\,H^{(mf)}_S})\Big|_{\alpha=0}  &= -\beta\int_0^1 \de x\; e^{-\beta\,H_S(1-x)}H_S^{(2)}(\hat{\omega})e^{-\beta\,H_Sx} =-\beta e^{-\frac{\beta\,H_S}{2}} H_S^{(2)}(\hat{\omega})  e^{-\frac{\beta\,H_S}{2} }  \norbra{\frac{2\sinh\norbra{\frac{\beta \hat\omega }{2}}}{\beta\hat \omega}} \,.\label{eq:H30x}
	\end{align}
	Equating the terms with the same frequency in Eq.~\eqref{eq:H30x} and Eq.~\eqref{eq:H28x} (which corresponds to choosing $\omega =\hat \omega +\tilde \omega$), we obtain:
	\begin{align}
		H_S^{(2)}&(\hat{\omega})  =A^\dagger_k(\tilde{\omega})A_l(\hat{\omega}+\tilde{\omega})\norbra{-\frac{\beta}{\pi}\int_0^1 \de x\int_{\frac{x-1}{2}}^{\frac{1-x}{2}}\de q\; e^{\beta \hat{\omega}q}e^{-\beta \frac{(2\tilde{\omega}+\hat{\omega})}{2}x}\, e^{\beta \Omega x}  \,\widehat{C}_{kl}(-\Omega)}\norbra{\frac{2\sinh\norbra{\frac{\beta \hat\omega }{2}}}{\beta\hat \omega}}^{-1} = 
		\\
		&=A^\dagger_k(\tilde{\omega})A_l(\hat{\omega}+\tilde{\omega}) \norbra{\frac{1}{\pi}\int_{-\infty}^\infty \de \Omega\;\norbra{\norbra{\frac{e^{\beta\Omega}-e^{\beta(\hat{\omega}+\tilde{\omega})}}{e^{\beta\tilde{\omega}}-e^{\beta(\hat{\omega}+\tilde{\omega})}}}\frac{\widehat{C}_{kl}(-\Omega)}{\Omega-(\hat{\omega}+\tilde{\omega})}+\norbra{\frac{e^{\beta\tilde{\omega}}-e^{\beta\Omega}}{e^{\beta\tilde{\omega}}-e^{\beta(\hat{\omega}+\tilde{\omega})}}}\frac{\widehat{C}_{kl}(-\Omega)}{\Omega-\tilde{\omega}}}}\,.\label{eq:H32x}
	\end{align}
	An equivalent expression was found in~\cite{winczewski2021renormalization}, where they presented a version in which the integral over $\Omega$ was computed, which gives:
	\begin{align}
		H_S^{(2)}(\hat{\omega})=
		&\frac{i}{2}A^\dagger_k(\tilde{\omega})A_l(\hat{\omega}+\tilde{\omega})\cdot\nonumber\\
		&\int_{-\infty}^\infty \de x\;\norbra{\frac{1}{e^{\beta\hat{\omega}}-1}}\norbra{\norbra{e^{-ix\tilde{\omega}}-e^{\beta\hat{\omega}}e^{-ix(\hat{\omega}+\tilde{\omega})}}C_{kl}(x)- e^{-\beta \tilde{\omega}}\norbra{e^{-ix\tilde{\omega}}-e^{ix(\hat{\omega}+\tilde{\omega})}}C_{lk}(x)}{\rm sign}(x)\,.
	\end{align}
	The expression in Eq.~\eqref{eq:H32x} should be compared with the one for $H_{LS}^{CG}$, that is:
	\begin{align}
		H_{LS}^{CG}(\hat{\omega}) = A^\dagger_k(\tilde\omega)A_l(\hat{\omega}+ \tilde{\omega}) \frac{e^{-(T(\alpha)\,\hat{\omega})^2/4}}{\sqrt{\pi}}\int_{-\infty}^{\infty}\frac{\de\Omega}{\pi}\int_{-\infty}^{\infty}\de\omega_2\;\frac{\widehat{C}_{kl}(\omega_2)}{(\omega_2+(2\tilde{\omega}+\hat{\omega}))-\frac{\Omega}{T(\alpha)}}\;e^{-\Omega^2} \,.
	\end{align}
	The difference between $H_S^{(2)}(\hat{\omega})$ and $H_{LS}^{CG}(\hat{\omega}) $ is evident, even more so when checking the behaviour for $\hat{\omega}\gg 1$: whereas in the case of the Lamb-shift term, these contributions are exponentially suppressed, the contribution to the mean force  Hamiltonian converges to the Hilbert transform of $\widehat{C}_{kl}(\omega)$ evaluated at $\hat{\omega}$. Thus, the two Hamiltonians are different even in the limit $\alpha\rightarrow0$.
	
	\subsection{Proof of Thm.~\ref{thm:interactionInsensitive}: Closeness of fixed points}\label{app:intInsensitive}
	
	Despite the negative result of the previous section, it can be proven that the renormalized state is sufficiently close to the bare Hamiltonian one. To this end, for a fixed renormalization $H_{LS}^{CG}$, let us define $\gamma_S^{\alpha^2}$ as:
	\begin{align}
		\gamma_S^{\alpha^2} := \frac{e^{-\beta (H_S+\frac{\alpha^2}{2} H_{LS}^{CG})}}{\mathcal{Z}_S^{\alpha^2}} = \gamma_S^0 -\beta\int_{0}^{\frac{\alpha^2}{2}}\de \nu\int_0^1 \de x\; (\gamma_S^{\nu})^{1-x}\norbra{H_{LS}^{CG}- \Tr{H_{LS}^{CG} \gamma_S^\nu}\idO}(\gamma_S^{\nu})^x\,, \label{eq:H35x}
	\end{align}
	where, as usual, we took the derivative and then integrated again to obtain the right hand side. Thanks to the symmetry of the integral in $x$, we can also rewrite Eq.~\eqref{eq:H35x} as:
	\begin{align}
		\gamma_S^{\alpha^2} =\gamma_S^0 -\frac{\beta}{2}\int_{0}^{\frac{\alpha^2}{2}}\de \nu\int_0^1 \de x\; \norbra{(\gamma_S^{\nu})^{1-x}\norbra{H_{LS}^{CG}- \Tr{H_{LS}^{CG} \gamma_S^\nu}\idO}(\gamma_S^{\nu})^x + (\gamma_S^{\nu})^{x}\norbra{H_{LS}^{CG}- \Tr{H_{LS}^{CG} \gamma_S^\nu}\idO}(\gamma_S^{\nu})^{1-x}}\,.
	\end{align}
	This expression is useful in order to estimate the difference between $\gamma_S^{\alpha^2}$ and $\gamma_S^{0}$. Indeed, let $A$ and $B$ be positive, and $X$ be any generic matrix. Then, for all $t\in [0,1]$, it holds that (see Corollary IX.4.10 in~\cite{bhatia2013matrix}):
	\begin{align}
		\|A^t X B^{1-t} + A^{1-t} X B^t\|_1 \leq \|A X + X B\|_1\,.
	\end{align}
	Then, applying this inequality to our case we obtain:
	\begin{align}
		\|\gamma_S^{0}-\gamma_S^{\alpha^2}\|_1 &\leq\, \frac{\beta}{2}\int_{0}^{\frac{\alpha^2}{2}}\de \nu\int_0^1 \de x\;\| \gamma_S^{\nu}\norbra{H_{LS}^{CG}- \Tr{H_{LS}^{CG}}\idO} +\norbra{H_{LS}^{CG}- \Tr{H_{LS}^{CG}\gamma_S^{\nu}}\idO}\gamma_S^{\nu}\|_1\leq
		\\
		&\leq\frac{\alpha^2\beta}{2}\, \|\norbra{H_{LS}^{CG}- \Tr{H_{LS}^{CG}\gamma_S^{\nu}}\idO}\gamma_S^{\nu}\|_1 \leq\frac{\alpha^2\beta}{2}\, \|\norbra{H_{LS}^{CG}- \Tr{H_{LS}^{CG}\gamma_S^{\nu}}\idO}\|_\infty\,,
	\end{align}
	where in the last step we applied Hölder's inequality, together with $\|\gamma_S^\nu\|_1=1$. Finally, the claim follows from the bound in Eq.~\eqref{eq:upperBoundHLS} together with the fact that $\|\Tr{H_{LS}^{CG}\gamma_S^{\nu}}\idO\|_\infty \leq \|H_{LS}^{CG}\|_\infty$.
}

\section{Exactly detailed balanced Lindbladian}\label{app:exactDBLind}
\subsection{Proof of Thm.~\ref{thm:approxGammaSquareRoot}: Approximation to the coarse-grained rates}\label{app:thmapproxGammaSquare}
Once again we resort to Eq.~\eqref{eq:gammarateCG} in order to prove the claim.
We will do this in two steps, corresponding to the two terms in the triangle inequality:
\begin{align}
	|\gamma_{kl}^{\omega,\tilde{\omega}} - e^{-(T(\alpha)\,\omega_-)^2/4}\;{\widehat{g}}_{k\lambda}\norbra{-\tilde\omega } {\widehat{g}}_{\lambda l}\norbra{-{\omega} }|\leq\;
	&|\gamma_{kl}^{\omega,\tilde{\omega}} - e^{-(T(\alpha)\,\omega_-)^2/4}\;{\widehat{C}}_{kl}\norbra{-\omega_+}|+\nonumber
	\\
	&+ e^{-(T(\alpha)\,\omega_-)^2/4}\left|{\widehat{C}}_{kl}\norbra{-\omega_+} -{\widehat{g}}_{k\lambda}\norbra{-\tilde\omega } {\widehat{g}}_{\lambda l}\norbra{-{\omega} } \right|\,.\label{eq:app:h2!}
\end{align}
It should be noticed that the first term actually is the same that appeared in the estimate in  Eq.~\eqref{eq:G22x}. Then, repeating the same calculation we have: 
\begin{align}
	|\gamma_{kl}^{\omega,\tilde{\omega}} - e^{-(T(\alpha)\,\omega_-)^2/4}\;{\widehat{C}}_{kl}\norbra{-\omega_+}|&= \left|\frac{e^{-(T(\alpha)\,\omega_-)^2/4}}{\sqrt{\pi}}\int_{-\infty}^{\infty}\de\Omega\; \int_{-\omega_+}^{\frac{\Omega}{T(\alpha)}-\omega_+}\de \omega_1\;{\widehat{C}}_{kl}'\norbra{\omega_1}\,e^{-\Omega^2} \right|\leq
	\\
	&\leq {e^{-(T(\alpha)\,\omega_-)^2/4}} \frac{\|{\widehat{C}}'_{kl}\|_\infty}{\sqrt{\pi}\,T(\alpha)}\,.\label{eq:app:h3!}
\end{align}
\ms
{
	On the other hand, the second term in Eq.~\eqref{eq:app:h2!}, can be estimated as:
	\begin{align}
		&{e^{-(T(\alpha)\,\omega_-)^2/4}}\left|{\widehat{C}}_{kl}\norbra{-\omega_+}-{\widehat{g}}_{k\lambda}\norbra{-\tilde\omega } {\widehat{g}}_{\lambda l}\norbra{-{\omega} }\right| =
		\\
		&={e^{-(T(\alpha)\,\omega_-)^2/4}}\left|{\widehat{g}}_{k\lambda}\norbra{-\tilde\omega -\frac{\omega_-}{2}}{\widehat{g}}_{\lambda l}\norbra{-\omega +\frac{\omega_-}{2}}-{\widehat{g}}_{k\lambda}\norbra{-\tilde\omega } {\widehat{g}}_{\lambda l}\norbra{-{\omega} }\right|=
		\\
		&={e^{-(T(\alpha)\,\omega_-)^2/4}} \left|\int_0^{\frac{\omega_-}{2}}\de\omega_1\;\norbra{{\widehat{g}}_{k\lambda}\norbra{-\tilde\omega}{\widehat{g}}'_{\lambda l}\norbra{-\omega +\omega_1}+{\widehat{g}}'_{k\lambda}\norbra{-\tilde\omega-\omega_1}\int_{0}^{\frac{\omega_-}{2}}\de\omega_2\norbra{{\widehat{g}}_{\lambda l}\norbra{-\omega}+{\widehat{g}}'_{\lambda l}\norbra{-\omega+\omega_2}}}\right|\leq
		\\
		&\leq\norbra{\|{\widehat{g}}_{k\lambda}\|_\infty\|{\widehat{g}}'_{\lambda l}\|_\infty+\|{\widehat{g}}'_{k\lambda}\|_\infty\|{\widehat{g}}_{\lambda l}\|_\infty+\frac{|\omega_-|}{2}\|{\widehat{g}}'_{k\lambda}\|_\infty\|{\widehat{g}}'_{\lambda l}\|_\infty}\;{{e^{-(T(\alpha)\,\omega_-)^2/4}} \frac{|\omega_-|}{2}} \leq
		\\
		&\leq  \norbra{\frac{\|{\widehat{g}}'_{k\lambda}\|_\infty\|{\widehat{g}}_{\lambda l}\|_\infty+\|{\widehat{g}}_{k\lambda}\|_\infty\|{\widehat{g}}'_{\lambda l}\|_\infty}{2}} \frac{\sqrt{2}}{\sqrt{e}\, T(\alpha)} +\|{\widehat{g}}'_{k\lambda}\|_\infty\|{\widehat{g}}'_{\lambda l}\|_\infty \frac{2}{e (T(\alpha)^2)}\,,\label{eq:app:h7!}
	\end{align}
	where in the second line we used the definition of $\omega_+ := \frac{\omega+\tilde{\omega}}{2}$ and $\omega_- := {({\omega}-\tilde\omega)}$; at this point, we derived and integrated again, and finally we took the superior over all possible values of $\omega_-$. Putting together Eq.~\eqref{eq:app:h3!} and Eq.~\eqref{eq:app:h7!}, together with the two estimates $\|{\widehat{C}}'_{kl}\|_\infty\leq \Gamma \tau$ and $\|{\widehat{g}}_{k\lambda}\|_\infty\|{\widehat{g}}'_{k\lambda}\|_\infty\leq \Gamma \tau$ and $\|{\widehat{g}}'_{k\lambda}\|_\infty\|{\widehat{g}}'_{k\lambda}\|_\infty\leq K$ (see App.~\ref{app:relationTimescales}), we obtain the claim.}

\subsection{Proof of Thm.~\ref{thm:dbLindlbadian}: $\widehat{\lind}^{DB*} $ is exactly detailed balanced and CP}\label{app:dbLindbladian}
\ms
{
	For ease of notation, in this section we drop the star from $\lind^{DB*}$, the jump operators and the Lamb-shift Hamiltonian.
	
	We begin by proving KMS detailed balance for the coefficients in Eq.~\eqref{eq:DBcoeffCP}. This holds by construction.} Indeed, the second property in Eq.~\eqref{eq:DBdefinition} is automatically satisfied, whereas the first condition can be proven by direct calculation:
\begin{align}
	\overline{\tilde{\gamma}_{kl}^{-\omega,-\tilde{\omega}}} = \overline{e^{-(T(\alpha)\,\omega_-)^2/4}}\;\overline{{\widehat{g}}_{k\lambda}\norbra{\tilde\omega } {\widehat{g}}_{\lambda l}\norbra{{\omega} }} = e^{-(T(\alpha)\,\omega_-)^2/4}\;{\widehat{g}}_{k\lambda}\norbra{-\tilde\omega } {\widehat{g}}_{\lambda l}\norbra{-{\omega} } e^{\beta\norbra{\frac{\omega+{\tilde\omega}}{2}}} =	\tilde{\gamma}_{kl}^{\omega,\tilde{\omega}} e^{\beta\norbra{\frac{\omega+{\tilde\omega}}{2}}}\,,
\end{align}
where, once again, we used the property in Eq.~\eqref{eq:corrKMS} together with the definition of ${\widehat{g}}_{k\lambda}\norbra{\omega }$. Then, it remains to prove that these coefficients induce a CP evolution. Once again, for the part in $\tilde{S}_{kl}^{\omega,\tilde{\omega}}$ it is sufficient to prove that the Lamb Shift Hamiltonian:
\begin{align}
	H^{DB}_{LS} := \int_{-\infty}^{\infty} \de\omega \int_{-\infty}^{\infty} \de\tilde\omega\;{\tilde S}_{kl}^{\omega,\tilde{\omega}}A_k^\dagger(\tilde\omega)A_l(\omega)
\end{align}
is Hermitian. Then, using the explicit expression of ${\tilde S}_{kl}^{\omega,\tilde{\omega}}$ we have:
\begin{align}
	(H^{DB}_{LS})^\dagger :&= -i\int_{-\infty}^{\infty} \de\omega \int_{-\infty}^{\infty} \de\tilde\omega\;\tanh\norbra{\beta\norbra{\frac{\omega-\tilde{\omega}}{4}}}\overline{\tilde{\gamma}_{kl}^{\omega,\tilde{\omega}}}	A^\dagger_l(\omega)A_k(\tilde\omega) =
	\\
	&= i\int_{-\infty}^{\infty} \de\omega \int_{-\infty}^{\infty} \de\tilde\omega\;\tanh\norbra{-\beta\norbra{\frac{\omega-\tilde{\omega}}{4}}}\overline{\tilde{\gamma}_{lk}^{{\omega},\tilde\omega}}	A^\dagger_k(\omega)A_l(\tilde\omega)  =
	\\
	&=i\int_{-\infty}^{\infty} \de\omega \int_{-\infty}^{\infty} \de\tilde\omega\;\tanh\norbra{\beta\norbra{\frac{\omega-\tilde{\omega}}{4}}}\tilde{\gamma}_{kl}^{\omega,\tilde{\omega}} 	A^\dagger_k(\tilde\omega)A_l(\omega) = H^{DB}\,,
\end{align}
thanks to the property of $\tilde{\gamma}_{kl}^{\omega,\tilde{\omega}} $:
\begin{align}
	\overline{\tilde{\gamma}_{lk}^{{\omega},\tilde\omega}} = e^{-(T(\alpha)\,\omega_-)^2/4}\;\overline{{\widehat{g}}_{l\lambda}\norbra{-\tilde\omega } {\widehat{g}}_{\lambda k}\norbra{-{\omega} }} =e^{-(T(\alpha)\,\omega_-)^2/4}\;{{\widehat{g}}_{\lambda l}\norbra{-\tilde\omega } {\widehat{g}}_{k\lambda}\norbra{-{\omega} }}=\tilde{\gamma}_{kl}^{\tilde{\omega},\omega} \,,
\end{align}
where we implicitly used Eq.~\eqref{eq:correlationFunctionHermitian}.

We can now focus on the dissipative term $\widehat{\mathcal{D}}^{DB}$. For brevity, let us introduce the notation $\mathcal{D}^{\omega,\tilde{\omega}}_{kl}$:
\begin{align}
	\mathcal{D}^{\omega,\tilde{\omega}}_{kl}[\rho] := \Big(A_l(\omega)\rho A^\dagger_k(\tilde\omega)-\frac{1}{2}\{A^\dagger_k(\tilde\omega)A_l(\omega),\rho\}\Big)\,.\label{eq:dissipatorFrequencyRep}
\end{align}
Then, we can express $\widehat{\mathcal{D}}^{DB}$ as:
\begin{align}
	\widehat{\mathcal{D}}^{DB}[\rho] &=  \int_{-\infty}^{\infty} \de\omega \int_{-\infty}^{\infty} \de\tilde\omega\; \tilde{\gamma}_{kl}^{\omega,\tilde{\omega}}\mathcal{D}^{\omega,\tilde{\omega}}_{kl}[\rho] =
	\\
	&=\frac{T(\alpha)}{\sqrt{\pi}}\int_{-\infty}^{\infty} \de\omega^* \; \int_{-\infty}^{\infty} \de\omega \int_{-\infty}^{\infty} \de\tilde\omega\;e^{-T(\alpha)^2\frac{(\omega^*-\omega)^2}{2}}e^{-T(\alpha)^2\frac{(\omega^*-\tilde{\omega})^2}{2}}\;\overline{{\widehat{g}}_{\lambda k}\norbra{-\tilde\omega } }{\widehat{g}}_{\lambda l}\norbra{-{\omega} }\mathcal{D}^{\omega,\tilde{\omega}}_{kl}[\rho] =
	\\
	&=  \int_{-\infty}^{\infty} \de\omega^* \Big(\hat{A}_\lambda(\omega^*)\rho \hat{A}^\dagger_\lambda(\omega^*)-\frac{1}{2}\{\hat{A}_\lambda^\dagger(\omega^*)\hat{A}_\lambda(\omega^*),\rho\}\Big) \label{eq:105}
\end{align}
where, in the second line we used the relation for Gaussians: 
\begin{align}
	e^{-(T(\alpha)\,(\omega-\tilde{\omega}))^2/4} = \frac{T(\alpha)}{\sqrt{\pi}}\int_{-\infty}^{\infty} \de\omega^* \; e^{-T(\alpha)^2\frac{(\omega^*-\omega)^2}{2}}e^{-T(\alpha)^2\frac{(\omega^*-\tilde{\omega})^2}{2}}
\end{align}
and in the last step we defined the jump operators:
\begin{align}
	&\hat{A}_\lambda(\omega^*) = \sqrt{\frac{T(\alpha)}{\sqrt{\pi}}}\int_{-\infty}^{\infty} \de\omega\;e^{-T(\alpha)^2\frac{(\omega^*-{\omega})^2}{2}}\;{\widehat{g}}_{\lambda l}\norbra{-{\omega}}A_l({\omega})\,.
\end{align}
Since Eq.~\eqref{eq:105} is in Lindblad form, the corresponding semigroup is completely positive.


\subsection{Quasilocality of the jump operators}\label{app:quasiLocalJump2}

We can now write the jump operators as

\begin{align}
	\hat{A}_\lambda(\omega^*) &= \sqrt{\frac{T(\alpha)}{\sqrt{\pi}}}\int_{-\infty}^{\infty} \de\omega\;e^{-T(\alpha)^2\frac{(\omega^*-{\omega})^2}{2}}\;{\widehat{g}}_{\lambda l}\norbra{-{\omega}}A_l({\omega})\,=
	\\
	& = \int_{-\infty}^{\infty} e^{i \omega^* t } A(t) \text{d}t \left ( \frac{1}{\sqrt{\sqrt{\pi} T(\alpha)}}\,\norbra{e^{\frac{-t^2}{2 T(\alpha)^2}}  \ast g_{\lambda l}}(t) \right )
\end{align}

By the same argument as in App. \ref{app:quasiLocalJump}, we have that these jump operators are again quasi-local, with their locality now controlled by the convolution of the Gaussian and the bath correlation function $g_{\lambda l}$.
In contrast to App. \ref{app:quasiLocalJump}, this means that now the locality also depends on the decay rate of these bath correlation functions, which in App.~\ref{app:relationTimescales} is defined as $\Gamma$. Considering the effect of the convolution, we expect the radius of the terms to be $\mathcal{O} \left ( v_{\text{LR}} T(\alpha) + \frac{v_{\text{LR}}}{\Gamma} \right)$, thus decaying roughly as $\left (\frac{\sqrt{v_{\text{LR}} T(\alpha) + \frac{v_{\text{LR}}}{\Gamma}}}{1+\sqrt{v_{\text{LR}} T(\alpha) + \frac{v_{\text{LR}}}{\Gamma}}} \right)^r$.	


\subsection{Reduction to Davies generator}\label{app:davies2}
\ms
{
	The derivation is very similar to the one in App.~\ref{app:davies}. Let us begin by expanding the rates $\tilde{\gamma}_{kl}^{\omega,\tilde{\omega}}$ as:
	\begin{align}
		\tilde{\gamma}_{kl}^{\omega,\tilde{\omega}} &= e^{-(T(\alpha)\,\omega_-)^2/4}\;{\widehat{g}}_{k\lambda}\norbra{-\tilde\omega } {\widehat{g}}_{\lambda l}\norbra{-{\omega} } =  \label{eq:i23}
		\\
		&= e^{-(T(\alpha)\,\omega_-)^2/4}\;\norbra{{\widehat{g}}_{k\lambda}\norbra{-\omega} {\widehat{g}}_{\lambda l}\norbra{-{\omega} }+\int_{0}^{\frac{\omega_-}{2}}\de \omega_1\;{\widehat{g}}'_{k\lambda}\norbra{-\tilde\omega -\omega_1} {\widehat{g}}_{\lambda l}\norbra{-{\omega} }}
	\end{align}
	Carrying out the sum over $\lambda$, the first term reduces to ${\widehat{g}}_{k\lambda}\norbra{-\omega} {\widehat{g}}_{\lambda l}\norbra{-{\omega} } = C_{kl}(-\omega)$, so it can be upper-bounded by $\Gamma_0$. On the other hand, the second term can also be controlled as:
	\begin{align}
		\left|\int_{0}^{\frac{\omega_-}{2}}\de \omega_1\;{\widehat{g}}'_{k\lambda}\norbra{-\tilde\omega -\omega_1} {\widehat{g}}_{\lambda l}\norbra{-{\omega} }\right|  \leq \frac{|\omega_-|\Gamma\,\tau}{2}
	\end{align}
	It should be noticed that the function $|\omega_-|e^{-(T(\alpha)\,\omega_-)^2/4}$ can be explicitly maximised (as done in App.~\ref{app:thmapproxGammaSquare}) to give $|\omega_-|e^{-(T(\alpha)\,\omega_-)^2/4}\leq \frac{\sqrt{2}}{\sqrt{e}T(\alpha)}$. On the other hand, it should also be noticed that a trivial upper-bound to Eq.~\eqref{eq:i23} can be given by using the fact that $|\widehat{g}_{k\lambda}(\omega)||\widehat{g}_{\lambda l}(\tilde\omega)|\leq \Gamma$ (see App.~\ref{app:relationTimescales}). Then, wrapping everything together we have:
	\begin{align}
		{\tilde \gamma}_{kl}^{\omega,\tilde{\omega}} = 
		\begin{cases}
			\widehat{C}_{kl}\norbra{-\omega}  +\bigo{\frac{\sqrt{2}\Gamma \tau}{\sqrt{e} T(\alpha)}}& \qquad\qquad{\rm if}\; \omega=\tilde{\omega}\\
			\bigo{e^{-(T(\alpha)\,\omega_-^{min})^2/4}\,\Gamma} & \qquad\qquad{\rm otherwise}
		\end{cases}
	\end{align}
	that is $\tilde{\gamma}_{kl}^{\omega,\tilde{\omega}}$ and ${\gamma}_{kl}^{\omega,\tilde{\omega}}$ coincide in the limit $T(\alpha)\rightarrow\infty$, and they both coincide with what one would expect from the Davies generator (see Eq.~\eqref{eq:app:g17} and Eq.~\eqref{eq:daviesGen}). 
	
	On the other hand, having renormalised the interaction picture, in order for $\widehat{\lind}^{DB*}$ to reduce to the Davies generator, we need the Lamb-shift term to reduce to zero. Then, using the upper-bound $\left|\tanh\norbra{\frac{\beta\omega_-}{4}}\right|\leq1$, we have:
	\begin{align}
		{\tilde S}_{kl}^{\omega,\tilde{\omega}} = 
		\begin{cases}
			0& \qquad\qquad{\rm if}\; \omega=\tilde{\omega}\\
			\bigo{e^{-(T(\alpha)\,\omega_-^{min})^2/4}\,\Gamma} & \qquad\qquad{\rm otherwise}
		\end{cases}
	\end{align}
	This proves that for $T(\alpha)\gg \omega_-^{min}$ we have the approximate equality $\widehat{\lind}^{DB} \simeq  \widehat{\lind}^{CG} \simeq\widehat{\lind}^{DAV}$, which become exact in the limit $T(\alpha)\rightarrow\infty$.
	
}

\subsection{Proof of Thm.~\ref{thm:lindbladErrorDB}: $\widehat{\lind}^{CG*} $ and $\widehat{\lind}^{DB*} $ are close}\label{app:thmlindbladErrorDB}

First, it should be noticed that thanks to the triangle inequality it holds that :
\begin{align}
	\|\widehat{\lind}^{CG*}[\rho]-\widehat{\lind}^{DB*}[\rho]\|_1\leq 	\|\widehat{\mathcal{D}}^{CG*}[\rho]-\widehat{\mathcal{D}}^{DB*}[\rho]\|_1+ \|[H_{LS}^{DB*},\rho]\|_1\,,\label{eq:163}
\end{align}
so that we can consider the two cases separately. 

We begin by analyzing the dissipative part. By starting from   Eq.~\eqref{eq:averagedLindblad}, and changing variables as in Eq.~\eqref{eq:substitutionQXT} ($\{t_1 = q-x/2\,; \, t_2= q+x/2\}$), we can rewrite the time averaged dissipator as:
\begin{align}
	\widehat{\mathcal{D}}^{CG*}[\rho]&=\int_{-\infty}^{\infty}\frac{\de t_1}{\sqrt{\pi}\, T(\alpha)}\int_{-\infty}^{\infty}\de t_2\int_{-\infty}^{\infty}\de t_3\;\norbra{e^{-\frac{t_1^2+t_2^2}{2T(\alpha)^2}}g_{k\lambda}\norbra{t_3+\frac{t_2-t_1}{2}}g_{\lambda l}\norbra{\frac{t_2-t_1}{2}-t_3}}\mathcal{D}^{t_1,t_2}_{kl}[\rho]\,,\label{eq:dissCGTime}
\end{align}
where we used the definition of $g_{k\lambda}(x)$ and we introduced the notation:
\begin{align}
	\mathcal{D}^{t_1,t_2}_{kl}[\rho] := \Big(A_l(t_1)\rho A_k(t_2)-\frac{1}{2}\{A_k(t_2)A_l(t_1),\rho\}\Big)\,.
\end{align} 
\ms
{It is worth reminding the reader that the evolution of the jump operators is defined with respect to the renormalized Hamiltonian.} Moreover,in the following, it will be useful to denote by $\gamma_{kl}^{t_1,t_2}$  the coefficient in the parenthesis of Eq.~\eqref{eq:dissCGTime}. We can now cast $\widehat{\mathcal{D}}^{DB*}$ in the same form.
First, the rate $\tilde{\gamma}_{kl}^{\omega,\tilde{\omega}}$ defined in Eq.~\eqref{eq:DBcoeffCP} can be transformed as follows:
\begin{align}
	\tilde{\gamma}_{kl}^{\omega,\tilde{\omega}}&=e^{-(T(\alpha)\,\omega_-)^2/4}\;{\widehat{g}}_{k\lambda}\norbra{-\tilde\omega } {\widehat{g}}_{\lambda l}\norbra{-{\omega} } =
	\\
	&= \int_{-\infty}^{\infty}\frac{\de\tau_1}{\sqrt{\pi}\, T(\alpha)}\int_{-\infty}^{\infty}\de \tau_2\int_{-\infty}^{\infty}\de \tau_3 \;  e^{-\tau_1^2/T(\alpha)^2}e^{i\omega (\tau_1-\tau_3)}e^{-i\tilde{\omega} (\tau_1+\tau_2)}{{g}}_{k\lambda}\norbra{\tau_2}{{g}}_{\lambda l}\norbra{\tau_3}\;=
	\\
	&= \int_{-\infty}^{\infty}\de t_1\int_{-\infty}^{\infty}\de t_2\;e^{i{\omega}t_1}e^{-i\tilde\omega t_2}\norbra{\int_{-\infty}^{\infty} \frac{\de t_3}{\sqrt{\pi}\,T(\alpha)} \;  e^{-\frac{(t_1+t_2-2t_3)^2}{4T(\alpha)^2}}{{g}}_{k\lambda}\norbra{t_3+\frac{t_2-t_1}{2}}{{g}}_{\lambda l}\norbra{\frac{t_2-t_1}{2}-t_3}}\,,\label{eq:134}
\end{align}
where, in the last step, we performed the change of variables $\{t_1 = \tau_1-\tau_3\,; \, t_2= \tau_1+\tau_2\,; \, t_3=\frac{\tau_2-\tau_3}{2}\}$. Once again, we denote by $\tilde{\gamma}_{kl}^{t_1,t_2}$ the coefficient in the parenthesis.

This gives us all the ingredients to estimate the first trace distance in Eq.~\eqref{eq:163}, which can be bound as:
\begin{align}
	\|\widehat{\mathcal{D}}^{CG*}[\rho]&-\widehat{\mathcal{D}}^{DB*}[\rho]\|_1  \leq\int_{-\infty}^{\infty}\de t_1\int_{-\infty}^{\infty}\de t_2\; \left\| \norbra{{\gamma}_{kl}^{t_1,t_2}-\tilde{\gamma}_{kl}^{t_1,t_2}}\mathcal{D}^{t_1,t_2}_{kl}[\rho]\right\|_1 \leq 
	\\
	&\leq2\int_{-\infty}^{\infty}\de \tau_1\int_{-\infty}^{\infty}\de \tau_2\; |{\gamma}_{kl}^{\tau_1,\tau_2}-\tilde{\gamma}_{kl}^{\tau_1,\tau_2}| \leq
	\\
	&\leq2\int_{-\infty}^{\infty}\de t_1\int_{-\infty}^{\infty}\de t_2\;|{{g}}_{k\lambda}\norbra{t_1}{{g}}_{\lambda l}\norbra{t_2}|\int_{-\infty}^{\infty}  \frac{\de t_3}{\sqrt{\pi}\,T(\alpha)}\; e^{-\frac{t^2_3}{T(\alpha)^2}}\left |  \norbra{e^{-\frac{((t_3+t_1)^2+(t_3-t_2)^2-2t_3^2)}{2T(\alpha)^2}}-1}\right | =\label{eq:171}
	\\
	&=2\int_{-\infty}^{\infty}\de t_1\int_{-\infty}^{\infty}\de t_2\;|{{g}}_{k\lambda}\norbra{t_1}{{g}}_{\lambda l}\norbra{t_2}|\int_{-\infty}^{\infty}  \frac{\de \Omega}{\sqrt{\pi}}\; e^{-\Omega^2}\left |  \norbra{e^{-\frac{1}{2}((\Omega+t_1/T(\alpha))^2+(\Omega-t_2/T(\alpha))^2) + \Omega^2}-1}\right | \,,\label{eq:142}
\end{align}
where in Eq.~\eqref{eq:171} we implicitly changed variables as $\{t_1 = \frac{\tau_1-\tau_2}{2}+\tau_3\,; \, t_2= \frac{\tau_1-\tau_2}{2}-\tau_3\,; \, t_3=\frac{\tau_1+\tau_2}{2}-\tau_3\}$, and in the last step we also substituted $\{ \Omega=\frac{t_3}{T(\alpha)}\}$. We are now interested in understanding the scaling of the integral in $\de \Omega$ as a function of $T(\alpha)^{-1}$ in the limit $T(\alpha)\rightarrow\infty$. In order to do so, let us first change variables as $T(\alpha)^{-1}\rightarrow \varepsilon$, so that our perturbative expansion can be recast as the study of a function of $\varepsilon$ around zero. Then, taking the derivative with respect to $\varepsilon$ and integrating again we obtain:
\begin{align}
	\int_{-\infty}^{\infty}  \frac{\de \Omega}{\sqrt{\pi}}\;& e^{-\Omega^2}\left |  \norbra{e^{\Omega\,\varepsilon{(t_2-t_1)}-\frac{\varepsilon^2}{2}{(t_1^2+t_2^2)}{}}-1}\right |=
	\\
	&=\int_{-\infty}^{\infty}  \frac{\de \Omega}{\sqrt{\pi}}\; e^{-\Omega^2}\left | \int_0^\varepsilon\de \sigma\;e^{\sigma \,\Omega\,(t_2-t_1)-\frac{\sigma^2}{2}(t_1^2+t_2^2)} \norbra{\Omega\,(t_2-t_1)- \sigma\, (t_1^2+t_2^2)} \right |\leq
	\\
	&\leq\int_{-\infty}^{\infty}  \frac{\de \Omega}{\sqrt{\pi}}\int_0^\varepsilon\de \sigma\;e^{-\Omega^2}e^{\sigma \,\Omega\,(t_2-t_1)-\frac{\sigma^2}{2}(t_1^2+t_2^2)}\norbra{\left |\Omega \right |\left |t_1-t_2 \right | + \sigma\, (t_1^2+t_2^2)}=
	\\
	&=\int_0^\varepsilon\de \sigma\;\norbra{\norbra{\sigma\,(t_1^2+t_2^2)+\sigma\,\frac{(t_1-t_2)^2}{2}{\rm erf}\norbra{\frac{|t_1-t_2|\sigma}{2}}}e^{-\frac{1}{4}\sigma^2 (t_1+t_2)^2}+\frac{|t_1-t_2|}{\sqrt{\pi}}e^{-\frac{\sigma^2}{2}(t_1^2+t_2^2)}}\leq
	\\
	&\leq\;  \frac{|t_1-t_2|}{\sqrt{\pi}}\,\varepsilon+(t_1^2+t_2^2+\frac{1}{2}(t_1-t_2)^2)\,\frac{\varepsilon^2}{2} \leq  \frac{|t_1-t_2|}{\sqrt{\pi}}\,\varepsilon+(|t_1|+|t_2|)^2\,\varepsilon^2\,,
\end{align}
where in the penultimate step we used the fact that both the exponential of negative value and the error function are upper bounded by $1$. Plugging this expression back in Eq.~\eqref{eq:142}, we get the bound:
\begin{align}
	\|\widehat{\mathcal{D}}^{CG*}[\rho]&-\widehat{\mathcal{D}}^{DB*}[\rho]\|_1  \leq 2\int_{-\infty}^{\infty}\de t_1\int_{-\infty}^{\infty}\de t_2\;\norbra{\frac{|t_1-t_2|}{\sqrt{\pi}\,T(\alpha)}+\frac{(|t_1|+|t_2|)^2}{T(\alpha)^2}}|{{g}}_{k\lambda}\norbra{t_1}{{g}}_{\lambda l}\norbra{t_2}|\leq
	\\
	&\leq\frac{2}{\sqrt{\pi}\,T(\alpha)}\int_{-\infty}^{\infty}\de t_1\int_{-\infty}^{\infty}\de t_2\;\norbra{|t_1|+|t_2|}|{{g}}_{k\lambda}\norbra{t_1}{{g}}_{\lambda l}\norbra{t_2}| + \frac{2\,K}{T(\alpha)^2}\leq \frac{4\Gamma\tau}{\sqrt{\pi}\,T(\alpha)} + \frac{2\,K}{T(\alpha)^2}\,,
\end{align}
\ms
{where we used the constant $K$ introduced in Eq.~\eqref{eq:KDef}.}

Let us now focus on the second trace distance in Eq.~\eqref{eq:163}. A simple application of Hölder's inequality gives:
\begin{align}
	\|[H_{LS}^{DB*},\rho]\|_1\leq 2\, \|H_{LS}^{DB*}\|_\infty\|\rho\|_1\,,\label{eq:181}
\end{align}
so it is sufficient to bound the operator norm of the Lamb-Shift Hamiltonian. To this end, let us transform $\tilde{S}_{kl}^{\omega,\tilde{\omega}}$ into the time representation:
\begin{align}
	\tilde{S}_{kl}^{\omega,\tilde{\omega}}&=i\tanh\norbra{{\frac{\beta \,\omega_-}{4}}}e^{-(T(\alpha)\,\omega_-)^2/4}\;{\widehat{g}}_{k\lambda}\norbra{-\tilde\omega } {\widehat{g}}_{\lambda l}\norbra{-{\omega} } =
	\\
	&=\int_{-\infty}^{\infty}\frac{\de\tau_1}{\sqrt{\pi}\, T(\alpha)} \int_{-\infty}^{\infty}\de \tau_2\int_{-\infty}^{\infty}\de\tau_3 \;  F(\tau_1)e^{i\omega (\tau_1-\tau_3)}e^{-i\tilde{\omega} (\tau_1+\tau_2)}{{g}}_{k\lambda}\norbra{\tau_2}{{g}}_{\lambda l}\norbra{\tau_3}\;=
	\\
	&= \int_{-\infty}^{\infty}\de t_1\int_{-\infty}^{\infty}\de t_2\;e^{i{\omega}t_1}e^{-i\tilde\omega t_2}\norbra{\int_{-\infty}^{\infty} \frac{\de t_3}{\sqrt{\pi}\,T(\alpha)} \;  F(t_1+t_2-2t_3){{g}}_{k\lambda}\norbra{t_3+\frac{t_2-t_1}{2}}{{g}}_{\lambda l}\norbra{\frac{t_2-t_1}{2}-t_3}}\,,
\end{align}
where we introduced the function:
\begin{align}\label{eq:LSft}
	F(t) := \frac{2}{\beta}\int_{-\infty}^{\infty}\de x\; {\rm sech}\norbra{\frac{2\pi x}{\beta}}\sin\norbra{\frac{\beta (t-x)}{2 T(\alpha)^2}} e^{-\frac{(t-x)^2}{T(\alpha)^2}}e^{\frac{\beta^2}{16T(\alpha)^2}}\,,
\end{align}
and in the last step we changed variables as $\{t_1 = \tau_1-\tau_3\,; \, t_2= \tau_1+\tau_2\,; \, t_3=\frac{\tau_2-\tau_3}{2}\}$. Then, we have the bound:
\begin{align}
	\|H_{LS}^{DB*}\|_\infty &\leq \int_{-\infty}^{\infty}\de t_1\int_{-\infty}^{\infty}\de t_2\; \left\| \tilde{S}_{kl}^{t_1,t_2}\, A_k^\dagger(t_2)A_l(t_1)\right\|_\infty\leq
	\\
	&\leq\int_{-\infty}^{\infty}\de t_1\int_{-\infty}^{\infty}\de t_2\;|\tilde{S}_{kl}^{t_1,t_2}|\,,
\end{align}
where we used $\|A_k\|_\infty=1$. We can bound the last integral as follows:
\begin{align}
	\int_{-\infty}^{\infty}\de t_1\int_{-\infty}^{\infty}\de t_2\; |\tilde{S}_{kl}^{t_1,t_2}| &= \int_{-\infty}^{\infty}\de \tau_2\int_{-\infty}^{\infty}\de\tau_3  |{{g}}_{k\lambda}\norbra{\tau_2}||{{g}}_{\lambda l}\norbra{\tau_3}| \int_{-\infty}^{\infty}\frac{\de\tau_1}{\sqrt{\pi}\, T(\alpha)}\;|F(\tau_1)|\,,\label{eq:194}
\end{align}
where we implicitly changed variables back to $\{t_1 = \frac{\tau_1-\tau_2}{2}+\tau_3\,; \, t_2= \frac{\tau_1-\tau_2}{2}-\tau_3\,; \, t_3=\frac{\tau_1+\tau_2}{2}-\tau_3\}$. Let us focus on the last integral:
\begin{align}
	\int_{-\infty}^{\infty}\frac{\de\tau_1}{\sqrt{\pi}\, T(\alpha)}\;|F(\tau_1) |&\leq  \frac{2\,e^{\frac{\beta^2}{16T(\alpha)^2}}}{\beta}\int_{-\infty}^{\infty}\frac{\de\tau_1}{\sqrt{\pi}\, T(\alpha)}\int_{-\infty}^{\infty}\de x\; {\rm sech}\norbra{\frac{2\pi x}{\beta}}\left |\sin\norbra{\frac{\beta (\tau_1-x)}{2 T(\alpha)^2}}\right | e^{-\frac{(\tau_1-x)^2}{T(\alpha)^2}}=
	\\
	&=\frac{2\,e^{\frac{\beta^2}{16T(\alpha)^2}}}{\beta}\int_{-\infty}^{\infty}\frac{\de\tau}{\sqrt{\pi}\, T(\alpha)}\int_{-\infty}^{\infty}\de x\; {\rm sech}\norbra{\frac{2\pi x}{\beta}}\left |\sin\norbra{\frac{\beta \,\tau}{2 T(\alpha)^2}}\right | e^{-\frac{\tau^2}{T(\alpha)^2}} \leq
	\\
	&\leq e^{\frac{\beta^2}{16T(\alpha)^2}}\int_{-\infty}^{\infty}\frac{\de\tau}{\sqrt{\pi}\, T(\alpha)}\left |\frac{\beta \,\tau}{2 T(\alpha)^2}\right | e^{-\frac{\tau^2}{T(\alpha)^2}} = \frac{\beta e^{\frac{\beta^2}{16T(\alpha)^2}}}{2 \sqrt{\pi}\,T(\alpha)}\,,
\end{align}
where we changed variables to decouple the hyperbolic secant from the rest of the terms, and then used a linear upper bound for the sine function (that is, $|\sin(x)|\leq |x|$). Putting back this estimate in Eq.~\eqref{eq:194} we finally obtain:
\begin{align}
	\int_{-\infty}^{\infty}\de t_1\int_{-\infty}^{\infty}\de t_2\; |\tilde{S}_{kl}^{t_1,t_2}|\leq  \frac{\Gamma\beta e^{\frac{\beta^2}{16T(\alpha)^2}}}{2 \sqrt{\pi}\,T(\alpha)}\,.
\end{align}
Wrapping everything up we obtain the claim.

\end{document}